\documentstyle[12pt,twoside,Mtex]{article}
\pdfoutput=1

\def \nats{{\Bbb N}}
\def \posnats{{\nats\setminus \lbrace 0\rbrace}}
\def \ints{{\Bbb Z}}
\def \cmplx{{\Bbb C}}
\def \calU{{\cal U}}
\def \calF{{\cal F}}
\def \calD{{\cal D}}

\def \calF{{\cal F}}
\def \varG{{\cal G}}
\def \varH{{\cal H}}
\def \calT{{\cal T}}
\def \calS{{\cal S}}
\def \cheb{{\vartheta}}
\def \abar{{\bar a}}
\def \bbar{{\bar b}}
\def \cbar{{\bar c}}
\def \hbar{{\bar h}}
\def \gbar{{\bar g}}
\def \fbar{{\bar f}}
\def \ubar{{\bar u}}
\def \Fbar{{\bar F}}
\def \Dbar{{\bar D}}
\def \mubar{{\bar \mu}}
\def \wpbar{{\bar \wp}}
\def \ftilde{{\tilde f}}
\def \gtilde{{\tilde g}}
\def \e{\varepsilon}
\def \fhat{{\hat f}}
\def \ghat{{\hat g}}
\def \hhat{{\hat h}}
\def \expn{\mathop{\rm expn}}
\def \coeff{\mathop{\rm coeff}}
\def \gal{\mathop{\rm Gal}}

\def \Psihat{{\hat \Psi}}
\def \dec{{DEC^F_\wp}}
\def \fld{{FIELDS^F_\wp}}
\def \grp{{GROUPS^F_\wp}}
\def \sep{{SEP^F_\wp}}

\def \apdec{{APDEC^F_\wp}}
\def \sapdec{{SAPDEC^F_\wp}}
\def \flag{{FLAGS^F_\wp}}
\def \cdec{{cDEC^F_\wp}}
\def \cfld{{cFIELDS^F_\wp}}
\def \cgrp{{cGROUPS^F_\wp}}
\def \csep{{cSEP^F_\wp}}
\def \cblk{{cBLOCKS^F_\wp}}
\def \capdec{{cAPDEC^F_\wp}}
\def \csapdec{{cSAPDEC^F_\wp}}
\def \cflag{{cFLAGS^F_\wp}}
\def \cdecall{{cDEC^F_*}}
\def \cfldall{{cFIELDS^F_*}}
\def \cgrpall{{cGROUPS^F_*}}

\def \capdecall{{cAPDEC^F_*}}

\def \monicpoly{{{\Bbb P}_F}}
\def \bbP{{\Bbb P}}

\def \bbG{{\Bbb G}}
\def \bbS{{\Bbb S}}
\def \bbB{{\Bbb B}}
\def \bbW{{\Bbb W}}
\def \bbA{{\Bbb A}}
\def \bbM{{\Bbb M}}
\def \xvec{{\vec x}}
\def \yvec{{\vec y}}

\def \ap{{\Bbb A_F}}
\def \bbV{{\Bbb V}}
\def \bbL{{\Bbb L}}
\def \bbT{{\Bbb T}}
\def \bbU{{\Bbb U}}
\def \uf{{\Bbb U}_F}
\def \ordfact{(r_m,r_{m-1},\ldots,r_1)}
\def \sf{{{\bf S}_F}}
\def \sfl{{{\bf S}^{(\ell)}_F}}
\def \apdiv{\mathrel{{{\rlap{$/$}\circ}}}}
\def \apdivides{\mathrel{{{\rlap{$\mskip2mu|$}\circ}}}}
\def \apequiv{\mathrel{{{\rlap{$\equiv$}\mskip2mu\circ\mskip\thinmuskip}}}}

\def \trans{\triangleright}

\def \wq{\kappa}
\def \pwr{{\bf P}}

\def \mt{\sqcap}
\def \jn{\sqcup}

\def \G[#1,#2]{\Gamma^{\scriptscriptstyle #1}_{\scriptscriptstyle #2}}

\def \nn{{n_{\scriptscriptstyle N}}}
\def \nd{{n_{\scriptscriptstyle D}}}
\def \rn{{r_{\scriptscriptstyle N}}}
\def \rd{{r_{\scriptscriptstyle D}}}
\def \sn{{s_{\scriptscriptstyle N}}}
\def \sd{{s_{\scriptscriptstyle D}}}
\def \fn{{f_{\scriptscriptstyle N}}}
\def \fd{{f_{\scriptscriptstyle D}}}
\def \gn{{g_{\scriptscriptstyle N}}}
\def \gd{{g_{\scriptscriptstyle D}}}
\def \hn{{h_{\scriptscriptstyle N}}}
\def \hd{{h_{\scriptscriptstyle D}}}
\def \an{{a_{\scriptscriptstyle N}}}

\def \nnbar{{\bar n_{\scriptscriptstyle N}}}
\def \ndbar{{\bar n_{\scriptscriptstyle D}}}
\def \rnbar{{\bar r_{\scriptscriptstyle N}}}
\def \rdbar{{\bar r_{\scriptscriptstyle D}}}
\def \snbar{{\bar s_{\scriptscriptstyle N}}}
\def \sdbar{{\bar s_{\scriptscriptstyle D}}}
\def \fnbar{{\bar f_{\scriptscriptstyle N}}}
\def \fdbar{{\bar f_{\scriptscriptstyle D}}}

\def \ulms[#1,#2]{{#1^{\scriptscriptstyle [#2]}}}
\def \smC{{\scriptscriptstyle C}}

\title{Functional Decomposition of Polynomials}
\author{Mark Giesbrecht}
\begin{document}

\vfill
\eject
\quad
\vskip 0.9in 
\centerline{\Large \bf Some Results on the}
\vskip 6pt
\centerline{\Large\bf Functional Decomposition of Polynomials}
\vskip 0.5in
\centerline{by}
\centerline{Mark William Giesbrecht}
\centerline{Department of Computer Science}
\vskip 1.5in
\centerline{A Thesis submitted in conformity with the requirements}
\centerline{for the Degree of Master's of Science in the}
\centerline{University of Toronto}
\vskip 0.7in
\centerline{Department of Computer Science}
\centerline{University of Toronto}
\centerline{Toronto, Ontario, Canada, M5S 1A4}
\vskip 0.5in
\centerline{Copyright \copyright ~1988 ~Mark Giesbrecht}

\newpage
\section*{Abstract}

If $g$ and $h$ are functions over some field, we can consider their
composition $f=g(h)$.  The inverse problem is decomposition: given $f$,
determine the existence of such functions $g$ and $h$.  In this thesis
we consider functional decompositions of univariate and multivariate
polynomials, and rational functions over a field $F$ of characteristic
$p$.  In the polynomial case, ``wild'' behaviour occurs in both the
mathematical and computational theory of the problem if $p$ divides
the degree of $g$.  We consider the wild case in some depth, and
deal with those polynomials whose decompositions are in some sense
the ``wildest'': the additive polynomials.  We determine the
maximum number of decompositions and show some polynomial time
algorithms for certain classes of polynomials with wild decompositions.
For the rational function case we present a definition of the
problem, a normalised version of the problem to which the general problem
reduces, and an exponential time solution to the normal problem.

\newpage
\section*{Acknowledgement.}

I would like to thank my supervisor Dr. Joachim von zur Gathen
for the long and fruitful hours he spent helping me with this
thesis, and Dr. Rackoff for being my second reader.
I would also like to thank my office mates and many
others for their helpful suggestions and proof reading.
Finally, I would like to thank NSERC for its scholarship support.

\newpage
\section*{Table of Contents}

{\obeylines
Introduction \dotfill 7

Chapter 1. Polynomial Decomposition \dotfill 13
\qquad A. Definition of the Problem \dotfill 13
\qquad B. Decomposition and the Subfields of $F(x)$ \dotfill 17
\qquad C. Separated Polynomials \dotfill 20
\qquad D. Multidimensional Block Decompositions \dotfill 22
\qquad E. Chebyshev Polynomials \dotfill 27
\qquad F. Complete Rational Decompositions \dotfill 30
\qquad G. The Number of Indecomposable Polynomials \dotfill 32
\qquad H. Multivariate Decomposition  \dotfill 34

Chapter 2. Decomposition Algorithms \dotfill 36
\qquad A. The Model of Computation  \dotfill 37
\qquad B. Computing Right Division  \dotfill 38
\qquad C. Univariate Decomposition Using Separated Polynomials \dotfill 39
\qquad D. Univariate Decomposition in the Tame Case \dotfill 40
\qquad E. Decomposition Using Block Decomposition \dotfill 41
\qquad F. A Lower Bound on the Degree of Splitting Fields \dotfill 43
\qquad G. Decompositions Corresponding To Ordered Factorisations \dotfill 46
\qquad H. Computing Complete Univariate Decompositions  \dotfill 48
\qquad I. Decomposition Multivariate Polynomial in the Tame Case \dotfill 48
\qquad J. Multivariate Decomposition using Separated Polynomials \dotfill 51

Chapter 3. Additive Polynomials \dotfill 53
\qquad A. Definition and Root Structure of Additive Polynomials \dotfill 53
\qquad B. Rationality and the Kernel \dotfill 57
\qquad C. Rational Decompositions of Additive Polynomials \dotfill 58
\qquad D. The Number of Bidecompositions of a Polynomial  \dotfill 59
\qquad E. Complete Decompositions of Additive Polynomials \dotfill 61
\qquad F. The Number of Complete Rational Normal Decompositions \dotfill 61

Chapter 4. The Ring of Additive Polynomials \dotfill 67
\qquad A. Basic Ring Structure  \dotfill 67
\qquad B. The Euclidean Scheme  \dotfill 69
\qquad C. The Structure of the Set of Decompositions \dotfill 74
\qquad D. Completely Reducible Additive Polynomials  \dotfill 83
\qquad E. The Uniqueness of Transmutation \dotfill 86
\qquad F. The Number of the Complete Decompositions \dotfill 89

Chapter 5. Decomposing Additive Polynomials \dotfill 90
\qquad A. The Model of Computation  \dotfill 90
\qquad B. The Cost of Basic Operations in $\ap$ \dotfill 91
\qquad C. The Minimal Additive Multiple \dotfill 92
\qquad D. Complete Rational Decomposition of Additive Polynomials \dotfill 94
\qquad E. General Rational Decomposition of Additive Polynomials \dotfill 97
\qquad F. General Decomposition of Completely 
~~~~~~~~~~Reducible Additive Polynomials \dotfill 100
\qquad G. Determining Transmutations of Additive Polynomials \dotfill 105
\qquad H. Bidecomposition of Similarity Free Additive Polynomials \dotfill 107
\qquad I. Absolute Decompositions of Additive Polynomials \dotfill 113

Chapter 6.  Rational Function Decomposition \dotfill 115
\qquad A. The Normalised Decomposition Problem \dotfill 115
\qquad B. Decomposing Normalised Rational Functions \dotfill 122

Conclusion. \dotfill 127

References.  \dotfill 129
}
\vfill
\eject

\section*{Introduction.}

A fundamental idea in computer science and mathematics is
that of composition.  
One way to understand an object, whether it is
abstract or concrete, is to understand how it combines with other objects
of the same type.  
A converse problem to this also exists: Given
an object, describe how it is made up as the composition of other objects.
This is decomposition.  We can introduce constraints on the decompositions
we wish to examine. 
What happens when we cannot further break down the object under consideration
given these constraints?
We say it is indecomposable.  A very natural question 
to look at is how an object
under consideration breaks down into indecomposable pieces.  And if we
relax the constraints somewhat, do these indecomposable
objects decompose once again?  
In mathematics, and especially algebra,
decomposition is a central concept.  Decomposing matrices, algebras and
groups are all well explored areas.  The factoring of polynomials is
a fundamental example of the decomposition in the ring of polynomials
under the usual operations of addition and multiplication.  
The computational aspects of factoring polynomials have been an extremely
active area of research over the last two decades.
But polynomials
can also be composed functionally, and form a ring under
addition and composition.  What does factorisation in this ring look
like?  Although this question has been addressed mathematically for
at least six decades, many unresolved questions still remain.  Computationally
the area is extremely new, having developed only over the last decade
or so.  Applications of polynomial decomposition
within the areas of coding theory and
cryptography exist (though will not be dealt with here), and the
problem is  of computational interest in its own right.
Though some progress has been made in the (mathematically)
well understood cases, the problem in general appears to be difficult.
We will address ourselves to some of these difficulties.

If $f_m,f_{m-1},\ldots,f_1$ are univariate
polynomials over a field $F$
of degrees \break $r_m,r_{m-1},\ldots,r_1\in\nats$ respectively, their
functional composition 
\[
f=f_m(f_{m-1}(\cdots(f_2(f_1))\cdots))\in F[x]
\] 
has degree $n=r_m r_{m-1} \cdots r_1$, and can be computed in a straightforward
manner.
In this thesis we examine a
converse problem. Namely, given $f$ and $r_m,\ldots,r_1$, determine
if there exist polynomials 
$f_m,\ldots,f_1\in F[x]$ such that $\deg f_i=r_i$
for $1\leq i\leq m$ and $f=f_m(f_{m-1}(\cdots(f_2(f_1))\cdots))$, and
if so, compute them.  We call this the polynomial decomposition problem.
When the problem is limited to decompositions into two composition factors
of given degree, we call this the bidecomposition problem.
A polynomial is considered to be indecomposable if there
is no way to decompose it into non-trivial (degree at least two)
composition factors.  We also consider decompositions into
indecomposable composition factors, which we call complete decompositions.
Further questions arise when we consider decompositions over arbitrary 
algebraic extension fields of the ground field, or absolute decompositions.
All these issues concerning decompositions have been dealt with mathematically
since the seminal paper of Ritt[1922], which showed a very strong,
``well behaved'' structure for decompositions of polynomials over the complex
numbers.  
Since then, the mathematical literature dealing with the
problem has been extensive, though far from complete.  
The difficulty in the decomposition problem seems to be connected
to the divisibility of the degrees 
by the characteristic $p$ of the ground field.
The ``tame''
case, where $p=0$ or $p\nmid r_i$ for $1<i\leq m$ is  well understood.
However, the ``wild'' case where $p|r_i$ for some $i>1$ is still
largely a mystery.  It is this case in which we will be most interested.

For some special classes of polynomials,
decompositions in the wild case are well understood. One such class
is the ``additive'' polynomials.  
These are the polynomials where only exponents which are powers of the
characteristic $p$ of the field have non-zero coefficients.  In 
some sense they are the ``wildest'' polynomials (see von zur Gathen[1988]).
The  theory of additive polynomials was introduced in Ore[1933b], and
will be presented here in some detail.  
Kozen and Landau[1986] give an (exponential time) reduction
of the general decomposition problem to univariate factorisation, and
give a formulation of this problem in terms of
the action of Galois groups.  This turns out to be somewhat
simpler for decompositions of separable, irreducible
polynomials (over arbitrary fields) than in the general case.  And for
irreducible polynomials over finite fields they give a complete 
description of the decomposition structure.  

Decompositions of multivariate polynomials have also been considered.
Evyatar and Scott[1972] show a structure very similar to the univariate
case.  We consider decompositions of a multivariate polynomial $f$ into
a univariate polynomial $g$ and a multivariate polynomial $h$. Completely
analogous tame and wild cases exist, although even less is known
about the wild case here than for univariate polynomials.

Computationally, polynomial decomposition has only been examined since
1976 by Barton and Zippel[1976,1985].   They give a general algorithm (for both
the tame and wild cases) which requires
a factoring subroutine and an exponential number of field
operations in the degree of the input polynomial.
Over arbitrary (``computable'') fields, the decomposition problem
is undecidable (see von zur Gathen[1988]).
Kozen and Landau[1986] exhibit an algorithm for the bidecomposition
problem in the tame case which
requires only a polynomial number of field operations in the input degree.
For multivariate polynomials there is a similar situation.
Fast algorithms which compute decompositions do exist in the
tame case (see Dickerson[1987] and von zur Gathen[1987]).  
And in the wild case 
we present an algorithm to perform multivariate decomposition
(in an exponential number of field operations).
Some special classes of the wild case have also been dealt with:  Kozen
and Landau[1986] give a decomposition algorithm for 
irreducible, separable  polynomials which requires a quasi-polynomial
number of field operations in the degree of the input,  and for irreducible
polynomials over finite fields, their algorithm requires
only a polynomial number of field operations in the input degree.  

This thesis is organised into six chapters.  In chapter one we
present a mathematical definition of the univariate decomposition problem and
five different formulations of it.   
Each of these formulations has been used in the mathematical or 
computational literature, in various forms.  Some were developed for
special cases, and some fall immediately from the problem definition.
We generalise these formulations and put them in a consistent
language and context, showing their basic equivalence.  
We also define the multivariate problem in a similar manner, showing
two basic, equivalent formulations.

In chapter two,
we present the computational approaches to polynomial decomposition
which have been developed for both the wild and tame cases. These algorithms
will be stated in terms of the formulation of the decomposition 
problem used, as
developed in chapter one.
We show that for certain ``nice'' families of polynomials
(polynomials for which an efficient algorithm for decomposition into
two composition factors of given degree exists, and 
for which such decompositions
are unique) the problem of decomposing a polynomial into an arbitrary number
of factors of given degree is reducible to the bidecomposition problem.
Using a structure theorem of Evyatar and Scott[1972], we also exhibit
an algorithm for decomposing multivariate polynomials (in both the tame
and wild cases) over any field supporting a factoring algorithm.

In chapters three through five we introduce the additive polynomials, 
a class of
polynomials with wild decompositions which are well understood.  
In chapter three we 
develop the theory of these polynomials 
with respect to the structure of their roots
in their splitting fields.  From this we garner  quasi-polynomial
lower bounds on the number of decompositions (both decompositions
into two factors and complete decompositions) of ``simple''
additive polynomials.  
This shows that any algorithm which produces all decompositions
of an arbitrary polynomial in the wild case
cannot be expected to work in a polynomial number of field operations.
In fact, we determine exactly the maximum number of decompositions of
simple additive polynomials of a given degree.  

In chapter four the
theory of Ore[1933a], which describes non-commutative Euclidean rings,
is developed for the additive polynomials.  We extend this theory
by further developing the relationship between different complete 
decompositions of a given polynomial.  We also
show a number of results concerning the uniqueness of decompositions.
Combining this formal structure with the algebraic structure from
chapter three, we show a quasi-polynomial upper bound on the number of
possible complete decompositions of additive polynomials in general.

In chapter five we make use of the two previous chapters to develop
algorithms for the decomposition of additive polynomials.  We show that
we can determine indecomposability in a polynomial number of field
operations, and in fact can generate one complete decomposition.  However,
the only way method we know to find a decomposition into an
arbitrary number of factors of given
degrees is by finding all complete decompositions.  Using the upper
bound from chapter four, we get an algorithm requiring a quasi-polynomial
number of field operations.  Two large subclasses of the additive polynomials
show more favourable results:
the completely reducible additive polynomials and
the similarity free additive polynomials.  Decomposition algorithms requiring
a polynomial number of field operations are shown in each case.
We also show a quasi-polynomial time algorithm for the absolute decomposition
of additive polynomials.  This algorithm may well run in a polynomial number
of field operations, subject to a conjectured (but unproven)
upper bound on the degrees of splitting
fields of additive polynomials.  This conjecture follows immediately
from a much stronger (and also unproven) conjecture of Ore[1933b].

In chapter six we define the rational function decomposition
problem.  We show a normalisation of this problem to a more
a more uniquely defined form.  We then show that the rational function
decomposition problem is reducible to this normal rational function
decomposition problem.  Finally, we present an algorithm for
solving the normal problem 
(in an exponential number of field operations in the input degree).

In summary the main original results of this thesis are:
\begin{list}{}{\genlist}
\item[(1)] five equivalent formulations of the univariate decomposition
problem and two formulations of the multivariate decomposition
problem,
\item[(2)] a reduction from the general problem of finding 
decompositions into an arbitrary number of factors of given degree
to the bidecomposition problem, for certain ``nice'' families of polynomials,
\item[(3)] an exponential time algorithm for decomposing 
multivariate polynomials (in both the tame
and wild cases)
over any field supporting a factoring algorithm,
\item[(4)] a precise determination of the maximum number of decompositions
of an additive polynomial (which is super-polynomial in the degree), 
giving a super-polynomial lower bound on 
the number of decompositions of a given polynomial in the wild case,
\item[(5)] a polynomial time algorithm for the complete decomposition
of additive polynomials, and hence an algorithm for determining
indecomposability,
\item[(6)] a quasi-polynomial time algorithm for the decomposition of
an additive polynomial into factors of given degrees,
\item[(7)] polynomial time algorithms for the decomposition of two
special classes of additive polynomials, the completely reducible
additive polynomials and the similarity free additive polynomials,
\item[(8)] a quasi-polynomial time algorithm for the absolute
decomposition of additive polynomials, which could well run in
polynomial time, subject to an unproven conjecture of Ore[1933b],
\item[(9)] a definition of the rational function decomposition 
problem, as well as a normalised form of this problem, and a
reduction from the general problem to the normal problem,
\item[(10)] a computational solution to the normal rational function
decomposition problem, requiring an exponential number of field operations.
\end{list}

\noindent
Results 5 through 8 assume the existence of a polynomial time algorithm
for factoring univariate polynomials.

\newpage
\section{Polynomial Decomposition}

\subsection{Definition of the Problem}

Let $F$  be an arbitrary field and $K$ an extension field of $F$.
A {\it decomposition} of a polynomial $f \in F[x]$ is an ordered
sequence of polynomials $f_i \in K[x]$ for $1 \leq i \leq m$ such
that
\[
\eqalign{f   & = f_m(f_{m-1}(\cdots (f_2(f_1))\cdots) \cr
             & = f_m \circ f_{m-1} \circ \cdots \circ f_2 \circ f_1. \cr}
\]
If $K=F$ then the decomposition is said to be {\it rational}.
The polynomial $f$ is considered to be {\it (rationally)
indecomposable} if for any
(rational) decomposition, all but 
one of the composition factors has degree one.
If this is even true when $K$ is allowed to be an algebraic closure of $F$,
then $f$ is {\it absolutely indecomposable}.  

Assume $f=g\circ h$ where $f \in F[x]$ and 
$g,h\in K[x]$. 
Assume also that $a$ and $c$ are the leading (high order) coefficients of
$f$ and $h$ respectively.  Then
\[
{f\over a} = {({1\over a}g(cx+h(0))) }\circ {{h-h(0)}\over c} \qquad
\]
is a decomposition of a monic polynomial into two monic polynomials, the
second of which has constant coefficient zero.  Thus, without loss
of generality, we can assume for any decomposition $f=g\circ h$ that
$f,g$, and $h$ are monic and $h(0)=0$.
Similarly, if $f=f_m \circ f_{m-1} \circ \cdots \circ f_1$, we can
assume that $f\in F[x]$ and $f_i\in K[x]$ for $ 1\leq i \leq m$ are monic and 
$f_i(0)=0$ for $ 1\leq i<m$.  Call any decomposition of this form a
{\it normal} decomposition.

Define the {\it rational normal decomposition problem} as follows.
For any $n,m\in\nats$, an  {\it ordered factorisation}  of $n$ of length
$m$ is an $m$-tuple
\[
\wp=\ordfact
\]
where $r_i\in\nats$ and $r_i\geq 2$ for $1\leq i\leq m$ and
\[
\prod_{1\leq i\leq m}r_i=n.
\]
Let $m\in\posnats$ and let 
$F$ be any field of characteristic $p$, where $p$ is a prime number.
Let $\monicpoly = \lbrace f\in F[x]: f  ~monic \rbrace$. Define
\[
\dec =
\left\lbrace
(f,(f_m,f_{m-1},\ldots,f_1)) \in \monicpoly \times (\monicpoly)^m \Biggm |
\eqalign{ 
& f=f_m\circ f_{m-1}\circ\cdots\circ f_1 \cr
& \hbox{where}~\deg f_i=r_i~\hbox{and} \cr
& f_i(0)=0 ~\hbox{for}~ 1\leq i<m\cr
}
\right\rbrace.
\]

\noindent
The computational problem is, given $f\in \monicpoly$ and  $\wp$ as above,
to decide whether there exist any
\[(f_m,f_{m-1},\ldots,f_1) \in (\monicpoly)^m\]
such that 
\[(f,(f_m,f_{m-1},\ldots,f_1)) \in DEC^F_\wp, \]
and in the affirmative case, to compute one or all of them.

The {\it rational normal bidecomposition problem} 
is a restriction of the above problem to ordered factorisations 
$\wp=(r_2,r_1)$ 
of length two.
Mathematically, this problem addresses many of the same questions 
as the general problem since 
we can always look
at decompositions into two parts, and then continue recursively
on the composition factors obtained.  This problem has
been examined extensively in the literature, but many unresolved
questions remain.
Two basic cases emerge in the mathematical behaviour of the 
bidecomposition problem.
The ``tame'' case, when $p \nmid r_2$, is as its name might suggest,
well behaved.  Kozen and Landau[1986] observed
that there exists at most one
decomposition for any given $f$ and $\wp$ (this also follows from
Fried and MacRae[1969a]).
Furthermore, they showed it can
be determined in polynomial time.  
As well, any normal decomposition of $f$ will be rational in this case.
This was shown for the case $F=\cmplx$ by Ritt[1922], for all fields of
characteristic zero by Levi[1942], and for the general ``tame'' case by Fried
and MacRae[1969a].

The ``wild'' case, when $p|r_2$, is much harder to deal with, both
mathematically and computationally. Fields are exhibited
over which the problem is undecidable in von zur Gathen[1988].  
Decompositions are not
necessarily unique as the following example shows (other examples can
be found in Ore[1933b]).
Let $F=GF(5)$.  Then
\[
\eqalign{
f=x^{5^3}+x^{5^2}+x^5+x 
&= (x^{5^2}+3x^5+2x)\circ(x^5+3x) \cr
&= (x^{5^2}+4x^5+3x)\circ(x^5+2x) \cr
&= (x^{5^2}+x)\circ(x^5+x). \cr
}
\]
Here $f$ has 3 distinct decompositions in $DEC^F_{(5^2,5)}$.
Also, in the ``wild'' case decompositions may not be rational.  With $F$ as
above consider the polynomial 
\[
\eqalign{
f = & x^{5^2}+x^5+x \cr
  = & (x^5 +\alpha x) \circ (x^5+\beta x) \cr
  = & x^{5^2}+(\beta^5+\alpha)x^5+\alpha\beta x. \cr
}
\]
It follows that $\alpha+\beta^5=\alpha\beta=1$, and the polynomial $f$
has a decomposition of this form if and only if $\beta$ is a root of
$\varphi=x^6-x+1\in F[x]$.  But $\varphi$ has no roots in $F$, and hence
$f$ has monic normal decompositions only in algebraic extensions of $F$.
It will be seen that even in small finite fields the number of
bidecompositions of a given polynomial of degree $n$ can be 
super-polynomial in $n$. 
Polynomial time decomposition algorithms for rational decompositions and
for decompositions in algebraic extensions are known to exist only for certain
classes of polynomials.

The ring of polynomials $F[x]$ under addition and composition is obviously
without zero divisors. It is not a (left or right) Euclidean ring
however, as right or left division with remainder 
of $f\in F[x]$ by $g\in F[x]$
makes sense only when the degree of $g$ divides the degree of $f$.

Let $F=GF(4)$, and let $\omega\in F$ be
a primitive cube root of unity.
Consider the polynomial
\[
f=(x^4-x)^3\in F[x].
\]
Dorey and Whaples[1974] show
\[
f=(x^4-x^3-x^2+x)\circ(x^3+\omega\alpha x+\alpha\omega^2)
\]
for any $\alpha\in F$. Hence left (compositional) division of $f$ by
$(x^4-x^3-x^2+x)$ is not unique.  A somewhat stronger statement can be made
about (compositional) right division.
Let $F$ be an arbitrary field and $K$ an extension field of $F$.

\medskip
\begin{lemma}
If $f,h\in F[x]$ and $g\in K[x]$ are nonzero 
of degrees $n,r$ and $s$ respectively, 
and $f=g\circ h$, then 
\begin{list}{}{\genlist}
\item[(i)] $g$ is uniquely determined by $f$ and $h$, and 
\item[(ii)] $g\in F[x]$.
\end{list}
\end{lemma}
\begin{proof}
\begin{list}{}{\genlist}
\item[(i)]
Assume $g^\prime\in K[x]$ is such that $f=g^\prime\circ h$.
Then
\[
\eqalign{
0 = & g\circ h - g^\prime\circ h\cr
  = & (g-g^\prime)\circ h,\cr
}
\]
and as $F[x]$ under composition has no zero divisors, this implies
that $g=g^\prime$.
\item[(ii)]
The coefficients of $f$ are $K$-linear combinations of the coefficients of
$h^i$ for $1\leq i\leq r$.  The coefficients of $f$ and $h^i$ are in $F$,
so the coefficients of $g$ are a solution to a system of linear equations over
$F$.  Since such a system has a  solution in $F$ if it has one over $K$, 
and $g$ is unique
by (i), the coefficients of $g$ are in $F$ and $g\in F[x]$. \QED
\end{list}
\end{proof}

Further structure can be derived  about the fields over which
decompositions exist.  Let $F$ be a field and let $\Fbar$ be a fixed
algebraic closure of $F$.

\medskip
\begin{lemma}
Let $f,g,h\in F[x]$ have degrees $n,r$, and $s$ respectively and $f=g\circ h$.
Assume $f$ has splitting field $K\subseteq\Fbar$. Then $g$ splits over $K$ and
for each root $\alpha\in K$ of $g$, $h-\alpha$ splits over $K$.
\end{lemma}
\begin{proof}
Assume 
\[
g=\prod_{1\leq i\leq r}(x-\beta_i)\]
where $\beta_i\in\Fbar$ for
$1\leq i\leq r$.  Then 
\[
f=\prod_{1\leq i\leq r}(h-\beta_i).
\]
Let
$\gamma=\beta_i$ for some $i\in\nats$ 
with $1\leq i\leq r$.  If $\alpha\in\Fbar$
is a root of $h-\gamma$, then $\alpha$ is a root of $f$.  Hence $\alpha\in K$
and $h-\gamma$ splits over $K$.  Since $\gamma$ is the product of the
roots of $h-\gamma$, $\gamma\in K$.  Therefore, $g$ splits over $K$
as well.
\QED
\end{proof}

This theorem implies a number of interesting facts about decompositions
over extensions of the ground field.

\medskip
\begin{corollary}
Let $F$ be an arbitrary field and $L\supseteq F$ an extension field.
Let $f\in F[x]$ be monic of degree $n$ with splitting field $K$.
Then, if $f$ is indecomposable in $K$, $f$ is indecomposable in $L$.
\end{corollary}
\begin{proof}
Suppose $f=g\circ h$ for some $g,h\in L[x]$ whose degrees are at least 
two.  Then, by lemma 1.2, $g$ splits in $K$, so $g\in K[x]$.  
Let $\gamma\in K$ be a root of $g$.  Then $h-\gamma$ splits over $K$
(also by lemma 1.2) and $h\in K[x]$.
But $f$ is indecomposable over $K$ and we get a contradiction.
\QED
\end{proof}

\noindent
Decomposition over an arbitrary field extension (or decomposition in
the splitting field as just shown) is called {\it absolute decomposition}.
We will see in theorem 2.8 
that over many fields $F$ there are polynomials $f\in F[x]$
whose splitting fields are of degree exponential in $n$ over $F$.  
Over infinite fields $F$ von zur Gathen [1987a] showed that there exist
polynomials of degree $n$ such that the coefficients of an absolute
decomposition generate a field extension of degree exponential in $n$ 
over $F$.  It is conjectured that such examples exist over finite fields
as well.

\subsection{Decomposition and the Subfields of $F(x)$. }

The decompositions of a polynomial $f \in F[x]$ have
a strong correspondence with the lattice of subfields between
$F(f)$ and $F(x)$, where $F(x)$ is an algebraic extension over $F(f)$ of the
same degree as the degree of $f$ (see van der Waerden[1970] section 10.2).
This was first examined by Levi[1942]
and later by Fried and MacRae[1969a,b].
Let $n\in\nats$ and $\wp=\ordfact$ be an 
ordered factorisation of $n$.
Let $\bbL$ be the set of all subfields of $F(x)$.
Define
\[
\fld =
\left\lbrace
(f,(F_m,\ldots,F_1)) \in \monicpoly\times\bbL^m\Biggm |
\eqalign{
& F_m=F(f),~F_0=F(x),\cr
& F_m\subseteq F_{m-1}\subseteq\cdots\subseteq F_1\subseteq F_0,\cr
& [F_{i-1}:F_i]=r_i,\; 1\leq i\leq m\cr
}
\right\rbrace
\]
where $[F_{i-1}:F_i]$ is the algebraic degree of $F_{i-1}$ over $F_i$.

Let $f\in F[x]$ be of degree $n$.  Also,
let $(f,(f_m,f_{m-1},\ldots,f_1))\in\dec$ and for $1\leq i\leq m$
let $h_i=f_i\circ f_{i-1}\circ \cdots f_1$ and $F_i=F(h_i)$, the field
$F$ with $h_i$ adjoined.
Define $\G[D,F]: DEC^F_\wp\rightarrow FIELDS^F_\wp$ by
$(f,(f_m,f_{m-1},\ldots,f_1))\mapsto (f,(F_m,F_{m-1},\ldots,F_1))$.
This is a map into $\fld$ by the fact that $[F_{i-1}:F_i]=r_i$.

\medskip
\begin{theorem}
$\G[D,F]$ is a bijection.
\end{theorem}
\begin{proof}
Different decompositions give
rise to different chains of fields because for any $h_i,h_i^\prime\in F[x]$,
$F(h_i)=F(h_i^\prime)$ if and only if $h_i=ah_i^\prime+b$ for some 
$a,b\in F$, $a\neq 0$.  If $h_i,h_i^\prime\in F[x]$ are monic  with 
$h_i(0)=h^\prime_i(0)=0$, this implies $h_i=h_i^\prime$.
Thus $\G[D,F]$ is injective.

Showing that $\G[D,F]$ is surjective is somewhat less trivial.
Let $f\in F[x]$ be monic of degree $n$.  Then $F(x)$ is a finite
extension of $F(f)$ of degree $n$.  Let $L$ be a field such that
$F(f)\subseteq L\subseteq F(x)$.  Then $L=F(h)$ for some $h\in F(x)$.
Since $h$ generates $L$ over $F$ and 
$f\in L$, $f=g\circ h$ for some $g\in F[x]$.  Thus $h$ is a root of 
$f-g(y)\in F(x)[y]$.  Since $f-g(y)$ is also in $F[x,y]$, all roots
in $F(x)$ must be in $F[x]$ (see van der Waerden[1970] section 5.4) and
$h\in F[x]$.
As $F(ah+b)=F(h)$ for
$a,b\in F$ and $a\neq 0$, we can assume $h$ is monic with $h(0)=0$.
Assume $h'\in F[x]$ is also monic with $h'(0)=0$ and $L=F(h')$.
We know $h'$ and $h$ have the same degree, that is $[F(x):L]$,
and because $h'\in F(h)$, $h'=ch+d$ for some $c,d\in F$, $c\neq 0$.
But both are monic with constant coefficient zero, so $h=h'$ and
$h$ is unique.
Assume $h$ has degree $s$.
The field $L$ is an algebraic extension of $F(f)$
of degree $r=n/s$,
and the degree of $g\in F[x]$ is $r$.

Now assume $(f,(F_m,F_{m-1},\ldots,F_1))\in\fld$.  Let $h_i\in F[x]$
with $h_i(0)=0$
be the unique monic generator of $F_i$ as above.
Because $F(h_{i-1})\supseteq F(h_i)$, we know 
$h_i=f_i\circ h_{i-1}$, for some (unique) $f_i\in F[x]$,
for $1\leq i<m$.  The degree of $F(h_{i-1})$ over $F(h_i)$ is
$r_i$, so the degree of $f_i$ is $r_i$.  
Because $f$ may have a non-zero constant term,
$f=h_m+c$, where $c\in F$ is the constant term of $f$.  As before, 
$F(h_{m-1})\supseteq F(h_m)$ so there exists a unique $\fbar_m$ of degree
$r_m$ such
that $h_m=\fbar_m\circ h_{m-1}$.  Letting $f_m=\fbar_m+c$,
it follows that $(f,(f_m,f_{m-1},\ldots,f_1))\in\dec$.  
It is easily seen that 
$\G[D,F](f,(f_m,f_{m-1},\ldots,f_1))=(f,(F_m,F_{m-1},\ldots,F_1))$
and so $\G[D,F]$ is surjective and hence bijective.
\QED
\end{proof}

Let $f\in F[x]$ be separable of degree $n$ (ie. 
${\partial\over{\partial x}} f\neq 0$).
In the separable case we can
study the lattice of fields between $F(f)$ and $F(x)$ by looking
at the Galois group of $F(x)$ relative to $F(f)$.  This was first done
in Dorey and Whaples[1974] for the set of additive polynomials
(a subset of $F[x]$ which will be dealt with in detail in a later section).
As $F(x)$ is
not necessarily a normal, separable, extension of $F(f)$, 
we construct the splitting field
$\Omega$ of the minimal polynomial of $x$ over $F(f)$. This minimal
polynomial is
\[
\Phi_f(y)=f(y)-f\in F(f)[y]\subseteq F(x)[y] ~,
\] 
since we know that $x$ has degree $n$ over $F(f)$ 
and $x$ satisfies $\Phi_f$ which also has degree $n$.  
Because $f$ is separable, $\Phi_f$ is separable, so
the field $\Omega$ is a normal, separable, extension of $F(f)$ containing 
$F(x)$. Let $\varG_f = \gal(\Omega/F(f))$, the Galois group of $\Omega$ 
relative to $F(f)$, and let $\varG_x \subseteq \varG_f$ 
be the subgroup fixing $F(x)$  pointwise. 
Let $\bbG$ be the set of all subgroups of $\varG_f$.
For $n\in\nats$ and $\wp=\ordfact$, an ordered factorisation of $n$, define
\[
\grp =
\left\lbrace
(f,(\varG_m,\ldots,\varG_1)) \in \monicpoly\times\bbG^m\Biggm |
\eqalign{
& \varG_m = \varG_f,~\varG_0=\varG_x\cr 
& \varG_m\supseteq\varG_{m-1}\supseteq\cdots\supseteq\varG_1\supseteq\varG_0\cr
& (\varG_i:\varG_{i-1})= r_i,\; 1\leq i\leq r\cr
}
\right\rbrace
\]
where $(\varG_i:\varG_{i-1})$ is the index of $\varG_{i-1}$ in $\varG_i$.

Let $f\in F[x]$ be separable of degree $n$ and
let $(f,(F_m,F_{m-1},\ldots,F_1))\in\fld$.
As above, let $\Omega$ be the splitting field of $\Phi_f$ and
let $\varG_f=
\gal(\Omega/F(f))$.  For $1\leq i\leq m$ let
$\varG_i\subseteq\varG_F$ be the group of automorphisms fixing $F_i$
pointwise.
Define $\G[F,G]: \fld\rightarrow\grp$ by 
\[
(f,(F_m,F_{m-1},\ldots,F_1))\mapsto (f,(\varG_m,\varG_{m-1}\ldots,\varG_1)).
\]
This map is simply the one described in the fundamental theorem of
Galois theory (see van der Waerden[1970] section 8.1-8.3).

\medskip
\begin{theorem}
If $f$ is separable then $\G[F,G]$ is a bijection.
\end{theorem}
\begin{proof}
By the  fundamental
theorem of Galois theory,
there is an inclusion inverting  bijection between fields between $F(x)$
and $F(f)$ and groups between $\varG_x$ and $\varG_f$.
An automorphism group $\varH$ such that 
$\varG_x\subseteq\varH\subseteq\varG_f$
corresponds to the field $L$ such that 
$F(x)\supseteq L \supseteq F(f)$
which it leaves fixed pointwise.  
Thus each chain of fields
\[
F(f)=F_m\subseteq F_{m-1}\subseteq\cdots\subseteq F_1\subseteq F(x)
\]
corresponds uniquely to a tower of groups
\[
\varG_f=\varG_m\supseteq\varG_{m-1}\supseteq\cdots\supseteq\varG_0=\varG_x.
\]
Also by the fundamental theorem, 
$(\varG_i:\varG_{i-1})=[F_{i-1}:F_i]=r_i$.
As $\G[F,G]$ is exactly this Galois mapping,
the fact that it is a bijection follows immediately.~\QED
\end{proof}

\subsection{ Separated Polynomials}

From the correspondence between decompositions and 
fields between $F(f)$ and $F(x)$
we get a  useful structural result.
This was originally due to Fried and MacRae[1969b]
and was later extended to the multivariate case
by Evyatar and Scott (this will be dealt with in a subsequent section).
Fried and MacRae[1969b] introduce a more general version of the
polynomials $\Phi_f=f(y)-f(x)\in F(f)[y]$ and 
$\Phi_h=h(x)-h(y)\in F(h)[y]$ previously described.
Let $F$ be an arbitrary field with independent indeterminates $x$ and $y$
over $F$.
A polynomial
$\Upsilon\in F[x,y]$ 
is said to be {\it separated} if $\Upsilon(x,y) = f_1(x)-f_2(y)$
where $f_1,f_2 \in F[x]$.
They then showed the following theorem linking separated polynomials with
the simultaneous bidecomposition of two polynomials with a common
left composition factor:

\medskip
\begin{fact}
Let  $f_1,f_2,h_1,h_2 \in F[x]$.  
Then $h_1(x)-h_2(y)|f_1(x)-f_2(y)$ if and only if there exists
a polynomial $g\in F[x]$ such that $f_1=g\circ h_1$ and
$f_2=g\circ h_2$.
\end{fact}

\noindent
If we let $f_1=f_2$ and $h_1=h_2$ we immediately have the following corollary:

\medskip
\begin{corollary}
Let $f,h\in F[x]$ be monic of degrees $n$ and $s$ respectively
with $h(0)=0$.
The following are equivalent:
\begin{list}{}{\genlist}
\item[(i)] There exists a $g\in F[x]$ such that $f=g\circ h$.
\item[(ii)] $\Phi_h|\Phi_f$.
\end{list}
\end{corollary}

We can now apply this theorem to get another formulation of
general decompositions. 
Let $\bbS=\lbrace h(x)-h(y)\in F[x,y] : h\in\monicpoly\rbrace$. 
Let $n\in\nats$ and $\wp=\ordfact$, an ordered factorisation of $n$.
Also, let $d_i=\prod_{1\leq j\leq i} r_j$.  Define
\[
\sep =
\left\lbrace
(f,(\Phi_m,\Phi_{m-1}\ldots,\Phi_1))\in \monicpoly\times\bbS^m \Bigm |
\eqalign{ 
&\deg_x \Phi_i=d_i,\,\Phi_m=\Phi_f \cr
&\Phi_i|\Phi_{i+1}~\hbox{for}~1\leq i<m\cr
}
\right\rbrace.
\]

\noindent
Let $(f,(f_m,f_{m-1},\ldots,f_1))\in\dec$ and, for $1\leq i\leq m$,
let 
\[
u_i=f_i\circ f_{i-1}\circ\cdots\circ f_1(x)
-f_i\circ f_{i-1}\circ\cdots\circ f_1(y)\in F[x,y].
\]
By corollary 1.7, $u_i|u_{i+1}$.  Define the map
$\G[D,S]:\dec \rightarrow\sep$ by
\[(f,(f_m,f_{m-1},\ldots,f_1))\mapsto (f,(u_m,u_{m-1}\ldots,u_1)).\]

\medskip
\begin{theorem}
$\G[D,S]$ is a bijection.
\end{theorem}
\begin{proof}
As distinct decompositions
will give a distinct sequences of $u_i$'s for $1\leq i\leq m$, 
this is an injective mapping from $\dec$ to $\sep$.

Now assume $(f,(v_m,v_{m-1},\ldots,v_1))\in\sep$ where
$v_i(x,y)=g_i(x)-g_i(y)$ for $1\leq i\leq m$. 
By corollary 1.7, we know that for
$1<i\leq m$, $g_i=f_i\circ g_{i-1}$ for some $f_i\in F[x]$ of
degree $r_i$.  Thus $(f,(f_m,f_{m-1},\ldots,f_3,f_2, g_1))\in\dec$.
Each member of $\sep$ will be mapped to a different
member of $\dec$ so there is an injection from $\sep$ to $\dec$.
This is obviously the inverse of $\G[D,S]$, and so $\G[D,S]$
is a bijection.
\QED
\end{proof}

\subsection{Multidimensional Block Decompositions}

Kozen and Landau[1986] developed another formulation of the bidecomposition
of a polynomial $f\in F[x]$ based on its Galois group $G_f$, and this
group's action on the roots of $f$ in its splitting field.  The roots
are partitioned by $G_f$ into blocks or systems of imprimitivity
(see van der Waerden[1970], section 7.9).  
A necessary and sufficient condition in
terms of 
these blocks is given
for there to be a corresponding bidecomposition of $f$.
We extend this formulation to general decompositions of $f$
corresponding to a given ordered factorisation $\wp$.

Let $\calU$ be a set.  A {\it multiset} $S$ 
over $\calU$ is any map $S:\calU\rightarrow\nats$.
An element $\alpha\in\calU$ is an element of $S$ ($\alpha\in S$) if and
only if $S(\alpha)>0$.  A multiset can be viewed as an extension of the
characteristic function of a set.
A multiset $T$ is a submultiset of a multiset $S$
($T\subseteq S$) if for all $\alpha\in \calU$, $T(\alpha)\leq S(\alpha)$.
If $\sigma:\calU\rightarrow\calU$ is a map, we consider the
multiset $\sigma S$ to be defined such that for all $\alpha\in\calU$
\[
(\sigma S)(\alpha)=\sum_{{\beta\in S}\atop {\sigma\beta=\alpha}} S(\beta).
\]
The cardinality of a multiset $S$ is
\[
|S|=\sum_{\alpha\in S}S(\alpha).
\]
At first reading, the reader is encouraged to think of multisets
as sets $S\subseteq\calU$ (or equivalently, as the characteristic functions
of sets); if the polynomial $f$ to be decomposed is squarefree,
indeed only such sets will occur.

We will see that decompositions into, say, three composition factors,
correspond in a natural way to certain ``multisets of multisets
of multisets of roots''.  
We introduce these ``typed'' objects over some
set $\calU$ as follows.  A {\it multiset with one level over $\calU$}
is a multiset over
$\calU$ and, for $i>1$, a {\it multiset with $i$ levels over $\calU$}
is a multiset over
the set of multisets with $i-1$ levels over $\calU$.  
Let $B$ be a multiset with $m$ levels over $\calU$.
A multiset $C$ is a level $k$
member of $B$ if
\[
C\in B_{k-1} \in\cdots\in B_1\in B.
\]
Notice that the structure of $B$ implies that $C$ must be a multiset with
$m-k$ levels.  $C$ also has a natural multiplicity within $B$, namely
\[
B(B_1)\cdot B_1(B_2)\cdots B_{k-2}(B_{k-1})\cdot B_{k-1}(C).
\]
This allows us to ``flatten'' the top $k$ levels of $B$ and speak
of the multiset of all multisets at level $\ell$ in $B$.  We denote this
multiset  with $m-k+1$ levels as $\ulms[B,\ell]$.

Let $m\in\nats$ and $\wp=(r_m,r_{m-1},\ldots,r_1)\in\nats^m$.
A multiset $B$ with $m$ levels over $\calU$ is a {\it $\wp$-block}
if either
\begin{list}{}{\genlist}
\item[(i)] $m=1$ and $B$ is a multiset over $\calU$ of cardinality $r_1$, or
\item[(ii)] $m>1$ and $B$ is a multiset with cardinality $r_m$ 
of $(r_{m-1},r_{m-2},\ldots,r_1)$-blocks over $\calU$.
\end{list}

\noindent
Let $\bbB_F$ be the set of all multisets with $i$ levels over $\Fbar$
for all $i>0$,
where $\Fbar$ is a fixed algebraic closure of $F$. Define the set
\[
BLOCKS^F_\wp=\left\{
(f,B)\in\bbP_F\times\bbB_F\,\Biggm|\;
\eqalign{
&\hbox{$B$ is a $\wp$-block over $\Fbar$ such that}\cr
&\hbox{~~~~~$\displaystyle f=\prod_{\alpha\in\ulms[B,m]} 
                        (x-\alpha)^{\ulms[B,m](\alpha)}$}\cr
}
\right\}.
\]

\noindent
Let $f\in F[x]$ be of degree $n$ with splitting field $K\subseteq \Fbar$ 
and Galois group
$G_f=\gal(K/F)$ and let $\wp=\ordfact$ be an ordered factorisation of $n$.  
A $\wp$-block $B$ over $K$ is a {\it $\wp$-block decomposition}
of $f$ if
\begin{list}{}{\genlist}
\item[(i)]
$\displaystyle f=\prod_{\alpha\in\ulms[B,m]} (x-\alpha)^{\ulms[B,m](\alpha)}$,
and
\item[(ii)] for any $\alpha,\beta\in\ulms[B,m]$ and $\sigma\in G_f$ such
that $\sigma\alpha=\beta$, and for $1\leq \ell<m$ and any $C$,$D$ over
$K$ with $C,D\in \ulms[B,m-\ell]$ such that $\alpha\in \ulms[C,\ell]$ 
and $\beta\in \ulms[D,\ell]$,
it is true that $\sigma \ulms[C,\ell]=\ulms[D,\ell]$.
\end{list}

\noindent
A $\wp$-block decomposition is said to be {\it functional} if there exist
monic  polynomials $h_1,h_2,\ldots,h_{m-1}\in F[x]$ 
such that for $1\leq \ell<m$, $h_\ell(0)=0$ and for all
$C\in\ulms[B,m-\ell]$ there exists a $\gamma_\smC\in K$ such that
\[
\prod_{\alpha\in\ulms[C,\ell]} 
(x-\alpha)^{\ulms[C,\ell](\alpha)} = h_\ell+\gamma_\smC.
\]
Define
\[
FBLOCKS^F_\wp=\left\lbrace (f,B)\in\bbP_F\times\bbB_F\Bigm |
\eqalign{ B~\hbox{is a functional $\wp$-block}\cr 
          \hbox{decomposition of}~f\cr}
\right\rbrace.
\]
Note that $FBLOCKS^F_\wp\subseteq BLOCKS^F_\wp$.

We now give an example of a $\wp$-block decomposition.
Let $p\in\nats$ be prime and let $F=GF(p)$.  Let $f\in F[x]$ be irreducible
of degree $n=12$ with splitting field $K=F[z]/(f)$ and Galois group
$G_f=\gal(K/F)=\{\sigma_j:x\rightarrow x^{p^i}~\hbox{for}~0\leq j<12\}$. 
We will
exhibit a block decomposition of $f$ in $BLOCKS^F_\wp$, where $\wp$ is
the ordered factorisation $(2,3,2)$.
Since $F$ is finite and $f$ is irreducible, $f$ has roots
$\{\alpha,\alpha^p,\alpha^{p^2},\ldots,\alpha^{p^{11}}\}\subseteq K$.
First, we find a block decomposition $B$ of $f$ in $BLOCKS^F_{(6,2)}$.  
Let $B=\{ C_0,C_1,\ldots,C_5\}$, where $C_i=\{\alpha^{p^i},\alpha^{p^{i+6}}\}$
for $0\leq i<6$.  Condition (i) of the definition of a $\wp$-block trivially
holds.  If, for $0\leq i,j,k<12$, 
$\sigma_k\alpha^{p^i}=\alpha^{p^{i+k}}=\alpha^{p^j}$, 
then $i+k\equiv j\bmod 12$ and 
$\sigma_k\alpha^{p^{i+6}}=\alpha^{p^{i+6+k}}=\alpha^{p^{j+6}}$ 
for $0\leq i,j,k<12$ since
$i+6+k\equiv j+6\bmod 12$.
Thus, condition (ii) in the definition holds as well.
In a similar way we find that
$A=\{D_0,D_1\}$ where $D_i=\{ \alpha^{2j+i}:0\leq j<5\}$ for $0\leq i<2$
is a decomposition of $f$ in $BLOCKS^F_{(2,6)}$.
Combining these two decomposition, it follows that
\[
E=\Bigl\{
     \bigl\{
         \{\alpha,\alpha^{p^6}\},
         \{\alpha^{p^2},\alpha^{p^8}\},
         \{\alpha^{p^4},\alpha^{p^{10}}\}
     \bigr\}\;,\;
     \bigl\{ 
         \{\alpha^p,\alpha^{p^7}\},
         \{\alpha^{p^3},\alpha^{p^9}\},
         \{\alpha^{p^5},\alpha^{p^{11}}\}
     \bigr\}
   \Bigr\}
\]
is a block decomposition of $f$ in $BLOCKS^F_{(2,3,2)}$.

We now proceed to describe a bijective map $\G[D,B]$ from 
$\dec$ to $FBLOCKS^F_\wp$.  We first define $\G[D,B]$ from
$\dec$ to $BLOCKS^F_\wp$.  We then show it is a map to $FBLOCKS^F_\wp$, and finally
that it is bijective.

Once again, let $n,m\in\nats$ and $\wp=\ordfact$, an ordered factorisation
of $n$.  Also, let $(f,(f_m,f_{m-1},\ldots,f_1))\in\dec$, where
$f\in F[x]$ has splitting field $K$.  We define the map $\G[D,B]$
recursively as follows.  If $m=1$, let $B$ be the multiset of roots of $f$.
It follows immediately that $B$ is a $\wp$-block of the roots of $f$ in $K$,
so we let $\G[D,B](f,(f))=(f,B)$.

Now assume $m>1$.  We know $f=f_m\circ h_{m-1}$ where 
$h_\ell=f_{\ell-1}\circ f_{\ell-2}\circ\cdots\circ f_1\in F[x]$ for
$1\leq \ell\leq m$. Let $D_m$
be the multiset of roots of $f_m$ in $K$ (we know they are in $K$ by
lemma 1.2).  Then
\[
f=\prod_{\alpha\in D_m}(h_{m-1}-\alpha)^{D_m(\alpha)}.
\]
For each $\alpha\in D_m$, let 
$E_\alpha=\G[D,B](h_{m-1}-\alpha,(f_{m-1}-\alpha,f_{m-2},f_{m-3},\ldots,f_1))$,
an $(r_{m-1},r_{m-2},\ldots,r_1)$-block over $K$ of the roots in $K$ of 
$h_{m-1}-\alpha$ by the recursive definition.  For each
$(r_{m-1},r_{m-2},\ldots,r_1)$-block $C$ over $K$, define the multiset $B$ such
that
\[
B(C)=\left\{
\eqalign{
&D_m(\alpha) & ~~~\hbox{if $C=E_\alpha$ for some $\alpha\in D_m$},\cr
&~~0 & ~~~\hbox{otherwise}.\cr
}
\right.
\]
$B$ is the multiset of the $E_\alpha$'s (for all $\alpha\in D_m$) with
appropriate multiplicity.  This is a $\wp$-block over $K$ of the roots of $f$
and hence $(f,B)\in BLOCKS^F_\wp$. We therefore define
$\G[D,B](f,(f_m,f_{m-1},\ldots,f_1))=(f,B)$.
We have completely described the map $\G[D,B]:\dec\rightarrow BLOCKS^F_\wp$.

\medskip
\begin{lemma}
$\G[D,B]$ is a map from $\dec$ to $FBLOCKS^F_\wp$.
\end{lemma}
\begin{proof}
Let $(f,(f_m,f_{m-1},\ldots,f_1))\in\dec$ and let $(f,B)$ be its image
in $BLOCKS^F_\wp$ under $\G[D,B]$.  
It follows immediately that $(f,B)$ is functional from the definition
of $\G[D,B]$. 
We must also show that condition (ii)
in the definition of $\wp$-block decomposition holds for $(f,B)$.

Assume $\alpha,\beta\in\ulms[B,m]$ and $\sigma\in G_f$ such that
$\sigma\alpha=\beta$.  Let $\ell\in\nats$ such that $1\leq\ell< m$ and  
let $C,D\in\ulms[B,m-\ell]$.
We know
\[
\eqalign{
\prod_{a\in \ulms[C,\ell]}(x-a)^{\ulms[C,\ell](a)} & = h_\ell+\gamma,\cr
\prod_{b\in \ulms[D,\ell]}(x-b)^{\ulms[D,\ell](b)} & = h_\ell+\delta,\cr
}
\]
for some $\gamma,\delta\in K$ by the definition of $\G[D,B]$.   Since
\[
\eqalign{0
&=\sigma(h_\ell(\alpha)+\gamma)\cr
&=h_\ell(\sigma\alpha)+\sigma\gamma\cr
&=h_\ell(\beta)+\sigma\gamma\cr
}
\]
and $h_\ell(\beta)+\delta=0$,
we also know $\delta=\sigma\gamma$.  Furthermore, since
\[
\eqalign{
\sigma(h_\ell+\gamma)
=&\sigma\prod_{a\in \ulms[C,\ell]}(x-a)^{\ulms[C,\ell](a)}\cr
=&\prod_{a\in\ulms[C,\ell]}(x-\sigma a)^{\ulms[C,\ell](a)},\cr
}
\]
and
\[
\eqalign{
\sigma(h_\ell+\gamma)=&h_\ell+\delta\cr
=&\prod_{b\in\ulms[D,\ell]}(x-b)^{\ulms[D,\ell](b)},\cr
}
\]
there is a bijection between the linear factors (over $K$)
of $h_\ell+\gamma$ and $h_\ell+\delta$.  As there is a trivial
bijection between the linear factors of a polynomial over its splitting
field and the multiset of roots of that polynomial,
$\sigma\ulms[C,\ell]=\ulms[D,\ell]$.  Therefore $\G[D,B]$ is a map from
$DEC^F_\wp$ to $FBLOCKS^F_\wp$.
\QED
\end{proof}

\medskip
\begin{theorem}
$\G[D,B]$ is a bijection.
\end{theorem}
\begin{proof}
$\G[D,B]$ is an injection since each different decomposition gives a
different sequence $h_1,h_2,\ldots,h_m$, and hence a different
block decomposition.  We now show it is also surjective by induction
on $m$.
Let $(f,B)\in FBLOCKS^F_\wp$.
If $m=1$ then $B$ is simply the multiset of roots of $f$ and
$\G[D,B](f,(f))=(f,B)$.  Assume $m$ is greater than one.
Then
\[
\eqalign{
f=& \prod_{\alpha\in\ulms[B,m]}(x-\alpha)^{\ulms[B,m](\alpha)}\cr
= & \prod_{D\in B}
    (\prod_{\alpha\in\ulms[D,m-1]}
    (x-\alpha)^{\ulms[D,m-1](\alpha)})^{B(D)}\cr
= & \prod_{D\in B} (h_{m-1}-\gamma_D)^{B(D)}\cr
= & \prod_{D\in B} (x-\gamma_D)^{B(D)}\circ h_{m-1}\cr
}
\]
for some $\gamma_D\in K$ for each $D\in B$.
It follows that there exists a polynomial $f_m\in K[x]$ such that
\[
f_m=\prod_{D\in B} (x-\gamma_D)^{B(D)}
\]
and $f=f_m\circ h_{m-1}$.  
By lemma 1.1, $f_m\in F[x]$ and this $f_m$ is unique.
Now let $C\in\ulms[B,m-\ell]$ for any $\ell$ with $1<\ell< m$.  Then
\[
\eqalign{
\prod_{\alpha\in\ulms[C,\ell]}(x-\alpha)^{\ulms[C,\ell](\alpha)}
= & h_\ell-\delta~~~~\hbox{for some}~\delta\in K\cr
= & \prod_{D\in C}
    (\prod_{\alpha\in\ulms[D,\ell-1]}
    (x-\alpha)^{\ulms[D,\ell-1](\alpha)})^{C(D)}\cr
= & \prod_{D\in C} (h_{\ell-1}-\gamma_D)^{C(D)}\cr
= & \prod_{D\in C} (x-\gamma_D)^{C(D)}\circ h_{\ell-1}\cr
}
\]
for some $\gamma_D\in K$ for each $D\in C$.
So there exists a polynomial $g_\ell\in K[x]$ such that
\[
g_\ell=\prod_{D\in C} (x-\gamma_D)^{C(D)}
\]
and $h_\ell-\delta=g_\ell\circ h_{\ell-1}$.  
Rearranging this,
$h_\ell=(g_\ell+\delta)\circ h_{\ell-1}$ and by lemma 1.1, 
$f_\ell=g_\ell+\delta\in F[x]$ and this $f_\ell$ is unique.
This shows $h_\ell=f_\ell\circ h_{\ell-1}$ for some uniquely determined
$f_\ell\in F[x]$ for $1<\ell<m$.

It follows that $f=f_m\circ f_{m-1}\circ\cdots\circ f_3\circ f_2\circ h_1$
where $f_i\in F[x]$ is monic of degree $r_i$ for $1<i\leq m$
and $\deg h_1=r_1$.  Therefore
$(f,(f_m,f_{m-1},\ldots,f_3,f_2,h_1))\in\dec$ and
$\G[D,B](f,(f_m,f_{m-1},\ldots,f_3,f_2,h_1))=(f,B)$.  This means that
$\G[D,B]$ is surjective and hence bijective.
\QED
\end{proof}

\subsection{Chebyshev Polynomials}

The Chebyshev polynomials, $T_i\in\cmplx[x]$ for $i\in\nats$, 
are usually defined over the complex numbers by the identity
\[
T_i(\cos\theta)=\cos i\theta.
\]
From the trigonometric identity
\[
\cos \theta_1+ \cos\theta_2=2\cos\left({{\theta_1+\theta_2}\over 2}\right)
\cos\left({{\theta_1-\theta_2}\over 2}\right),
\]
we get
\[
\cos i\theta+\cos((i-2)\theta)=2\cos((i-1)\theta)\cos\theta
$$
and
\[
T_i(\cos\theta)+T_{i-2}(\cos\theta)=2\cos\theta T_{i-1}(\cos\theta).
\]
This gives the defining recurrence relation
\[
\eqalign{
T_0 =& 1,\cr
T_1 =& x,\cr
T_i =& 2xT_{i-1}-T_{i-2},\qquad(i>1)\cr
}
\]
so that
\[
\eqalign{
T_2 = & 2x^2-1,\cr
T_3 = & 4x^3-3x,\cr
T_4= & 8x^4-8x^2+1,\cr
\vdots\cr
}
\]

Note that $T_i\in\ints[x]$ for all $i\in\nats$, so Chebyshev polynomials
are in fact well defined (by this recurrence) in arbitrary fields
of arbitrary characteristic, and have coefficients in the prime
field of this characteristic.  We will prove a number of useful
theorems concerning Chebyshev polynomials over arbitrary fields. 
Obviously, no analytic properties of trigonometric
functions have meaning in fields of positive characteristic, so we will
not make use of any of these.

If $F$ has characteristic two, then 
\[
\eqalign{
T_0 = & 1,\cr
T_1 = & x,\cr
T_i = & T_{i-2}\qquad\hbox{for}~i>2.\cr
}
\]
Therefore $T_i=1$ if $i$ is even and $T_i=x$ if $i$ is odd.

Let $F$ be any field of characteristic $p\neq 2$, and for $i\in\nats$, let
$T_i$ be the $i^{\rm th}$ Chebyshev polynomial.  A quick examination of
the defining recurrence reveals that $\deg T_i=i$.

\medskip
\begin{theorem}
\[
T_i\left({{x+x^{-1}}\over{2}}\right)={{x^i+x^{-i}}\over 2}.
\]
\end{theorem}
\begin{proof}
We will proceed by induction on $i$.  Easily, the theorem
holds for $T_0$ and $T_1$.  Assume it holds for $T_j$ with $0\leq j<i$.
Then
\[
\eqalign{
T_i\left({{x+x^{-1}}\over 2}\right)
= & 2\cdot{{x+x^{-1}}\over 2}\cdot T_{i-1}\left({{x+x^{-1}}\over 2}\right)
    -T_{i-2}\left({{x+x^{-1}}\over 2}\right)\cr
= & (x+x^{-1})\left({{x^{i-1}+x^{-(i-1)}}\over 2}\right)
    -{{x^{i-2}+x^{-(i-2)}}\over 2}\cr
= & {{x^i+x^{-i}}\over 2}\cr
}
\]
and the theorem holds for all $T_i$, $i\in\nats$.
\QED
\end{proof}

Using theorem 1.11, we can show the following fact about the composition
of Chebyshev polynomials over arbitrary fields $F$ of characteristic $p$.

\medskip
\begin{theorem}
For $i,j\in\nats$, $T_i\circ T_j=T_{ij}=T_j\circ T_i$.
\end{theorem}
\begin{proof}
If $F$ has characteristic two, then the theorem holds trivially.
If the characteristic $p$ of $F$ does not equal two then,
\[
\eqalign{
T_i\circ T_j\left({{x+x^{-1}}\over 2}\right)
= & T_i\left({{x^j+x^{-j}}\over 2}\right)\cr
= & {{x^{ij}+x^{-ij}}\over 2}\cr
= & T_{ij}\left({{x+x^{-1}}\over 2}\right )\cr
}
\]
From this identity in $F(x)$, we conclude that $T_i\circ T_j=T_{ij}$.
\QED
\end{proof}

In fields of characteristic $p>2$, a useful theorem can be shown about
the Chebyshev polynomials of degree $p^i$ for $i\geq 1$.

\medskip
\begin{theorem}
Let $F$ be any field of characteristic $p>2$.
For $i\in\nats$, $T_{p^i}=x^{p^i}$.
\end{theorem}
\begin{proof}
By theorem 1.12,
\[
T_{p^i}=\overbrace{T_p\circ T_p\circ\cdots\circ T_p}^{i\rm\;times}
\]
so it is sufficient to show $T_p=x^p$.  We know
\[
\eqalign{
T_p\left({{x+x^{-1}}\over 2}\right)
= & {{x^p+x^{-p}}\over 2}\cr
= & \left({{x+x^{-1}}\over 2}\right)^p.\cr
}
\]
From this identity in $F(x)$, we conclude that $T_p=x^p$.
\QED
\end{proof}

\subsection{Complete Rational Decompositions}

A {\it complete} rational decomposition of a 
polynomial $f\in F[x]$ is of the form
\[
f = f_m \circ f_{m-1} \circ \cdots \circ f_2 \circ f_1
\]
where each $f_i\in F[x]$ is indecomposable and nontrivial (ie. with degree
greater than one).  
A natural question to ask concerns the uniqueness of such decompositions.
As we do not want to worry about affine linear transformations of
composition factors, we consider only complete rational normal decompositions
where $f$ is monic and $f_i\in F[x]$ are monic for $1\leq i\leq m$ and
$f_i(0)=0$ for $1\leq i<m$.

Two types of ambiguous decompositions emerge.
If $u\in F[x]$, then
$(x^m\cdot u^r)\circ x^r=x^r\circ(x^m\cdot u(x^r))$ for $m,r\in\nats$.  
Call this an 
{\it exponential} ambiguity.
As seen in the previous section, the Chebyshev polynomials $T_i\in F[x]$
for $i\in\nats$ have the property that 
$T_i\circ T_j=T_j\circ T_i$.
Call this a {\it trigonometric} ambiguity.  
Ritt[1922] showed that if $F=\cmplx$, all   complete normal
decompositions differ only by ambiguities of these two forms.
Engstrom[1941] showed that in fields $F$ of characteristic zero that 
\begin{list}{}{\genlist}
\item[(i)] polynomials indecomposable over $F$ are indecomposable over
any algebraic extension of $F$ (ie. all decompositions are rational), and
\item[(ii)] all complete normal decompositions differ only by trigonometric
and exponential ambiguities.
\end{list}

\noindent
These two theorems are known as Ritt's first and second theorems.
Fried and MacRae[1969a] showed them
true  when the characteristic of $F$ is 
greater than the degree of the polynomial.

For an arbitrary field $F$ of characteristic $p$ this is not necessarily
true.  Dorey and Whaples[1974] give the following example
of two complete rational decompositions of the polynomial $f\in F[x]$:
\[
\eqalign{
f= & x^{p^3+p^2}-x^{p^3+1}-x^{p^2+p}+x^{p+1} \cr
= & x^{p+1}\circ(x^p+x)\circ(x^p-x)\cr
= & (x^{p^2}-x^{p^2-p+1}-x^p+x)\circ x^{p+1}.
}
\]
The composition factor $x^{p^2}-x^{p^2-p+1}-x^p+x$ is 
indecomposable because the composition
of two polynomials of degree $p$  can never have a term of degree $p^2-p+1$.

The various equivalent formulations to polynomial decompositions can be
extended to complete decompositions in the obvious manner. 
Let $\wp$ be an ordered factorisation of $n$ of length $m$.
Let $\cdec\subseteq\dec$ be the set of complete decompositions of
polynomials corresponding to  ordered factorisation $\wp$.  
The image of $\cdec$ in
$\fld$, $\grp$, $\sep$, and $FBLOCKS^F_\wp$ under the bijections described 
in this chapter will be called, respectively, 
$\cfld$, $\cgrp$, $\csep$ and $\cblk$.  Obviously, 
any member of any one of these sets 
will correspond to a  complete rational normal decomposition.

The sets $\cfld$ and $\cgrp$ have useful characterisations 
in their own right.
If $(f,(F_m,F_{m-1},\ldots,F_1))\in\cfld$, then 
\[
F(f)=F_m\subseteq F_{m-1}\subseteq\cdots\subseteq F_1\subseteq F(x)
\] 
is a maximal
chain of fields.  If a field did exist between $F_i$ and $F_{i+1}$
then $f_{i+1}$, the $i+1$'st composition factor from the corresponding
element \linebreak
\mbox{$(f,(f_m,f_{m-1},\ldots,f_1)) \in \cdec$},  would be decomposable.
In a similar fashion, if 
$(f,(\varG_m,\varG_{m-1},\ldots,\varG_1))\in\cgrp$, then  
\[
\varG_x\subseteq\varG_1\subseteq
\cdots\subseteq\varG_{m-1}\subseteq\varG_m=\varG_f
\]
is a maximal tower of groups.

When dealing with complete decompositions of a polynomial $f\in F[x]$,
we often wish to deal with all decompositions of $f$ regardless of the
ordered factorisations to which they correspond.  With this in mind we define
\[
\cdecall=\bigcup_{\wp\in\bbT}\cdec.
\]
where $\bbT$ is the set of all finite tuples of integers greater than one.
Similarly we can define $\cfldall$, $\cgrpall$, etc, and restate 
Ritt's second theorem in this context: For any monic $f\in F[x]$,
all decompositions of $f$ in $\cdecall$ are equivalent up to
trigonometric and exponential ambiguities.

\subsection{The Number of Indecomposable Polynomials}

It can shown that ``most'' polynomials over an arbitrary field
$F$ are indecomposable.  This can be done using an algebraic dimension
argument over an algebraically closed field and by a counting argument
over a finite field.

Let $F$ be a field, and $\bbM\subseteq F[x]$ 
be the set of monic polynomials with
constant coefficient zero.  Also, for $n\in\nats$, 
let $\bbM_n=\{f\in\bbM\,|\,\deg f=n\}$
and for $r,s\in\nats$ with $rs=n$ and $g\in\bbM_r$ and $h\in\bbM_s$, define 
$\alpha_{(r,s)}:\bbM_r\times\bbM_s\rightarrow\bbM_{rs}$ by
$\alpha_{(r,s)}(g,h)=g\circ h$ (ie. the composition function).  Assume
$f=\sum_{0\leq i\leq n} a_i x^i\in \bbM_n$ where
$a_i\in F$ for $0\leq i\leq n$.  
We define the map $\lambda_n:\bbM_n\rightarrow F^{n-1}$
by $\lambda_n(f)=(a_{n-1},a_{n-2},\ldots,a_1)$.  This is obviously
a bijective map from $\bbM_n$ to $F^{n-1}$.  
Assume $r$ and $s$ are at least two
and define $\beta_{(r,s)}:F^{r-1}\times F^{s-1}\rightarrow F^{n-1}$ by
$\beta_{(r,s)}=\lambda_n\circ\alpha_{(r,s)}
\circ (\lambda^{-1}_r\times\lambda^{-1}_s)$.
This is the composition map in $F^{n-1}$.
Let 
\[
D_{(r,s)}=\{\beta_{(r,s)}(A,B)\in F^{n-1}\,|\,A\in F^{r-1},\, B\in F^{s-1}\}
\]
be the image of $\beta_{(r,s)}$ in $F^{n-1}$.
We will show that the ``size'' of $F^{n-1}$ is ``much larger'' than
the ``size'' of 
\[D=\bigcup_{rs=n} D_{(r,s)},\]
the set of all decomposable polynomials in $\bbM_n$. Because we can
normalise any decomposition, this is in fact a general statement
about the number of indecomposable polynomials in $F[x]$.

Consider the case where $F$ is an algebraically closed field.
For $r,s\in\nats$ with $rs=n$ and $r>1$, $\Dbar_{(r,s)}$ (the Zariski 
closure of $D_{(r,s)}$) is an algebraic set of
dimension at most $r+s-2$.  Therefore
\[
\Dbar=\bigcup_{{rs=n}\atop{r,s>1}} \Dbar_{(r,s)}
\]
has dimension at most
\[
\max\{ \dim\Dbar_{(r,s)}\,:\, rs=n,\, r,s>2\}\leq {n\over 2}
\]
and this is less than the dimension $n-1$ of $F^{n-1}$.
Therefore, over an arbitrary infinite field, 
``most'' polynomials are indecomposable
even over an algebraic closure of that field, in a strong algebraic sense.

Turning to the case $F=GF(q)$ where $q=p^i$ for some prime
number $p$ and $i\in\nats$, we can make a counting argument to show that
only an exponentially small number of polynomials in $F[x]$ of degree
$n$ are decomposable.
For any ordered factorisation $(r,s)$ of $n$ with $s>1$, we know
$\# D_{(r,s)}\leq q^{r-1}q^{s-1}=q^{r+s-2}$.  Summing
over all possible ordered factorisation $(r,s)$ of $n$ where
$s>1$, we get
\[
\eqalign{
\# D\leq & \sum_{rs=n} q^{r+s-2}\cr
\leq & d(n) q^{2+n/2-2}\cr
\leq & d(n) q^{n/2}\cr
}.
\]
where $d(n)$ is the number of divisors of $n$.
From Hardy and Wright[1960] (theorem 317) we get 
$d(n)\leq c_\epsilon n^\epsilon$ for any $\epsilon>0$ and some $c_\epsilon>0$
(depending on $\epsilon$).  Fixing an $\epsilon>0$,
\[
\eqalign{
\# D \leq & c_\epsilon q^{n\over 2}\cr
\leq & c_\epsilon q^{\epsilon\log_qn+{n\over 2}}\cr
\leq & kq^{2n/3}.\cr
}
\]
for some $k>0$.
This shows that 
only an exponentially small fraction of the polynomials of degree $n$ 
over $GF(q)$ are decomposable.

\subsection{Multivariate Decomposition}

Let $F$ be an arbitrary field and let $x,x_1,\ldots,x_\ell,y,y_1,\ldots,y_\ell$
be algebraically independent indeterminates over $F$ for 
$\ell\in\posnats$. For convenience we write the sequences
$x_1,\ldots, x_\ell$ and $y_1,\ldots,y_\ell$ as $\xvec$ and $\yvec$
respectively.
For $f\in F[\xvec]$, let $\deg f$ be the total degree of $f$.
We will simply refer to this as the degree of $f$.
For
$f\in F[\xvec]$ of degree $n$, a decomposition of $f$ is a pair
$(g,h)\in F[x]\times F[\xvec]$ such that $f=g\circ h$.
Note that if $g$ has degree $r$ and $h$ has degree $s$, then 
$f$ has degree $n=rs$.
For any $\alpha\in F$, we have $f=[g\circ (x+\alpha)]\circ [(x-\alpha)\circ h]$
so we can assume $h(0,\ldots,0)=0$.
Let $(r,s)$ be an ordered factorisation of $n$.
For any positive integer $\ell$, define the set
\[ 
MDEC_{(r,s)}^{F,\ell}=\left\lbrace
(f,(g,h))\in F[\xvec]\times (F[x]\times F[\xvec])
\Bigm | 
\eqalign{
f=g\circ h,~\deg g=r,\cr
\deg h=s, h(0,\ldots,0)=0\cr
}
\right\rbrace.
\]
If $(f,(g,h))\in MDEC_{(r,s)}^{F,\ell}$ then for any $\alpha\in F$,
$(f,(g(\alpha x), \alpha^{-1} h))\in MDEC_{(r,s)}^{F,\ell}$.  We say
the two decompositions $(f,(g,h))$, and  $(f,(g(\alpha x), \alpha^{-1} h))$
are {\it linearly equivalent}.
Removing linearly equivalent decompositions from $MDEC_{(r,s)}^{F,\ell}$ and
choosing a canonical representative from each equivalence class is not
as natural as in the univariate case and will not be attempted here.
Two different approaches to this problem will be presented when dealing with
multivariate decompositions algorithmically.
As in the univariate case we define the tame case to be when
$p\nmid r$.  In von zur Gathen [1987b] it is shown that in the tame case
for any $f\in F[x]$ of degree $n$ and any ordered factorisation $(r,s)$
of $n$,
all decompositions of $f$ (if any) in $MDEC^{F,\ell}_{(r,s)}$ are linearly
equivalent.

Evyatar and Scott[1972] show the following multivariate generalisation
of the Fried and MacRae[1968a] theorem concerning separated polynomials
(see section 1.C).

\medskip
\begin{fact}
If $f,h\in F[\xvec]$ then there exists a $g\in F[x]$ such that
$f=g\circ h$ if and only if $h(\xvec)-h(\yvec)| f(\xvec)-f(\yvec)$.
\end{fact}

\noindent
Define the set $\bbW_\ell=\{ h(\xvec)-h(\yvec)|h\in F[\xvec]\rbrace$.
Also define
\[
MSEP^{F,\ell}_{(r,s)}=\left\lbrace
(f,(\Phi,\Psi))\in F[\xvec]\times (\bbW_\ell)^2 \Bigm |
\eqalign{
&\Phi=f(\xvec)-f(\yvec),~\Psi|\Phi,\cr
&\deg \Phi=rs,\; \deg\Psi=r\cr
}
\right\rbrace.
\]
Considering fact 1.14, there is a map 
$\G[MD,MS]:MDEC_{(r,s)}^{F,\ell}\rightarrow MSEP_{(r,s)}^{F,\ell}$.
Namely,  for $(f,(g,h))\in MDEC^{F,\ell}_{(r,s)}$, 
$(f,(g,h))\mapsto(f,(f(\xvec)-f(\yvec),h(\xvec)-
h(\yvec)))$.

\medskip
\begin{theorem}
$\G[MD,MS]$ is a bijection.
\end{theorem}
\begin{proof}
Assume $f,h\in F[\xvec]$ and $g,g'\in F[x]$ where $h\neq 0$ and
$f=g\circ h=g'\circ h$.  Then $g\circ h-g'\circ h=(g-g')\circ h=0$
and $g-g'=0$.  Thus $g$ is uniquely determined by $f$ and $h$ and
$\G[MD,MS]$ is injective.
Conversely, if 
$(f,(f(\xvec)-f(\yvec),h(\xvec)-h(\yvec)))\in MSEP^{F,\ell}_{(r,s)}$,
then by fact 1.14 there exists a $g\in F[x]$ such that 
$f=g\circ h$ and the inverse map is also injective.  Therefore,
$\G[MD,MS]$ is a bijection.
\QED
\end{proof}


\newpage
\section{Decomposition Algorithms}

The development of algorithms for the decomposition of polynomials
has occurred relatively recently.  Although related problems for
power series were examined by Brent and Kung[1976,1977], polynomial
decomposition algorithms (for univariate polynomials) 
were not truly examined until  Barton and Zippel[1976,1985].  
Their algorithms require an exponential number of
field operations (in the degree of the input polynomial)
and work over any field which supports a factoring algorithm.  
Alagar and Thanh[1986] showed a similar  algorithm which also
requires an exponential number of field operations.
The breakthrough came when 
Kozen and Landau[1986] developed a decomposition algorithm
for the tame case which required a polynomial number of	
field operations (in the degree of the input polynomial) as well
as giving a fast parallel algorithm.  
In von zur Gathen[1987] this result for the tame case
was improved, and a very fast
parallel algorithm was developed.
Kozen and Landau[1986] also
show a decomposition algorithm for the general univariate case
based on block decomposition, for fields supporting a polynomial
factorisation algorithm.  This algorithm requires an exponential 
number of field operations 
in the degree of the input polynomial, plus the cost of factoring the
input polynomial.   
For separable irreducible polynomials
over arbitrary fields their algorithm is shown to work in a quasi-polynomial
number of field operations.
And for irreducible polynomials over finite fields, their
algorithm requires only a polynomial number of field operations.
All this is reported in von zur Gathen, Kozen, and Landau[1987].  
Complete decompositions are dealt with in the tame case in
von zur Gathen[1987].  
We also consider computing decompositions of 
polynomials corresponding to
a given ordered factorisation of their degrees.

Multivariate polynomial decomposition in the tame case was 
examined by Dickerson[1987] and von zur Gathen[1987].  Both showed
algorithms requiring a polynomial number of field operations
(in the size of the input polynomial):
Dickerson[1987] for the ``monic'' tame case
and von zur Gathen[1987] for the tame case in general.
We present an algorithm for
multivariate decomposition over any field supporting a univariate polynomial
factoring algorithm, based on the
theorem of Evyatar and Scott[1982] and the univariate algorithm
of Barton and Zippel[1985].  In general, it will require an exponential
number of field operations.

\subsection{The Model of Computation}

The model of computation used is the ``arithmetic Boolean circuit'' (see
von zur Gathen [1986]).  This model uses inputs $x_1, x_2,\ldots, x_n$
from a field $F$.  Operations are the arithmetic (field) operations
$+$, $-$, $\times$, $/$, and Boolean operations $\land$, $\lor$,
and $\lnot$.  The connection between the arithmetic and Boolean
parts of the circuit is provided by two types of gates.  The
zero test gate gives a Boolean indication of whether or not an input
field value is zero.  The  selection gate outputs one of two input
field values depending upon the value of a third, Boolean, input.
The cost of algorithms will be measured in the number of field
operations performed.  Often, the input will be a polynomial $f\in F[x]$
and the number of field operations will be counted in terms of the
degree $n$ of $f$ and the characteristic $p$ of $F$.  If $F=GF(p^e)$ for
some $e\geq 1$, we will also consider the cost of computation over the
prime field $\ints_p$, and hence in terms of $e$ as well.

Assume we can factor an arbitrary univariate polynomial $f\in F[x]$ 
of degree $n$ into irreducible factors in $O(\sf(n))$ field operations.
Then we can also factor a multivariate polynomial 
$g\in F[x_1,x_2,\ldots,x_\ell]$ 
of total degree $n$ into irreducible factors.  
Assume this can be accomplished in $O(\sfl(n))$ field operations
(where $\sfl(n)$ is a function of the size $(n+1)^\ell$ of a dense
representation of the input).
Let $M(n)$ be such that the product of two polynomials of degree at
most $n$ can be computed in $O(M(n))$ field operations.  We can
choose $M(n)=n\log n\log\log n$ (Sch${\ddot {\rm o}}$nhage[1977],
Cantor and Kaltofen[1987]), and
$M(n)=n\log n$ if $F$ supports a Fast Fourier Transform.
Also, assume two $n\times n$ matrices can be multiplied in $O(n^\mu)$
field operations for some $\mu>2$.  Coppersmith and Winograd[1987]
show $\mu<2.38$.

In some of our algorithms we use $\pwr(S)$ to denote the set of all subsets
(the power set) of a set $S$, and $S^*$ to denote the set of finite sequences
of elements of $S$.

\subsection{Computing Right Division}

Given $f,h\in F[x]$ of degrees $n$ and $s$ respectively with $s|n$,
we would like to determine if there is a $g\in K[x]$ , where
$K$ is some algebraic extension of $F$, such that
$f=g\circ h$.  Lemma 1.1 shows us that if such a $g\in K[x]$ exists 
it will 
be in $F[x]$.  We find $g$ by the usual divide and conquer approach,
which is used in von zur Gathen [1987b].

\pagebreak
{\tt\obeylines
RightDivide: $F[x]\times F[x]\rightarrow F[x]$

~~~~Input:~~-~$f,h\in F[x]$ of degrees $n$ and $s$ respectively,
~~~~~~~~~~~~~~with $s|n$.
~~~~Output:~-~$g\in F[x]$ of degree $r$ such that $f=g\circ h$
~~~~~~~~~~~~~~if such a $g$ exists.

~~~~If $\deg f\leq 0$
~~~~~~~~then return $f\in F$.
~~~~Else if $0<\deg f<\deg h$
~~~~~~~~then Quit (there is no solution).
~~~~Else if $\deg h\leq \deg f$,
~~~~~~~~1) Let $t:=\lceil r/2\rceil$.
~~~~~~~~2) Let $v:= h^t$.
~~~~~~~~3) Find $Q,R\in F[x]$ such that 
~~~~~~~~~~~$f=Qv+R$ with $\deg R<\deg v$.
~~~~~~~~4) Recursively call RightDivide on $(R,h)$ yielding
~~~~~~~~~~~$g_0\in F[x]$ and $(Q,h)$ yielding $g_1\in F[x]$.
~~~~~~~~5) Return $g_1 x^t+g_0$.
}

\medskip
\noindent
This algorithm requires $O(M(n)\log n)$ field operations, 
with step two the dominant step at each recursive stage of the algorithm.
We have the following:

\medskip
\begin{theorem}
Given $f,h\in F[x]$, we can determine if there exists a $g\in F[x]$ 
such that $f=g\circ h$ and if so, find it 
in $O(M(n)\log n)$ field operations.
\end{theorem}

\subsection{Univariate Decomposition using \mbox{Separated~Polynomials}}

The algorithm of Barton and Zippel[1985] exploits the relationship
between separated polynomials and polynomial decompositions
described in section 1.C.
Let $F$ be an arbitrary field of characteristic $p$.  
Let $f\in F[x]$ be of degree $n$ and let $(r,s)\in\nats^2$ be an
ordered factorisation of $n$.
We present a modified version of the Barton and Zippel[1985] algorithm
conforming to our definition of the problem.

\medskip
{\tt \obeylines
SepBidecomp : $F[x]\times\nats^2\rightarrow DEC^F_*$
~~~~Input:~~-~$f\in F[x]$ monic of degree $n$.
~~~~~~~~~~~~-~$(r,s)\in\nats^2$, an ordered factorisation of $n$.
~~~~Output:~-~$(g,h)\in F[x]$ such that $(f,(g,h))\in DEC^F_{(r,s)}$
~~~~~~~~~~~~~~if such a decomposition exists.
~~~~1) Factor $f(x)-f(0)=xq_1(x)q_2(x)\cdots q_m(x)$
~~~~~~~where each $q_i\in F[x]$ is irreducible for $1\leq i\leq m$.
~~~~2) For each subset $S$ of $\lbrace 1,\ldots,m\rbrace$, 
~~~~~~~~~2.1) Let $h={\displaystyle x\prod_{i\in S}} q_i\in F[x]$.
~~~~~~~~~2.2) If $\deg h=s$, attempt to compute $g\in F[x]$ such that
~~~~~~~~~~~~~~$f=g\circ h$ using RightDivide.  If such a $g$ is found,
~~~~~~~~~~~~~~then goto step 4.
~~~~3) Quit, $f$ has no decomposition in $DEC^F_{(r,s)}$.
~~~~4) Return $(f,(g,h))\in DEC^F_{(r,s)}$.
}
By theorem 1.6, for any polynomials $f,h\in F[x]$, there exists a
$g\in F[x]$ such that $f=g\circ h$ if and only if $h(x)-h(y)|f(x)-f(y)$.
Thus, $h(x)-h(0)|f(x)-f(0)$.  By looking at all factors $h$ of $f(x)-f(0)$,
we are guaranteed to find all possible right composition factors.
Since there are $2^n$ subsets which must be checked 
for separation in step 2, the algorithm requires $O(\sf(n)+2^nM(n)\log n)$
field operations.
It does, however, work over any field for which a factorisation
algorithm exists (in both the tame and wild cases).

\subsection{Univariate Decomposition in the Tame Case}

Kozen and Landau [1986] present an algorithm for univariate decomposition
in the tame case over an arbitrary field, which uses a polynomial number
of field operations in the degree of the input polynomial.
For $f\in F[x]$ of degree $n$, they look at the decompositions of $f$ into
$(g,h)$ as solutions to systems of $n+1$ non-linear equations for the
coefficients of $f$ in terms of the coefficients of $g$ and $h$.

Specifically, for $u\in F[x]$ and $i\in \nats$, let $\coeff(u,i)\in F$  be the
coefficient of $x^i$ in $u$.
Let 
\[
\eqalign{
f= & \sum_{0\leq i\leq n}a_ix^i\in F[x] & \hbox{with}~a_i\in F
     ~\hbox{for}~ 0\leq i\leq n,\cr
g= & \sum_{0\leq i\leq r}b_ix^i\in F[x] & \hbox{with}~b_i\in F
     ~\hbox{for}~ 0\leq i\leq r,\cr
h= & \sum_{1\leq i\leq s}c_ix^i\in F[x] & \hbox{with}~c_i\in F
     ~\hbox{for}~ 1\leq i\leq s,\cr
\mu_k = & \sum_{s-k+1\leq i\leq s}c_ix^i\in F[x] \quad & \hbox{with}~c_i\in F
     ~\hbox{for}~ s-k+1\leq i\leq s~\hbox{and}~1\leq k\leq s.\cr
}
\]
If $f=g\circ h$, the following facts are easily seen to be true:
\[
\halign{#\hfil &$#$\hfil &$#$\hfil &$#$\hfil\cr

(i) & \coeff(h^r,n-e) & = \coeff(f,n-e)=a_{n-e}&~~\hbox{for}~0<e<s,\cr
(ii)~~& \coeff(h^r,n-e) & = \coeff(\mu_k^r,n-e)  &~~\hbox{for}~e<k\leq s.\cr
}
\]
This implies $a_{n-e}=\coeff(f,n-e)=\coeff(\mu_k^r,n-e)$ for $e<k\leq s$.
For $1\leq k<s$, we know
$\mu_{k+1}=\mu_k+c_{s-k}x^{s-k}$. By binomial expansion we get
\[
\eqalign{
\mu_{k+1}^r= &(\mu_k+c_{s-k}x^{s-k})^r\cr
= & \mu_k^r+rc_{s-k}x^{s-k}\mu_k^{r-1}+\varphi,\cr
}
\]
where $\varphi\in F[x]$ and $\deg \varphi\leq rs-2k$.
Thus $\coeff(\mu_{k+1}^r,rs-k)=a_{rs-k}=\coeff(\mu_k^r,rs-k)+rc_{s-k}$.
This gives the simple recurrence
\[
c_{s-k}=
{{a_{rs-k}-\coeff(\mu_k^r,rs-k)}\over r},
\]
which allows the computation of $c_s,c_{s-1},\ldots,c_1$ in turn, and
hence the calculation of $h$.  Note that it is at this point, and only
this point, that we require that $p\nmid r$.  This distinguishes the
tame and wild cases.

This system of equations uniquely determines an $h\in F[x]$ but a
$g\in F[x]$ such that $f=g\circ h$ may or may not exist.  We can determine
the existence of such a $g$, and if it exists, find it,
using {\tt RightDivide} as described earlier.  Kozen and Landau[1986]
show that a decomposition can be computed in $O(n^3)$ field operations
in general and $O(n^2\log n)$ field operations in a field which supports
a Fast Fourier Transform.
In fact, the algorithm works over any ring in which $r$ is a unit.

In von zur Gathen[1987], an improvement of the result of Kozen and
Landau[1986] is shown.  
Given a monic $f\in F[x]$ of degree $n$
and $(r,s)$ an ordered factorisation of $n$ with $p\nmid r$,
his algorithm determines 
if there exists a decomposition of $f$ in $DEC^F_{(r,s)}$
and, if so, finds it, in $O(M(n)\log n)$ field operations.  
The number of field operations required is 
dominated by the cost of {\tt RightDivide} to obtain $g$ from $f$ and $h$.
Von zur Gathen[1987] uses this algorithm for decomposition to
obtain the set of separated factors of a given polynomial $f\in F[x]$
of degree $n$ in  polynomial time in the tame case.

A very fast parallel algorithm is also presented by von zur Gathen[1987]
for univariate bidecomposition in the tame case.  He shows that over any field
$F$, given $f\in F[x]$ and $(r,s)$, an ordered factorisation of $n$
such that $p\nmid r$, it can be determined if there exists a
decomposition of $f$ in $DEC^F_{(r,s)}$, and if so, it can be found,
with a depth $O(\log n)$ circuit over $F$.

\subsection{Decomposition using Block Decomposition}

As seen in section 1.D, the polynomial decomposition problem can
be reformulated as one of finding functional block decompositions.
Let $f\in F[x]$ be monic of degree $n$, and $(r,s)$ an ordered
factorisation of $n$.
Kozen and Landau[1987] adapt the techniques from Landau and Miller[1983]
to construct all block decompositions of dimension two of $f$
in $BLOCKS^F_{(r,s)}$.
They then check each such decomposition to see if it is functional.
In general, however, their algorithm requires a number 
of field operations exponential in $n$.
If $f$ is separable and irreducible over $F$, they show that there can
be at most $n^{\log n}$ block decompositions in $BLOCKS^F_{(r,s)}$, 
and that each block decomposition can be 
constructed in a polynomial number
of field operations.  Testing a block decomposition to see if it is
functional requires only a polynomial number of field operations,
but we may have to check all of them.
Therefore, for separable irreducible polynomials $f\in F[x]$,
it can be determined if $f$ has a decomposition in
$DEC^F_{(r,s)}$, and if so, this decomposition can be
found, in a quasi-polynomial  number ($n^{O(\log n)}$) of field
operations over $F$.

The block decompositions of irreducible polynomials  over a finite field 
$F=GF(q)$ (where $q=p^e$ for some $e\geq 1$) have a stronger structure.
Let $f\in F[x]$ of degree $n$ be irreducible with splitting field $K=F[x]/(f)$,
and let $(r,s)$ be an ordered factorisation of $n$.
The roots of 
$f$ in $K$ have the form 
$\lbrace \alpha, \alpha^q,\alpha^{q^2},\ldots,\alpha^{q^{n-1}}\rbrace$
for any one root $\alpha\in K$ of $f$.
The Galois group of $K$ relative to $F$ is the set of automorphisms
$\lbrace \sigma_i:0\leq i<r\}$ with $\sigma_i\gamma=\gamma^{q^i}$ for
any $\gamma\in K$.
Kozen and Landau[1986] note that the only possible block
decomposition of $f$ has the form
$B=\lbrace C_i|0\leq i<r\rbrace $
where
$C_i=\lbrace \alpha^{q^{i+jr}}|0\leq j<s\rbrace$
for $0\leq i<r$.
It is functional (and hence corresponds to
a polynomial decomposition) if and only if
there exists an $h\in F[x]$ such that
for $0\leq i<r$, there exists a $\gamma_i\in K$ such that
\[
\prod_{0\leq j<s} (x-\alpha^{q^{i+jr}})=h-\gamma_i.
\]
The splitting field 
$K$ of $f$ is an algebraic extension of degree $n$
over $F$, so we can easily compute a representation of these roots
(in $K$), 
and check if this block decomposition is functional
in a polynomial number of field operations.
Kozen and Landau[1986] show that in this case, 
it can be determined if a polynomial $f$ has a bidecomposition
in $DEC^F_{(r,s)}$, and if so, this decomposition can be found, with
a circuit of depth $O(\log epn\log^2n)$ and size $(epn)^{O(1)}$.
We show the sequential analysis of this algorithm in the following theorem.

\medskip
\begin{theorem}
Let $F=GF(q)$ for some $q,p,e\in\nats$ with $p$ prime and $q=p^e$,
and let $f\in F[x]$ be irreducible of degree $n$.  If $(r,s)$ is an
ordered factorisation of $n$ we can determine if there exists a decomposition
of $f$ in $DEC^F_{(r,s)}$, and if so, find it, in $O( n^2M(n)\log q)$
field operations over $F$.
\end{theorem}
\begin{proof}
Let $K=F[z]/(f)$ and let $\alpha\equiv z\bmod f\in K$.  Multiplication
in $K$ requires $O(M(n))$ field operations in $F$. We can therefore compute
$\alpha^{q^{ri}}$ for all $i$ with $0\leq i<s$ with 
\[
\eqalign{
\sum_{0\leq i<s} ri\log q = & O(rs^2\log q)\cr
= & O(n^2\log q)\cr
}
\]
field operations over $K$ or $O(n^2M(n)\log q)$ field operations over $F$.
We then check if $\prod_{0\leq i<s}(x-\alpha^{q^{ri}})=h+c$ where $h\in F[x]$
and $c\in K$.  If so, there exists a $g$ such that $(f,(g,h))\in DEC^F_{(r,s)}$
and this can be found in $O(M(n)\log n)$ field operations 
by theorem 2.1.  We can compute
$\prod_{0\leq i<s}(x-\alpha^{q^{ri}})$ in $O(n^2)$ field operation over $K$ or
$O(n^2M(n))$ field operation over $F$.  Therefore the bidecomposition
problem can be solved sequentially 
for irreducible polynomials over finite fields
with $O(n^2M(n)\log q)$ field operations over $F$.
\QED
\end{proof}

\subsection{A Lower Bound on the Degrees of Splitting Fields}

Let $F$ be a field such that for any $m\in\nats$, 
there exists an algebraic
extension of $F$ of degree $m$ over $F$.
We will now show that in any such field,
for any $n\in\nats$, there exist
polynomials of degree $n$ over $F$ with splitting fields of degree exponential
in $n$ over $F$.  Note in particular that finite fields are included
in this theorem.
One implication of this is that we cannot construct a standard representation
of elements of such a splitting field in a polynomial number of 
field operations.  
It has been known for a long
time that over the rationals and some other infinite fields that
for any $n$, there exist polynomials of degree $n$ whose Galois
groups are $S_n$.  The splitting fields of these polynomials are of 
algebraic degree $n!$ over their ground fields.  In general however,
such polynomials do not exist (see van der Waerden section 8.10, Jacobson
section 4.10).  We instead make the following construction in an arbitrary
field $F$.  Let $p_i\in\nats$ be the $i^{\rm th}$ smallest rational prime. 
Also define 
\[
\eqalign{
\cheb(\ell)= & \sum_{{p~{\rm prime}}\atop p\leq\ell} \log p
~~~~\hbox{the Chebyshev}~\cheb~\hbox{function},\cr
\pi(\ell)= & \sum_{{p~{\rm prime}}\atop p\leq\ell}1~\cr
}
\]
(where all logarithms here and throughout this section are natural).
Let $f_i\in F[x]$ be an irreducible polynomial of degree $p_i$.  The
splitting field $K_i$ of $f_i$ has degree at least $p_i$ over $F$.
If $F$ is a finite field, $[K_i:F]=p_i$.
The polynomial $h_i=f_1f_2\ldots f_i$ will have splitting field $L_i$
generated by the elements of $K_1\cup K_2\cup\cdots\cup K_i$.  This is
a field of algebraic degree at least $p_1p_2\ldots p_i$ over $F$.
Let 
\[
\eqalign{
S(\ell)=\sum_{{p~{\rm prime}}\atop p\leq\ell} p,\cr
R(\ell)=\prod_{{p~{\rm prime}}\atop p\leq\ell} p.\cr
}
\]
Note that $R(\ell)=\exp(\cheb(\ell))$.

Let $n\in\nats$.  If $k=\max\{i|\,p_i\leq \ell\}$,
then $h_k$ has a splitting field of degree $R(\ell)$ over $F$.
We will show that if $\ell\leq 0.77\sqrt{n\log n}$, then 
$\deg h_k=S(\ell)\leq n$.  
It follows that $R(0.77\sqrt{n\log n})$ is exponential in $n$.  
If $f\in F[x]$ is any polynomial of degree $n$
with divisor $h_k$, we show that $f$ has a splitting field of degree
at least $\exp(0.5\sqrt{n\log n})$ over $F$.  
We will use the following bounds from Rosser and Schoenfeld[1962]:

\pagebreak
\begin{fact}
\quad\\[-1\baselineskip]
\begin{list}{}{\itemsep=0pt}
\item[(i)] $p_k<1.4k\log k$ for $k\geq 6$;
\item[(ii)] $\pi(\ell)\leq 1.26\ell/\log\ell$ for 
$\ell\geq 17$;
\item[(iii)] $\ell(1-1/\log\ell)<\cheb(\ell)$ for $n\geq 41$.
\end{list}
\end{fact}

\medskip
\noindent
First, we show an upper bound on the function
$\sigma(k)=S(p_k)=\sum_{1\leq i\leq k} p_i$, the
sum of the first $k$ primes, for $k\in\nats$.

\medskip
\begin{lemma}
For $k\geq 6$,
$\sigma(k)\leq 0.86 k^2\log k$
\end{lemma}
\begin{proof}
\[
\eqalign{
\sigma(k)\leq & 2+3+5+7+11+1.4\sum_{6\leq i\leq k} i\log i\cr
\leq &28+1.4\int^k_6(i+1)\log (i+1) di\cr
\leq &28+1.4(0.5(i+1)^2\log(i+1)-0.25(i+1)^2\Bigm |^k_6)\cr
\leq &0.86k^2\log k\qquad\mbox{\QED}\cr
}
\]
\end{proof}

\medskip
\begin{lemma}
For any $n\geq 109$, $S(0.77\sqrt{n\log n})\leq n$.
\end{lemma}
\begin{proof}
Applying lemma 2.4 to the the upper bound on the number of primes 
less than $\ell$ provided by fact 2.3(ii),
\[
\eqalign{
S(\ell)
\leq & 0.86\left({{1.26\ell}\over{\log\ell}}\right)^2\log
    \left({{1.26\ell}\over{\log\ell}}\right)\cr
\leq & 0.86(1.26)^2{{\ell^2}\over {(\log\ell)^2}}
    (\log (1.26\ell)-\log\log\ell)\cr
\leq & 1.37 {{\ell^2}\over {(\log\ell)^2}}(\log\ell+\log 1.26 -\log\log\ell)\cr
\leq & {{1.7\ell^2}\over {\log\ell}}\cr
}
\]
for $\ell\geq 17$.
For $n\geq 109$ this gives us
\[
S(0.77\sqrt{n\log n})\leq{{1.7(0.77\sqrt{n\log n})^2}
\over {\log(0.77\sqrt{n\log n})}}
\leq {n\over{\log(0.77\sqrt{n\log n})}}\leq n,
\]
and the lemma follows.
\QED
\end{proof}

\medskip
\begin{theorem}
For $n\geq 109$ there exists a polynomial $f\in F[x]$ of degree $n$
such that the splitting field of $f$ has degree over $F$ greater than
$\exp(0.5\sqrt{n\log n})$.
\end{theorem}
\begin{proof}
By lemma 2.5, if $\ell\leq 0.77\sqrt{n\log n}$, then $S(\ell)\leq n$.  Let
$\ell=\lfloor 0.77\sqrt{n\log n}\rfloor$.  Let $k=\max\{i|\,p_i\leq\ell\}$.
The polynomial $h_k$ has a splitting field $L_k$
with degree at least $R(\ell)$.
By definition
\[
\eqalign{
R(\ell)
= & \exp(\cheb(\ell))\cr
\geq & \exp (\ell(1-1/\log\ell))\cr
\geq & \exp (0.77\sqrt{n\log n}(1-1/\log(0.77\sqrt{n\log n})))\cr
\geq & \exp(0.5\sqrt{n\log n}),\cr
}
\]
for $n\geq 109$.
Therefore 
$h_k$ has a splitting field of degree at least $\exp(0.5\sqrt{n\log n})$.
Let $f$ be any polynomial of degree $n$ such that $h_k$ divides $f$
($h_k$ has degree less than $n$).  The polynomial $f$ has a splitting field
of degree at least $\exp(0.5\sqrt{n\log n})$ over $F$.
\QED
\end{proof}

\subsection{Decompositions Corresponding To \mbox{Ordered Factorisations}}

Let $f\in F[x]$ be of degree $n$ and let $\wp=(r_m,r_{m-1},\ldots,r_1)$
be an ordered factorisation of $n$.  A natural generalisation of
the computational bidecomposition problem is to compute the
decompositions of $f$ in $DEC^F_\wp$ (if any).  Let {\tt GenericBidecomp}
be an algorithm such that given $f\in F[x]$ of degree $n$, and 
$(r,s)\in\nats^2$, an ordered factorisation of $n$, it will
find the (possibly empty) set $B$ of decompositions of $f$ in $DEC^F_{(r,s)}$
using $D(n)$ field operations.  Consider the following algorithm:

\medskip
{\tt\obeylines
OrdFactDecomp: $F[x]\times\pwr(\nats)\rightarrow \pwr(DEC^F_*)$

~~~~Input:~~- $f\in F[x]$ of degree $n$,
~~~~~~~~~~~~- $\wp=(r_m,r_{m-1},\ldots,r_1)$, an ordered~factorisation~of~$n$.
~~~~Output:~- the set of decompositions of $f$ in $DEC^F_\wp$.

~~~~If $m=1$
~~~~~~~~then return $(f,(f))$
~~~~~~~~else
~~~~~~~~~~~~1) Find the set $B$ of bidecompositions
~~~~~~~~~~~~~~~$(f,(g,h))\in DEC^F_{(t_2,r_1)}$
~~~~~~~~~~~~~~~where $t_2=\prod_{2\leq i\leq m} r_i$,
~~~~~~~~~~~~~~~using GenericBidecomp.
~~~~~~~~~~~~2) Let $T:=\emptyset$.
~~~~~~~~~~~~3) For each decomposition $(f,(g,h))\in B$, 
~~~~~~~~~~~~~~~~~~~3.1) Recursively attempt to 
~~~~~~~~~~~~~~~~~~~~~~~~find a decomposition
~~~~~~~~~~~~~~~~~~~~~~~~$(g,(g_m,g_{m-1},\ldots,g_2))\in %
                      DEC^F_{(r_m,r_{m-1},\ldots,r_2)}$.
~~~~~~~~~~~~~~~~~~~3.2) If such a decomposition of $g$ is found 
~~~~~~~~~~~~~~~~~~~~~~~~add $(f,(g_m,g_{m-1},\ldots,g_2,h))$ to $T$.
~~~~~~~~~~~~4) Return $T$.
}

\medskip
This is simply a recursive application of the bidecomposition algorithm,
and can easily be seen to return the set of decompositions of $f$ in
$DEC^F_\wp$.

We now define  a {\it  $\wp$-easy} family of polynomials, a
family in which such decompositions can be computed quickly.
For $1\leq i\leq m$, let $\wp_i=(r_m,r_{m-1},\ldots,r_i)$ and
$t_i=\prod_{i\leq j\leq m} r_j$.  A set $\calF_\wp\subseteq F[x]$ is
{\it  $\wp$-easy} if
\begin{itemize}
\item[(i)]
for any $i$ with $1\leq i<m$, any $f\in\calF_\wp$ of degree $d$
has at most one decomposition in $DEC^F_{(t_{i+1},r_i)}$,
\item[(ii)]
it can be determined if such a decomposition exists, and if it does,
it can be found with $O(D(d))=d^{O(1)}$ field operations, and
\item[(iii)]
if $f\in\calF_\wp$ and $(f,(g,h))\in DEC^F_{(t_{i+1},r_i)}$ 
then $g\in\calF_\wp$.
\end{itemize}

\noindent
If $\calF_\wp\subseteq F[x]$ is a  $\wp$-easy family of
polynomials, then the bidecompositions of $f\in\calF_\wp$
in step 1 can be found in $D(n)$ field operations.  Thus, computing
{\tt OrdFactDecomp} on $f\in\calF_\wp$ with ordered factorisation $\wp$
requires $O(\sum_{1\leq i<m}D(t_i))$ field operations.
Let $\ell=\lceil\log_2n\rceil$ and let $\wq=(e_\ell,e_{\ell-1},\ldots,e_1)$,
where $e_i=2^i$ for $1\leq i\leq\ell$.  Noting that $n\leq\ell<2n$,
it follows immediately that $e_{\ell-j}\geq f_{m-j}$ for $0\leq j<m$.
Therefore $\sum_{1\leq i<m} D(t_i)\leq\sum_{1\leq i<\ell} D(e_i)$.
Since $D$ is polynomially bounded, $\sum_{1\leq i<\ell}D(e_i)=O(D(n))$.
We have shown the following theorem:

\medskip
\begin{theorem}
Let $n\in\nats$ and let $\wp$ be an ordered factorisation of $n$.  Also,
let $\calF_\wp\subseteq F[x]$ be  $\wp$-easy.  Then, for
any $f\in\calF_\wp$, we can determine if there exists a decomposition 
of $f$ in $DEC^F_\wp$, and if so, find it,  in $O(D(n))$ field operations.  
\end{theorem}

This theorem says that the general problem of computing the set of 
decompositions of a polynomial with a given ordered factorisation
is Cook reducible to the bidecomposition problem for $\wp$-easy families of
polynomials.

Two  $\wp$-easy families present
themselves immediately.  If $F$ is an
arbitrary field and $p\nmid r_i$ for $1<i\leq m$ then $F[x]$ is 
a  $\wp$-easy family of polynomials.  This follows 
because all the bidecompositions performed in step 1 are tame.  From
von zur Gathen[1987] and theorem 2.7 above, it can be determined
whether a decomposition of $f\in F[x]$ exists in $DEC^F_\wp$
and if so such a decomposition
can be found
in $O(M(n)\log n)$ field operations.  

If $F=GF(q)$ and $\calF_\wp$ is the set of polynomials 
irreducible over $F$, then $\calF_\wp$ is 
$\wp$-easy.  This follows since, if $f\in\calF_\wp$ and
$(f,(g,h))\in DEC^F_*$, then $g$ is also irreducible over $F$.
By theorem 2.2 and theorem 2.7, a decomposition of any $f\in\calF_\wp$
can be found in $O(n^2M(n)\log q)$ field operations.

\subsection{Computing Complete Univariate Decompositions}

The following method for computing complete decompositions was suggested in
von zur Gathen[1987] for the tame case and can be applied whenever
we can do
bidecomposition.
Let $D(n)$ be the number of field operations required to find a 
bidecomposition of $f\in F[x]$ corresponding to an ordered factorisation
$(r,s)$ of $n$.  The following algorithm computes a complete decomposition
of $f$ in $DEC^F_*$.

\medskip
{\tt\obeylines
CompleteDecomposition: $\monicpoly\rightarrow cDEC^F_*$
~~~~Input:~~-~$f\in F[x]$.
~~~~Output:~-~$(f,(f_m,f_{m-1},\ldots,f_1))\in cDEC^F_*$.

~~~~1) Compute the integer factorisation $n=p_1^{e_1}p_2^{e_2}\cdots p_k^{e_k}$ of $n$.
~~~~2) Let $d(n)=(e_1+1)\cdots(e_k+1)$ be the number of divisors 
~~~~~~~of $n$ and $1=r_1<r_2<\cdots<r_{d(n)}=n$ be the divisors~of~$n$.
~~~~3) Let $j>1$ be the smallest number such that $f$ has a 
~~~~~~~decomposition $(f,(g,h))\in DEC^F_{(r_j,n/r_j)}$.
~~~~4) If $j=d(n)$ then $f$ is indecomposable; otherwise 
~~~~~~~decompose $h$ recursively ($g$ is indecomposable 
~~~~~~~since any left composition factor of $g$ is a 
~~~~~~~composition factor of $f$ of smaller degree than $g$).
}

\medskip
\noindent
The number of field operations required by this algorithm is $O(D(n)d(n))$.
Hardy and Wright[1960] (theorem 317) show that $d(n)=O(n^\epsilon)$ for
all $\epsilon>0$.
Therefore, we can compute complete decompositions in
$O(D(n)n^\epsilon)$ field operations for any $\epsilon>0$.
This algorithm finds the {\it lexicographically first complete decomposition}
of $f$.

\subsection{Decomposing Multivariate Polynomials in the Tame Case}

Once again we denote the sequence of indeterminates $x_1,\ldots,x_\ell$
as $\xvec$.
We define the set $\bbW^{(\ell)}_F\subseteq F[\xvec]$ 
of {\it monic} polynomials
in $F[\xvec]$ as follows:
\[
\bbW^{(\ell)}_F=\left\lbrace
f=x_1^{d_1}x_2^{d_2}\cdots x_\ell^{d_\ell}+\fhat\biggm|\;
\eqalign{
&\fhat\in F[\xvec],\; \deg\fhat<\deg f\cr
&d_i\in\posnats,\; \deg_{x_i}\fhat\leq d_i~\hbox{for}~1\leq i\leq\ell\cr
}
\right\rbrace
\]
where $\deg f$ and $\deg \fhat$ are the total degrees of $f$ and $\fhat$
respectively.
Dickerson[1987] uses much the same method as Kozen and Landau[1986] did
for the univariate case to decompose monic multivariate polynomials
in the tame case.  Given a monic $f\in F[\xvec]$ of 
degree $n$ and $r\in\nats$, he
shows how to find a monic $g\in F[x]$ of degree $r$ 
and monic $h\in F[\xvec]$ of degree $r=n/s$ such that $f=g\circ h$.  
The computation requires $O(n^{3\ell})$ field operations.
From the algorithm
it is seen that if such a decomposition exists, it is unique.
Note that monic multivariate polynomials are a very special
case of multivariate polynomials.  Just because they can be
decomposed
does not mean that multivariate polynomials can be decomposed in the
tame case in general (though it is possible a reduction from the
general case exists).

In von zur Gathen[1987], the tame case for the decomposition of 
multivariate polynomials is dealt with completely.  He first defines
the set of polynomials $\bbP^{(\ell)}_F\subseteq F[\xvec]$ which are 
{\it strongly monic} in $x_1$ as follows:
\[
\bbP^{(\ell)}_F=\left\lbrace
f=\sum_{0\leq i\leq n} f_ix_1^i\in F[\xvec]\biggm|
\eqalign{
&n\in\nats,\; f_0,\ldots,f_n\in F[x_2,\ldots,x_\ell]\cr
&f_n=1,\; \deg f=n
}
\right\rbrace.
\]
If $f\in F[\xvec]$ is strongly monic and $f=g\circ h$, then
$f(x_1,0,\ldots,0)=g\circ h(x_1,0,\ldots,0)$.  As we know 
$f(x_1,0,\ldots,0)$ is of degree $n$ ($f$ is strongly monic),
the univariate decomposition of $f(x_1,0,\ldots,0)$ in
$DEC^F_{(r,s)}$ completely determines $g\in F[x]$ in the multivariate
decomposition.  Once $g$ is computed, a linearly convergent Newton iteration
is used to compute $h\in F[\xvec]$ in a number
of field operations polynomial in the input size.
Given $f\in F[\xvec]$ and
$g\in F[x]$, the process of finding $h\in F[\xvec]$ such that
$f=g\circ h$ is a special case of (multivariate) power series reversion.
This is dealt with extensively by Brent and Kung[1977,1978].  They show
the problem is linearly equivalent to power series composition when
$(\partial g/\partial x)\neq 0$, which is true in the tame case.  Furthermore,
they show that multivariate polynomial reversion can be computed in 
$O((n\log n)^{1\over 2} M(n)^\ell)$
field operations.

For an arbitrary $f\in F[\xvec]$ we can use substitutions of the
form $\sigma(x_i)=x_i+\sigma_ix_1$ with $x_i\in F$ for $2\leq i\leq m$
to make $f$ strongly monic.
For such a substitution $\sigma$ we write 
$\sigma f=f(x_1,x_2+\sigma_2x_1,\ldots,x_m+\sigma_mx_1)$.  This
substitution $\sigma$ may be inverted by the substitution
$\sigma^{-1}=(-\sigma_2,-\sigma_3,\ldots,-\sigma_m)$ and 
$\sigma\sigma^{-1}f=f$.  
For a suitably chosen substitution $\sigma$, 
$\fbar=a\sigma f$ is strongly monic in $x_1$
(where $a\in F$ is chosen to make the highest order coefficient of
$x_1$ in $a\sigma f$ one).  If $\fbar=\gbar\circ\hbar$ for
$\gbar,\hbar\in F[x]$ then
$f=(a^{-1}\gbar)\circ (\sigma^{-1}\hbar)$ is a corresponding decomposition
of $f$.

For $f\in F[\xvec]$ of total degree $n$,
von zur Gathen[1987] shows how to choose a substitution $\sigma$
such that $\sigma f$ is strongly monic.  This can
be done in a polynomial number of field operations 
in $m,n$, and $k$, where $k$ is the number of monomials in $f$.
(the sparse representation of $f$ has size $O(km\log n)$).
For $0\leq i\leq n$, let $u_i\in F[x_2,\ldots,x_m]$ be the homogenous
part of degree $n-i$ of the coefficient of $x_1^i$ in $f$.  The homogenous
part of degree $n$ of $f$ is therefore $\sum_{0\leq i\leq n} u_ix^i\neq 0$,
and by the homogeneity of the $u_i$'s, $u=\sum_{0\leq i\leq n}u_i$
is also nonzero, and of degree at most $n$.
Let $K$ be an extension field of $F$ with more than $n$ points.  $K$ can
be chosen as a field of degree $O(\log n)$ over $F$.  Now, 
for a substitution $\sigma=(\sigma_2,\sigma_3,\ldots,\sigma_m)$,
$\deg_{x_1}\sigma f=n$ if and only if 
\[
\eqalign{
\deg f(x_1,\sigma_2x_1,\ldots,\sigma_mx_1)
=\deg[(\sigma )(x_1,0,\ldots,0)]=n.\cr
}
\]
This is true if and only if $u(\sigma_2,\sigma_3,\ldots,\sigma_m)\neq 0$.
To find $\sigma_2,\ldots,\sigma_m$ we proceed in stages $i$ from 2 to $m$.
At stage $i$ we choose $\sigma_i\in K$ such that
\[
u(\sigma_2,\ldots,\sigma_i,x_{i+1},\ldots,x_m)
\] 
is nonzero.  We do
this by considering $v_i=u(\sigma_2,\ldots,\sigma_{i-1},x_i,\ldots,x_m)$
as a polynomial in $K(x_{i+1},\ldots,x_m)[x_i]$ of degree in $x_i$ at most
$n$.  Thus $v_i$ has at most $n$ roots in $K(x_{i+1},\ldots,x_m)$ and
we can find a non-root $\sigma_i\in K$ of $v_i$ with at most $n$ evaluations
of $v_i$ at points in $K$.  Assume $f$ is the sum of at most $k$ monomials.
Then $u$ is also the sum of at most $k$ monomials and $\sigma$ can
be found in $O(kmn\log n)$ field operations over $K$ or 
$O(kmn\log nM(\log n))$ field
operations over $F$.
Decomposing multivariate polynomials is, therefore,
polynomial time (in the input degree and the size of the sparse
representation) reducible to decomposing strongly monic multivariate
polynomials.

\subsection{Multivariate Decomposition \mbox{Using Separated Polynomials}}

Using theorem  1.15 we can generalise the algorithm of Barton and Zippel[1985]
to the multivariate case and obtain a multivariate decomposition algorithm
for any field supporting a factorisation algorithm.
To do this we must show a ``right division''
algorithm for the multivariate case.  Namely, given $f,h\in F[\xvec]$, we
must be able to find a $g\in F[x]$ such that 
$f=g\circ h$ (if such a $g$ exists).
We cannot use the ``Taylor Expansion'' method of the univariate case 
directly.  Instead we use the methods of von zur Gathen[1987] to transform
the problem to one involving strongly monic polynomials.  Another simple
transformation yields a univariate problem such that the solution is 
the same as in the original problem.

\medskip
{\tt\obeylines
MultiRightDivide: $F[\xvec]\times  F[\xvec]\rightarrow F[x]$
~~~~Input:~~-~$f,h\in F[\xvec]$ of total degrees $n$ and $s$ respectively.
~~~~Output:~-~$g\in F[x]$ of degree $r=n/s$ such that $f=g\circ h$
~~~~~~~~~~~~~~(if such a $g$ exists).
~~~~1) Let $K$ be an algebraic extension of $F$ with 
~~~~~~~more than $n$ elements. 
~~~~~~~Let $\sigma=(\sigma_2,\sigma_3,\ldots,\sigma_m)\in K^{m-1}$ 
~~~~~~~be a substitution and $a\in K$ such that
~~~~~~~$\fbar=a\sigma f=f(x_1,x_2+\sigma_2x_1,\ldots,x_m+\sigma_mx_1)\in K[x]$
~~~~~~~is strongly monic (see previous section). 
~~~~2) Let $\hbar=\sigma h$.
~~~~3) Using RightDivide determine if there exists a $\gbar\in K[x]$
~~~~~~~such that $\fbar(x,0,\ldots,0)=\gbar\circ\hbar(x,0,\ldots,0)$ 
~~~~~~~and if so find it.  If no such $\gbar$ exists, quit.
~~~~4) Return $g=a^{-1}\gbar$.
}

\medskip
\noindent 
In the previous section we saw that step 1 can be
performed in \linebreak 
$O(kmn\log nM(\log n))$ field operations over $F$, where $k$
is the number of monomials in $f$.
It follows that $f=g\circ h$ if and only if
$a\sigma f=ag\circ \sigma h$.  Since
$a\sigma f$ is strongly monic, there exists a $g\in F[x]$
such that $a\sigma f=ag\circ \sigma h$ if and
only if $a\sigma f(x,0,\ldots,0)=ag\circ \sigma h(x,0,\ldots,0)$.
Using {\tt RightDivide} we can determine the existence of $\gbar=ag$,
and if it exists find it, in $O(M(n)\log n)$ field operations over $K$
or $O(M(n)\log nM(\log n))$ field operations over $F$.
If $\gbar$ exists we can immediately compute $g$, and the whole computation
requires $O((kmn\log n+M(n)\log n)M(\log n))$ field operations over $F$

The algorithm for multivariate decomposition over any field supporting
a factorisation algorithm proceeds in much the same way as the
Barton and Zippel[1985] algorithm for the univariate case.

\medskip
{\tt \obeylines
MultiSepDecomp : $F[\xvec]\times\nats^2\rightarrow MDEC^F_*$
~~~~Input:~~-~$f\in F[\xvec]$ of degree $n$
~~~~~~~~~~~~-~$(r,s)\in\nats^2$, an ordered factorisation of $n$
~~~~Output:~-~$(g,h)\in (F[x]\times F[\xvec])$ such that $f=g\circ h$
~~~~~~~~~~~~~~if such a decomposition exists

~~~~1) Factor $f(\xvec)-f(0,\ldots,0)=\xvec q_1(\xvec)q_2(\xvec)\cdots q_m(\xvec)$
~~~~~~~where each $q_i\in F[\xvec]$ is irreducible for $1\leq i\leq m$
~~~~2) For each subset $S$ of $\lbrace 1,\ldots,m\rbrace$ 
~~~~~~~~~2.1) Let $h={\displaystyle \xvec\prod_{i\in S}} q_i\in F[\xvec]$.
~~~~~~~~~2.2) If $\deg h=s$, attempt to compute $g\in F[x]$ such that
~~~~~~~~~~~~~~$f=g\circ h$ using MultiRightDivide.  
~~~~~~~~~~~~~~If such a $g$ is found, then goto step 4.
~~~~3) Quit, $f$ has no decomposition in $MDEC^F_{(r,s)}$.
~~~~4) Return $(f,(g,h))\in MDEC^F_{(r,s)}$.
}

\medskip
The number of subsets of $S$
is $2^n$.  This algorithm can, therefore, be completed
with $O(\sf(n)+2^n(kmn\log n+M(n)\log n)M(\log n))$
field operations over $F$.
It does, however,
work in both the tame and wild cases over any field supporting a 
polynomial factorisation algorithm.


\newpage
\section{Additive Polynomials}

\subsection{Definition and Root Structure \mbox{of Additive Polynomials}}

Let $F$ be an arbitrary field of characteristic $p$ greater than zero.
Define a polynomial $f\in F[x]$ to be an {\it additive polynomial}
if, for independent indeterminates $x$ and $y$,
\[
f(x+y)=f(x)+f(y).
\]
The non-zero additive polynomials in $F[x]$ are exactly those of the form
\[
f=\sum_{0\leq i\leq\nu}a_ix^{p^i}
\]
where $\nu\in\nats$, $a_i\in F$ for $1\leq i\leq \nu$,  and $a_\nu\neq 0$.
The integer $\nu$ is called the {\it exponent} of $f$, and we write
``$\expn f =\nu$''.  We denote the set of additive polynomials over 
$F$ as $\ap$.

The additive polynomials have a well understood decomposition
structure which leads to a number of interesting results on decomposition
in the general case.  This structure was first developed in Ore[1933b],
who investigated the vector space structure of the roots of additive
polynomials
(as well as investigating the ring structure of the additive 
polynomials under composition -- see chapter 4).  This work was continued
by Whaples[1954], who examined the Galois groups of additive polynomials
and characterised additive polynomials in terms of these groups.  In Dorey
and Whaples, the Galois group $\varG_f$ of 
$\Phi_f=f-f(x)\in F(f)[y]$ (where $f\in \ap$)
is used (see section 1.B) to show that all normal decompositions of
additive polynomials are decompositions into additive polynomials.  
We use this approach to develop much of the structure of the roots
of additive polynomials in terms of $\varG_f$.  Though the theorems in this
section are for the most part known (with the possible exception of 
theorem 3.2(i)),
the extension of the approach of Dorey and Whaples is of interest.  
For a given additive polynomial $f\in\ap$, it 
serves to illustrate the strong connection between the
separated factors of $\Phi_f$,
the Galois structure of $f$ (which
is the basis for block decompositions), and
the Galois structure of $\Phi_f$.  Not coincidentally, each of these
three approaches leads to at least one algorithm -- the first being the
separated polynomial algorithms of Barton and
Zippel and Alagar and Thanh, 
the second being the block decomposition algorithm of Kozen and
Landau, and the last being a number of  algorithms specifically for additive
polynomials, which are presented in chapter 5.

\medskip
\begin{theorem}\quad
\begin{itemize}
\item[(i)] Let $f\in\ap$  be monic, with
exponent $\nu$ such that $f$ is squarefree ($a_0=f^\prime(0)\neq 0$).  
Let $K$ be the
splitting field of $f$.  Then the roots of $f$ in $K$ form a vector
space $V_f$ over $\ints_p$ of dimension $\nu$.
\item[(ii)] For each finite $\ints_p$-vector space $V\subseteq F$ 
of dimension $\nu$, 
there exists a unique
monic $f\in \ap$ with exponent $\nu$ such that the roots
of $f$ are exactly the elements of $V$.
\end{itemize}
\end{theorem}
\begin{proof}
\begin{itemize}
\item[(i)]
For $\alpha,\beta\in K$ such that $f(\alpha)=f(\beta)=0$, we see that
\[
\eqalign{
f(\alpha+\beta) = & f(\alpha)+f(\beta)=0,&\cr
f(k\alpha)= & kf(\alpha)=0 &\quad\hbox{for any}~k\in \ints_p.\cr
}
\]
Since $f^\prime(0)\neq 0$, the greatest common divisor of $f$ and $f^\prime$
is one, and hence $f$ has no multiple roots.  Therefore $V_f$ has
$p^\nu$ distinct elements and dimension $\nu$.
\item[(ii)]
Let $(\theta_1,\ldots,\theta_\nu)$ be a basis for $V$ in $F$ over $\ints_p$.  
The polynomial
\[
\Psi_1=x^p-\theta_1^{p-1}x\in\ap
\]
has roots $k\theta_1$ for all $k\in \ints_p$.
For $i\geq 2$, define
\[
\Psi_i=(x^p-\Psi_{i-1}(\theta_i)^{p-1}x)\circ\Psi_{i-1}\in\ap.
\]
If $\Psi_{i-1}(\alpha)=0$ for $\alpha\in F$ then $\Psi_i(\alpha)=0$.
Also,
\[
\Psi_i(\theta_i)=\Psi_{i-1}(\theta_i)^p-
\Psi_{i-1}(\theta_i)^{p-1}\Psi_{i-1}(\theta_i)=0.
\]
Since $\Psi_i$ is additive, $\Psi_i$ has roots consisting of all
$\ints_p$-linear combinations of $\lbrace\theta_1,\ldots,\theta_i\rbrace$.
Thus the roots of $\Psi_i$ are exactly the members of the vector
space with basis $(\theta_1,\ldots,\theta_i)$.  Let $f=\Psi_n$, which
is monic and additive.  This $f$ is also unique by virtue of being
a monic interpolant of degree $p^\nu$ of $p^\nu$ distinct 
points.  Note also that
$f^\prime(0)\neq 0$ as $f$ has $p^\nu$ distinct roots in $f$ and degree
$p^\nu$.
\end{itemize}
\end{proof}

\noindent 
Call the $\ints_p$-vector space $V$ of 
roots of an additive polynomial $f\in F[x]$
the {\it kernel} of $f$. Say an additive polynomial is {\it simple} if
it is monic and $f^\prime(0)\neq 0$.  
In this section we will deal almost exclusively with simple additive
polynomials. Non-simple monic additive polynomials are just simple polynomials
composed on the right with $x^{p^\ell}$ for some $\ell>0$.  Assume 
$f\in\ap$ is monic and 
$f=g\circ x^{p^\ell}\in \ap$ where $g\in \ap$ is  simple and
\[
g=\sum_{0\leq i\leq\sigma} b_ix^{p^i}
\]
with $\sigma\in\nats$, $b_i\in F$, and $b_\sigma\neq 0$.
Then 
\[
\eqalign{
f=&\sum_{0\leq i\leq\sigma}b_ix^{p^{i+\ell}}\cr
 =& x^{p^\ell}\circ\sum_{0\leq i\leq\sigma}(b_i)^{1\over{p^\ell}}x^{p^i}\cr
 =&x^{p^\ell}\circ\gbar\cr
}
\]
where
\[
\gbar=\sum_{0\leq i\leq\sigma}(b_i)^{1\over{p^\ell}}x^{p^i}\in \bbA_K
\]
and $K$ is an algebraic extension of $F$.  So $f$ has
a kernel of dimension $\sigma$, namely the kernel of $\gbar$.  If $F$
is perfect (and hence closed under $p^{th}$ roots) $\gbar$ will be in $F[x]$
as well.

Let $f\in \ap$ be   simple with exponent $\nu$,  
splitting field $K$ and kernel $V_f\subseteq K$.
The structure of the kernel of $f$ and that of the fields between
$K(f)$ and $K(x)$ (and hence the structure of the decompositions of
$f$ over its splitting field) are closely linked.
Let $\Phi_f=f(y)-f\in F(f)[y]\subseteq F(x)[y]$ 
with splitting field $\Omega\supseteq F(f)$ and Galois
group $\varG_f=Gal(\Omega/F(f))$ as in theorem 1.5.  Because 
${\partial\over {\partial y}}\Phi_f
={\partial\over {\partial y}}f(y)\neq 0$, $\Omega$ is a separable
field extension of~$F(f)$.

\medskip
\begin{theorem}\quad\\[-15pt]
\begin{itemize}
\item[(i)] $K(x)$ is the splitting field of $\Phi_f$,\\[-15pt]
\item[(ii)] $\varG_f$ is the group $\lbrace x\mapsto x+\alpha\;|\;\alpha\in K,~
f(\alpha)=0 \rbrace $ under composition, and\\[-15pt]
\item[(iii)] $V_f\cong \varG_f$.
\end{itemize}
\end{theorem}

\pagebreak
\begin{proof}
\begin{itemize}
\item[(i)]
For $\alpha\in V_f$, 
\[
\eqalign{
\Phi_f(x+\alpha)= & f(x+\alpha)-f(x)\cr
= & f(x)+f(\alpha)-f(x)\cr
= & 0.\cr
}
\]
Since $\Phi_f$ has degree $p^\nu$ over $F(f)$, $x+V_f$ is the complete
set of roots of $\Phi_f$.  We know that $x\in x+V_f$ and $K$ is the smallest
extension field of $F$ containing $V_f$, so $\Omega=K(x)$.
\item[(ii)] From (i), all roots of $\Phi_f$ are of the form $x+\alpha$ where
$\alpha$ is a root of $f$ in $K$.  Therefore, $\varG_f$ contains the
$p^\nu$ automorphisms $\lbrace x\mapsto x+\alpha\;|\;\alpha\in K~,~
f(\alpha)=0 \rbrace$. Since $[F(x):F(f)]=p^\nu$, 
this is the entire Galois group.
\item[(iii)]
From (ii), $\varG_f$ is isomorphic to a set of monic linear elements
in $K[x]$ under
composition.  Trivially this is isomorphic to the group of constant
terms of these elements under addition.  These constant terms are all
the roots of $f$ in $K$, so $\varG_f\cong V_f$.
\QED
\end{itemize}
\end{proof}

\medskip
\begin{theorem}
Let $f\in\ap$ be simple of exponent 
$\nu$. Let $g,h\in F[x]$ be of degrees  $r=p^\rho$ and $s=p^\sigma$ 
respectively such that $(f,(g,h))\in DEC^F_{(r,s)}$.
Then 
\begin{itemize}
\item[(i)] $g$ and $h$ are additive and simple, and
\item[(ii)] $h$ has kernel $V_h\cong\varG_h$, where $\varG_h\subseteq\varG_f$ 
is the subgroup fixing $F(h)\subseteq F(x)$ pointwise.
\end{itemize}
\end{theorem}
\begin{proof}
By theorem 1.5, the automorphisms in $\varG_f$ fixing $F(h)$ form
a group $\varG_h$ such that $\varG_x\subseteq\varG_h\subseteq\varG_f$, and
the index of $\varG_x$ in $\varG_h$ is $p^\sigma$. 
From theorem 3.2(i), $K(x)$ is the splitting field of $\Phi_f$,
so $\varG_x$ is the identity group, and the cardinality of 
$\varG_h$ is $p^\sigma$.
From the isomorphism between $\varG_f$ and $V_f$, there 
is a subspace $W$ of $V_f$ of dimension $\sigma$ 
corresponding to the subgroup $\varG_h$.  
Let $\hbar\in K[x]$ be the simple additive 
polynomial with kernel
$W$.  For all $\alpha\in W$, $\hbar(x+\alpha)=\hbar(x)+\hbar(\alpha)=\hbar(x)$.
Thus $F(\hbar)$ is fixed by $\varG_h$.  By theorem 3.1, $\hbar$ is
unique, so $\hbar=h$.  
Now, for algebraically independent indeterminates $x$ and $y$,
$f=g(h(x+y))=g(h(x)+h(y))$ and since
$f$ is additive $f=g(h(x))+g(h(y))$.  Thus $g$ is monic and additive.
If $g$ is not simple then $g=\gbar\circ x^{p^\ell}$ for some
simple additive polynomial $\gbar\in K[x]$ and $\ell>0$.  But then
$f=\gbar\circ x^{p^\ell}\circ h=\gbar\circ h^{p^l}$ which is not simple.
So $g$ is simple as well.
\QED
\end{proof}

\subsection{Rationality and the Kernel}

If $f\in \ap$ is simple with kernel $V_f\subseteq K$
and splitting field $K$, a
subspace $V_h\subseteq V_f$ is said to be {\it rational} if the
simple polynomial $h\in\bbA_K$ corresponding to $V_h$ is in 
$\ap$.  We would also like to formulate rationality in terms of the
structure of the kernel.

\medskip
\begin{theorem}
A subspace $V_h$ of $V_f$ is rational if and only if 
$V_h$ is invariant (as a set) under $G_f=\gal(K/F)$.
\end{theorem}
\begin{proof}
Assume $h\in\ap$.  The coefficients of $h$ are the values of the elementary
symmetric functions of the roots of $h$ in $K$.
The automorphisms in $G_f$ leave these coefficients fixed, and must
therefore leave $V_h$ invariant (as a set).

Conversely, if $V_h$ is invariant under $G_f$ then the values of the 
elementary symmetric functions of the elements of $V_h$ are fixed
under $G_f$, and so are in $F$.  These are exactly the coefficients of $h$
and so $h\in\ap$.
\QED
\end{proof}

\noindent
When dealing with a finite field $F$ 
a somewhat stronger structure can be shown.  

\medskip
\begin{theorem}
Let $F=GF(q)$ where $q=p^e\in\nats$ for some $p,e\in\nats$ with $p$ prime. 
Let $K$ be an algebraic extension of $F$ and $f\in\bbA_K$ of exponent $\nu$
with kernel $V_f$ and splitting field $L$.  
Then $f\in\ap$ if and only if $V_f^q=V_f$.
\end{theorem}
\begin{proof}
If $f\in \ap$ and $\alpha$
is a root of $f$ in $L$, then so is $\alpha^q$.
This follows since, if $g\in F[x]$ is the minimal polynomial of $\alpha$, $g|f$
and $g(\alpha^q)=0$ 
($\alpha$ and $\alpha^q$ are conjugates since $F$ is finite).
Thus $V_f^q\subseteq V_f$.  Since $x\rightarrow x^q$
is an automorphism of $L$ over $F$, $V_f^q=V_f$.

If $V_f^q=V_f$ then we must show that $f\in F[x]$.  The group $H$ 
of automorphisms
of $L$ over $F$ is the group generated by the automorphism
$x\rightarrow x^q$.
Thus $V_f$ is invariant (as a set) under $H$.  As the coefficients
of $f$ are symmetric functions of the elements of $V_f$, they
are fixed by $H$, and therefore must be in $F$.  Hence $f\in\ap$.
\QED
\end{proof}

The preceding theorem gives the following alternative formulation of
the bidecomposition problem for additive polynomials:

{\narrower\noindent
Let $F=GF(q)$ where $q,p,e\in\nats$, $p$ is prime, and $q=p^e$.  Let $K$
be an algebraic extension of $F$ and let $V\subseteq K$ be a $\ints_p$
vector space of dimension $\nu$ over $\ints_p$ 
such that $V^q=V$.  For a given $\sigma$
with $1\leq\sigma\leq\nu$, 
determine if there exists a $\sigma$ dimensional subspace $W$ of $V$
such that $W^q=W$, and if so, give a basis for some predetermined number
of them. \par
}

\noindent
Since $V^q=V$, all the elements of $V$ can be specified as the roots of a single
additive polynomial $f\in\ap$ 
of exponent $\nu$.  The found subspace $W$ (if it exists)
will be the kernel of a right composition factor $h\in\ap$ 
of $f$ of exponent $\sigma$.

\subsection{Rational Decompositions of Additive Polynomials}

We can now talk about decompositions of simple additive polynomials
in general and their relationship to their kernels.
For any $n,m\in\nats$, let 
\[
\eqalign{
\wp = & (r_m,r_{m-1},\ldots,r_1)\cr
    = & (p^{\rho_m},p^{\rho_{m-1}},\ldots,p^{\rho_1})\cr
}
\]
be an ordered factorisation of $n$. Define
\[
\apdec =
\left\lbrace
(f,(f_m,\ldots,f_1)) \in \ap \times (\ap)^m \Bigm |
\eqalign{ 
& f=f_m\circ\cdots\circ f_1 \cr
& \hbox{and}~\deg f_i=r_i=p^{\rho_i} \cr
}
\right\rbrace.
\]
Similarly, for simple additive polynomials, define
\[
\sapdec =
\left\lbrace
(f,(f_m,\ldots,f_1)) \in \ap \times (\ap)^m \Biggm |
\eqalign{ 
& f ~\hbox{simple},~f=f_m\circ\cdots\circ f_1, \cr
& \hbox{and}~\deg f_i=r_i=p^{\rho_i} \cr
}
\right\rbrace.
\]
Obviously $\sapdec\subseteq\apdec\subseteq\dec$.

\noindent
Let $d_i={\displaystyle \prod_{\scriptscriptstyle 1\leq   j\leq i}}r_j$ 
and $\bbV$ be the set of all finite $\ints_p$
-vector spaces in $F$.  Define
\[
\flag =
\left\lbrace
(f,(V_m,\ldots,V_1)) \in \ap \times \bbV^m \Biggm |
\eqalign{ 
& f ~\hbox{simple},~V_m ~\hbox{is the kernel of}~ f,  \cr
& V_m\supseteq V_{m-1}\supseteq\cdots\supseteq V_1,\cr
& \dim V_i=d_i, \,\hbox{for}\;1\leq i\leq m\cr
}
\right\rbrace.
\]
The sequence $V_m\supseteq V_{m-1}\supseteq\cdots\supseteq V_1$ 
is called a {\it flag} of vector spaces associated with $f$.

Let $f\in\ap$ be simple.  For any $(f,(f_m,\ldots,f_1))\in\sapdec$ let
$h_i=f_i\circ f_{i-1}\circ\cdots\circ f_1\in\ap$ and let $V_i$ be the
kernel of $h_i$.  Then by theorems 3.1 and 3.3 
$(f,(V_m,\ldots,V_1))\in\flag$.  Therefore, we define the map
$\G[SA,FL]: \sapdec\rightarrow\flag$ by
$(f,(f_m,\ldots,f_1)) \mapsto (f,(V_m,\ldots,V_1))$.

\medskip
\begin{theorem}
$\G[SA,FL]$ is a bijection between $\sapdec$ and $\flag$.
\proof
$\G[SA,FL]$ is injective since distinct additive polynomials have
distinct kernels.
If $(f,(V_m,V_{m-1},\ldots,V_1))\in\flag$, then by theorem 3.3 the
additive polynomial $h_i$ with kernel $V_i$ has a factor $h_{i-1}$
with kernel $V_{i-1}$.  Thus $h_i=f_i\circ h_{i-1}$ for some
unique $f_i\in\ap$ of degree $d_i/d_{i-1}=r_i$.  Thus
$(f,(f_m,\ldots,f_1))\in\sapdec$ and this map from $\flag$ to
$\sapdec$ is injective and is in fact the inverse of $\G[SA,FL]$.
Thus $\G[SA,FL]$ is a bijection.
\QED
\end{theorem}

\subsection{The Number of Bidecompositions of a Polynomial}
  
We will now compute the exact number of  bidecompositions of
a simple additive polynomial $f\in\ap$ into two simple additive polynomials
over its splitting field $K$.  Assume $f$ has exponent $\nu$.
The number of simple additive right
composition factors of $f$ in $K[x]$ of exponent $\sigma$ is exactly the 
number of $\sigma$-dimensional subspaces of the kernel of $f$.  
This is calculated in the following well-known lemma.

\medskip
\begin{lemma}
The number of $\sigma$-dimensional subspaces of a $\nu$-dimensional
vector space $V$ over $\ints_p$ is
\[
\calS^\nu_\sigma=
{
{\displaystyle \prod_{0\leq i<\sigma}(p^\nu-p^i)} \over
{\displaystyle \prod_{0\leq i<\sigma}(p^\sigma-p^i)}}.
\]
\end{lemma}
\begin{proof}
The number of linearly independent $\sigma$-tuples of vectors
in $V$ is 
\[\prod_{0\leq i<\sigma}(p^\nu-p^i).\]  
This is the
number of all bases for all vector spaces of dimension $\sigma$.
Each $\sigma$ dimensional vector space has 
$\prod_{0\leq i<\sigma}(p^\sigma-p^i)$ bases.
The lemma follows.
\QED
\end{proof}

\noindent
The desired cardinality theorem now follows immediately.

\medskip
\begin{theorem}
Let  $f\in\ap$ be simple of exponent $\nu$
with splitting field $K$.
The number
of bidecompositions of $f$ in 
$APDEC^K_{(p^{\nu-\sigma},p^\sigma)}$
is $\calS^\nu_\sigma$.
\end{theorem}

This theorem gives a super-polynomial lower bound for the number of 
decompositions of an arbitrary polynomial over an algebraic extension field.

\medskip
\begin{theorem}
For any even $\nu\in\nats$,
there exist monic polynomials $f\in F[x]$ of degree $n=p^\nu$ 
with splitting
field $K$
such that there are at least
$n^{\lambda\log n}$ decompositions of $f$ in 
$DEC_{(\sqrt{n},\sqrt{n})}^K$ where $\lambda=(6\log p)^{-1}$.
\end{theorem}
\begin{proof}
Assume $n=p^\nu$ where $\nu$ is even and 
let $f$ be a simple additive polynomial of
exponent $\nu$.  Then $f$ has $\calS^\nu_{\nu\over 2}$
decompositions in $DEC_{(\sqrt{n},\sqrt{n})}^K$.
\[
\eqalign{
\calS^\nu_{\nu\over 2}= &
     {\displaystyle 
          \prod_{0\leq i<{\nu\over 2}}(p^\nu-p^i)} \over
     {\displaystyle 
          \prod_{0\leq i<{\nu\over 2}}(p^{\nu\over 2}-p^i)}\cr
\geq & {(p^{\nu-1})^{\nu\over 2}}\over {(p^{\nu\over 2})^{\nu\over 2}}\cr
= & p^{{{\nu^2}\over 4}-{\nu\over 2}}\cr
\geq & (p^\nu)^{\nu\over 6}\cr
= & n^{\lambda\log n},
}
\]
where $\lambda=(6\log p)^{-1}$.
\QED
\end{proof}

\subsection{Complete Decompositions of \mbox{Additive Polynomials}}

Let $n\in\nats$ and let $\wp$ be a length $m$ ordered factorisation of
$n$.
Complete decompositions of additive polynomials with ordered
factorisation $\wp$ 
will be considered
in a straightforward manner.  Define the set $\capdec\subseteq\cdec$ to 
be the set
of complete rational decompositions of additive polynomials
with ordered factorisation $\wp$.
By theorem 3.3 these will be decompositions into additive polynomials.
Similarly, define the set $\csapdec\subseteq\capdec\subseteq\cdec$ 
to be the set of complete
rational decompositions of simple additive polynomials with ordered
factorisation $\wp$.
The image of $\csapdec$
in $\flag$ will be called $\cflag$ and it too corresponds to the set of 
rational
complete decompositions of simple additive polynomials.  
Members of $\cflag$ can also be characterised
as those members of $\flag$ whose flags are maximal.

\subsection{The Number of Complete Rational Normal\\ Decompositions}

In much the same way as we calculated the number of right composition
factors of given exponent of a polynomial in theorem 3.7, 
we calculate the number of complete
decompositions of a polynomial over an extension field.
Let $f\in\ap$ be simple with exponent $\nu$
and kernel $V_f$.  The number of complete decompositions $\calF^\nu$ 
of $f$ over
its splitting field $K$ is equal to the number of maximal flags
in $V_f$, and turns out to be dependent only on $\nu$, not on $f$.
As all subspaces of $V_f$ are rational in $K$, each maximal
flag will have $\nu$ subspaces and will have the form
\[
V_f=V_\nu\supsetneq V_{\nu-1}\supsetneq\cdots\supsetneq V_1
\]
where $\dim V_i=i$. The corresponding complete decompositions will be
into exponent one, {\it p-linear}, composition factors.

\medskip
\begin{lemma}
The number of $\sigma$-dimensional subspaces of a $\nu$-dimensional
vector space $V$ over $\ints_p$ which contain a given $(\sigma-1)$-dimensional
vector space $W$ is
\[
\calT^\nu_\sigma
={{p^\nu-p^{\sigma-1}}\over{p^\sigma-p^{\sigma-1}}}
\]
\end{lemma}
\begin{proof}
There are $p^\nu-p^{\sigma-1}$ vectors of $V$ which are
linearly independent with $W$.  
A given $\sigma$-dimensional vector space containing $W$ is generated by 
$W$ plus any one of $p^\sigma-p^{\sigma-1}$ vectors.
The lemma follows.
\QED
\end{proof}

\noindent
The following lemma gives bounds for $\calF^\nu$, and hence for the
number of complete decompositions of $f$ over $K$.

\medskip
\begin{lemma}
Let $f\in\ap$ be simple of exponent $\nu$ with splitting field $K$.
The maximum number $\calF^\nu$ of distinct 
complete normal decompositions  of $f$ over
$K$ is bounded by
\[
p^{{\nu^2}\over 2}\leq\calF^\nu\leq p^{{{\nu^2}\over 2}+{{3\nu}\over 2}}.
\]
\end{lemma}
\begin{proof}
The fact that there are $\calF^\nu$ distinct complete normal
decompositions of $f$ over $K$ follows from the preceding
discussion. We get the bounds as follows:

\[
\eqalign{
\calF^\nu = &\prod_{1\leq i\leq \nu} \calT^\nu_i\cr
=& \prod_{1\leq i\leq \nu} {{p^\nu-p^{i-1}}\over{p^i-p^{i-1}}}\cr
\leq & \prod_{1\leq i\leq\nu} p^{\nu-i+1}\cr
\leq & p^{{{\nu^2}\over 2}+{{3\nu}\over 2}},\cr
}
\]

\[
\eqalign{
\calF^\nu = &\prod_{1\leq i\leq \nu} \calT^\nu_i\cr
=& \prod_{1\leq i\leq \nu} {{p^\nu-p^{i-1}}\over{p^i-p^{i-1}}}\cr
\geq & \prod_{1\leq i\leq \nu} p^{\nu-i}\cr
\geq & p^{{\nu^2}\over 2}. \hspace*{50pt}\mbox{\QED} \cr
}
\]
\end{proof}

$\calF^\nu$ is at least $n^{\mu\log n}$ where $\mu=(2\log p)^{-1}$ and
is super-polynomial in the degree $n$ of $f$, and so there is a 
super-polynomial number of different complete normal decompositions
of $f$ over $K$.  However, this does not guarantee that these decompositions
are inequivalent in the sense that Ritt[1922] considered for the 
characteristic zero case.
We now consider this question in the wild case.

As we saw in section 1.F, in the tame case there are two types of ambiguous
decompositions.  Recall that if $u\in F[x]$ and $m,r\in\nats$, then
$(x^m\cdot u^r)\circ x^r=x^r\circ (x^m\cdot u(x^r))$, an
exponential ambiguity.  If $x^m\cdot u^r$
and $x^mu(x^r)$ are indecomposable and additive, then since they are 
necessarily squarefree, $r$ and $m$ are at most one. 
In the case of additive polynomials
therefore, exponential ambiguity is simply identity.  The second kind of
ambiguity in the tame case are trigonometric ambiguities -- ambiguities
arising from the commutative properties of the Chebyshev polynomials under
composition.
As we saw in theorem 1.13,
the Chebyshev polynomial $T_{p^i}=x^{p^i}$, for $i\in\nats$, in fields of
characteristic $p$ greater than two.
In fields of characteristic $p=2$, $T_{p^i}=1$ if $i$ is even and
$x$ if $i$ is odd.  
Instead of restricting ourselves  to equivalence under these two
types of ambiguities,
we define the more general concept of a {\it permutation}
ambiguity. Two complete normal decompositions are {\it permutation 
equivalent} if the composition factors of one are a permutation
of the composition factors of the other.  
Trigonometric ambiguities are certainly encompassed in this definition.
We now proceed to
construct a class of polynomials which have a super-polynomial
number (in their degrees) of permutation inequivalent decompositions
over their splitting fields.

\medskip
\begin{theorem}
Let $p\in\nats$ be prime, $\nu\in\nats$ and $F=GF(p^\nu)$.
Also, let $K$ be an algebraic extension of $F$ of degree $p^\nu$ over $F$.
Then there exist simple additive polynomials $\fhat\in K[x]$ of 
exponent $\nu$ which have $\calF^\nu\geq p^{{\nu^2}\over 2}$
pairwise permutation inequivalent complete normal decompositions in 
$cSAPDEC^K_*$.  
\end{theorem}
\begin{proof}
Let $(\theta_1,\ldots,\theta_\nu)$ be a basis for an algebraic extension 
$F$ of $\ints_p$ of degree $\nu$.  Also, let $f=x^{p^\nu}-x\in F[x]$,  
and $\e\in\Fbar$ 
(where $\Fbar$ is an algebraic closure of $F$) be algebraic
of degree $p^\nu$ over $F$.
Consider the polynomial $\fhat$ with roots consisting of all elements
of $K=F[\e]$ of the form $\e a$ for $a\in F$.  The roots of $\fhat$ have
a basis $(\e\theta_1,\ldots,\e\theta_\nu)$ over $\ints_p$.
As in the construction of theorem 3.1, we now describe complete decompositions
of $f$ and $\fhat$ with respect to the bases $(\theta_1,\ldots,\theta_\nu)$
and $(\e\theta_1,\ldots,\e\theta_\nu)$. Let
\[
\eqalign{
\Psi_1 =    & x^p-\theta_1^{p-1}x\in F[x],~~~~~~~\hbox{and}\cr
\Psihat_1 = & x^p-(\e\theta_1)^{p-1}x\in K[x].\cr
}
\]
For $i>1$, define
\[
\eqalign{
\Psi_i =    & 
(x^p-\Psi_{i-1}(\theta_i)^{p-1}x)\circ\Psi_{i-1}\in F[x],~~~~~~~\hbox{and}\cr
\Psihat_i = & 
(x^p-\Psihat_{i-1}(\e\theta_i)^{p-1}x)\circ\Psihat_{i-1}\in K[x].\cr
}
\]
Then
\[
\eqalign{
f= & \Psi_\nu = & (x^p-\Psi_{\nu-1}(\theta_\nu)^{p-1}x)
     \circ\cdots\circ(x^p-\theta_1^{p-1}x),~~~~~~~\hbox{and}\cr
\fhat= & \Psihat_\nu = & 
     (x^p-\Psihat_{\nu-1}(\e\theta_\nu)^{p-1}x)
     \circ\cdots\circ(x^p-(\e\theta_1)^{p-1}x).\cr
}
\]
Since, for $1\leq i\leq\nu$,
\[
\Psihat_{i}=\prod_{(a_1,\ldots,a_{i})\in\ints_p^{i}} 
     (x-\sum_{1\leq j\leq i}a_j\theta_j\e)~,
\]
we find that
\[
\eqalign{
\Psihat_{i-1}(\e\theta_i)= &
     \prod_{(a_1,\ldots,a_{i-1})\in\ints_p^{i-1}}
     (\e\theta_i-\sum_{1\leq j\leq i-1}a_j\theta_j\e)\cr
=&   \e^{p^{i-1}}\prod_{(a_1,\ldots,a_{i-1})\in\ints_p^{i-1}}
     (\theta_i-\sum_{1\leq j\leq i-1}a_j\theta_j)\cr
=&   \e^{p^{i-1}}\Psi_{i-1}(\theta_i).\cr
}
\]
Thus, in any decomposition of $\fhat$ into $p$-linear components
in $K[x]$, for $1\leq i\leq\nu$, the $i^{th}$ composition factor has the form
\[
x^p-a\e^{p^{i-1}(p-1)}x
\]
for some $a\in F$.  If any non-identity permutation of a decomposition
was also a decomposition of $\fhat$, then 
\[ 
a\e^{p^{i-1}(p-1)}=b\e^{p^{j-1}(p-1)}
\] 
for some $1\leq i<j\leq\nu$ and $a,b\in F$.  But then
$\e$ would satisfy a polynomial in $F[x]$ of degree less than $p^\nu$,
giving a contradiction.  It follows that
for the class of polynomials just constructed
there are $\calF^\nu$ permutation inequivalent, complete normal
decompositions.
\QED
\end{proof}

\noindent
The above theorem gives a super-polynomial lower bound on the number of
permutation inequivalent, complete normal decompositions possible 
for an arbitrary polynomial.

\medskip
\begin{theorem}
Let $p$ be a  prime number, $\nu\in\nats$, and $n=p^\nu$.
There exist fields $K$  of algebraic degree at most $n\log n$ over
$\ints_p$, and monic polynomials of degree $n$ in $K[x]$ which have
$n^{\mu\log n}$ decompositions in $cDEC^K_*$ 
which are inequivalent up to exponential and permutation ambiguities
(where $\mu=(2\log p)^{-1}$).
\end{theorem}
\begin{proof}
Let $f\in\ap$ be  simple  of degree $n=p^\nu$
as constructed in theorem 3.11.  
By lemma 3.11 we know $f$ has at least
\[
p^{{v^2}\over 2} = n^{\mu\log n}
\]
complete normal decompositions in $K[x]$ (where $\mu=(2\log p)^{-1}$), 
and these decompositions are
inequivalent up to exponential and permutation ambiguities.  The field
$K$ in theorem 3.11 has degree $\nu p^\nu=O(n\log n)$ over $\ints_p$.~\QED
\end{proof}

Additive polynomials are certainly not the only class of polynomials which
potentially have a super-polynomial number (in their degrees) of inequivalent 
decompositions.  For example, let ${\cal Q}$ be a set of additive polynomials
which have a super-polynomial number of inequivalent decompositions in their
degrees (such as that defined in theorem 3.11).  Define
a new set of polynomials
\[
\calD=\lbrace g\circ f\circ g | f\in {\cal Q},
\; g\in F[x],~\deg f=\deg g=n\rbrace.
\]
Each $f\in\calD$ has a super-polynomial number of decompositions in its degree
and yet $\calD$ is not a set of additive polynomials.


\newpage
\section{The Ring of Additive Polynomials}

\subsection{Basic Ring Structure}

Ore[1933a] considers rings of polynomials $R_F\subseteq F[x]$ under the
usual polynomial addition ($+$), and a (possibly non-commutative) 
multiplication 
($\times$).  The only further assumption he makes is the existence of a
degree function $\delta:R_F\setminus \lbrace 0\rbrace\rightarrow\nats$ 
such that if $f,g\in R_F$
with $\delta(f)=r$ and $\delta(g)=s$, then $\delta(f\times g)=r+s$.
In Ore[1933b] he applies this theory to the ring $\ap$ of additive
polynomials with composition as ring multiplication and exponent
as the degree function.  In this chapter, in sections A-D we present a summary
of the theory of Ore as applied to additive polynomials.
In section E we investigate the uniqueness properties of decompositions,
and some properties of the indecomposable composition factors.
We also strengthen a theorem of Ore[1933a] as applied to additive
polynomials.  In section F we use the relationships between decompositions
developed in the previous sections to give an upper bound on the
number of complete rational decompositions of an arbitrary additive
polynomial.
In chapter 5 we will then
use the theory developed here 
to construct decomposition algorithms for additive
polynomials.

Recall from chapter 3 that if $F$ is a field of characteristic $p$
then $f\in F[x]$ is additive if $f(x+y)=f(x)+f(y)$
for independent indeterminates $x$ and $y$.  We denote the
set of all additive polynomials over $F$ as $\ap\subseteq F[x]$ 
and for $f\in\ap$,
\[
f=\sum_{0\leq i\leq\nu} a_ix^{p^i}
\]
with $a_i\in F$ for $0\leq i\leq\nu$ and $a_\nu\neq 0$.  The integer
$\nu\geq 0$ is called the exponent of $f$ and we write $\expn f=\nu$.
It is easy to see that $\ap$ is a ring without zero divisors.  
We will also show it has
a right division algorithm (ie. if $f,g\in\ap$, with 
$g\neq 0$, then there exists
$Q,R\in\ap$ such that $f=Q\circ g+R$ and $\expn R<\expn g$), and is
therefore a left-Euclidean ring (the terminology is derived from the
fact that the right division algorithm makes it a principal
left ideal ring).  Let $f,g\in\ap$ 
with $g\neq 0$ and $\expn f=\nu$, and $\expn g=\rho$.  Assume also that
$f$ and $g$ have leading (high order) coefficients $a\in F$ and $b\in F$
respectively.  If $\nu<\rho$ then division is trivial.  If $\nu\geq\rho$ then
with $f^{(\nu)}=f$, define
\[
h^{(\nu)}=ab^{-p^{\nu-\rho}}x^{p^{\nu-\rho}}\in\ap
\]
and
\[
f^{(\nu-1)}=f^{(\nu)}-h^{(\nu)}\circ g\in\ap.
\]
Then $f^{(\nu)}=h^{(\nu)}\circ g+f^{(\nu-1)}$ and $f^{(\nu-1)}$ has exponent
less than that of $f^{(\nu)}$.  Iterating this process we get
\[
f=(h^{(\nu)}+h^{(\nu-1)}+\cdots+h^{(\rho)})\circ g+f^{(\rho-1)}
\]
and the exponent of $f^{(\rho-1)}$ is less than the exponent of $g$.  This
gives a right hand division algorithm for $\ap$.

Let $f,g,h\in\ap$.  If $f=g\circ h$, then we write $h\apdivides f$,
meaning $h$ is a right composition factor of $f$.  We will write
$g=f\apdiv h$ meaning $g$ is the compositional quotient after dividing
$f$ by $h$ on the right (provided $h$ does divide $f$ on the right).
This quotient is unique because of the existence of the division
algorithm shown above (or by lemma 1.1).  
Finally, if $h\apdivides f-g$, then we write
$f\apequiv g\bmod h$.
As an example, with $F=\ints_3$, let
\[
\eqalign{
f = & x^{27}+2x^9+x^3+2x,\cr
g = & x^9+x^3+x.\cr
}
\]
Then
\[
\eqalign{
f = & x^3\circ g + x^9+2x\cr
  = & x^3\circ g + x\circ g + 2x^3+x\cr
  = & (x^3+x)\circ g + (2x^3+x).\cr
}
\]

\subsection{The Euclidean Scheme}

From the existence of a right division algorithm for $\ap$ follows
the existence of a right Euclidean algorithm.  Given $f_1,f_2\in\ap$, we
proceed with the Euclidean scheme in the usual fashion 
(see van der Waerden [1970] pp. 55).  Assume
$\expn f_1 \geq \expn f_2$.  At each stage $i>2$, let $f_i$ be the remainder
of $f_{i-2}$ divided on the right by $f_{i-1}$.  We get the following
sequence:
\[
\eqalign{
f_1 = & Q_1\circ f_2+f_3,\cr
f_2 = & Q_2\circ f_3+f_4,\cr
f_3 = & Q_2\circ f_4+f_5,\cr
\vdots~ & \cr
f_{n-2} = & Q_{n-2}\circ f_{n-1} + f_n,\cr
f_{n-1} = & Q_{n-1}\circ f_n,\cr
}
\]
where $Q_i,f_i\in\ap$ and $\expn f_i<\expn f_{i-1}$.  The number of steps
$n$ is at most the exponent of $f_2$.  The polynomial $af_n\in\ap$, where
$a\in F$ is such that $af_n$ is monic, is the greatest common 
(right compositional) divisor or {\it meet} of $f_1$ and $f_2$.  We denote
the meet
$f_1 \mt f_2$.  As an example, assume as before that $F=\ints_3$
and
\[
\eqalign{
f_1 = &  x^{27}+2x^9+x^3+2x,\cr
f_2 = &  x^9+x^3+x.\cr
}
\]
Following the Euclidean scheme,
\[
\eqalign{
f_1 = & (x^3+x)\circ f_2+(2x^3+x),\cr
f_2 = & (2x^3+x)\circ (2x^3+x),\cr
f_3 = & 2x^3+x.\cr
}
\]
Normalising to make the meet monic,
\[
\eqalign{
f_1 \mt f_2 = & 2^{-1}(2x^3+x)\cr
            = & x^3+2x.\cr
}
\]

The existence of a Euclidean algorithm means $\ap$ is a principal
left ideal ring.  
Let $f_1$ and $f_2$ be two additive polynomials,
and let $(f_1)$ and $(f_2)$ be the left ideals generated by them.  
The ideal $D=(f_1)+(f_2)$ 
consists of all sums of left multiples of
$f_1$ with left multiples of $f_2$.  
Because $\ap$ is principal, $D=(u)$ for some
unique monic $u\in\ap$ and this $u$ is the meet of $f_1$ and $f_2$.
The set $L=(f_1)\cap (f_2)$ is also an ideal and consists of all common
left multiples of $f_1$ and $f_2$.  
Assume $f_1f_2\neq 0$. We must now show that $L\neq (0)$.
Let $D=f_1 \mt f_2$.  From the extended Euclidean scheme we know that there
exist $A_1,A_2\in\ap$ such that $A_1\circ f_1+A_2\circ f_2=D$.  
If $f_2=R_2\circ D$ for $R_2\in\ap$, then
$R_2\circ A_1\circ f_1+R_2\circ A_2\circ f_2=f_2$ and 
$R_2\circ A_1\circ f_1=(x-R_2\circ A_2)\circ f_2$.  Thus $f_1$ and $f_2$ are
both right factors of $R_2\circ A_1\circ f_1$, and since this is nonzero,  
$L\neq (0)$.  The ring $\ap$ is a principal left ideal ring, so $L=(h)$
for some unique  monic $h\in F[x]$. 
This is the common left multiple of $f_1$ and $f_2$ of least exponent,
which we will call the {\it join} of $f_1$ and $f_2$.
We denote the join of $f_1$ and $f_2$ by $f_1 \jn f_2$.
Some properties of the join are summarised in the following lemma.

\pagebreak
\begin{lemma}
Let $f,g,h\in\ap$.
\begin{list}{}{\itemsep=0pt\parsep=0pt}
\item[(i)] $f\jn g=g\jn f$,
\item[(ii)] $f \jn (g\jn h) = (f\jn g)\jn h$ \qquad
(we will often write $f\jn g\jn h$),
\item[(iii)] $g\jn (f\circ g)=f\circ g$,
\item[(iv)] $(g\circ h)\jn (f\circ h)= (g\jn f)\circ h$,
\item[(v)] if $g\apdivides f$ and $h\apdivides f$ then $g\jn h\apdivides f$.
\end{list}
\end{lemma}
\begin{proof}
Let $(f)$, $(g)$, $(h)$ be the left ideals generated by $f$, $g$, and
$h$ respectively.
\begin{list}{}{\itemsep=1pt\parsep=0pt}
\item[(i)] The polynomial
$f\jn g$ is the unique monic 
generator of the ideal $(f)\cap(g)=(g)\cap (f)$, and
as intersection is commutative, so is the join.
\item[(ii)] The polynomial
$f\jn (g\jn h)$ is the unique monic generator of the ideal 
$(f)\cap ((g)\cap (h))= ((f)\cap (g))\cap (h)=(f)\cap (g)\cap (h)$
and by the associativity of intersection, join is associative.
\item[(iii)]  Since $g\apdivides f\circ g$, $(f\circ g)\subseteq (g)$
and $(f\circ g)\cap (g)=(f\circ g)$.
\item[(iv)] The polynomial $(g\circ h)\jn (f\circ h)$ is the unique monic
generator of the ideal
$(f\circ h)\cap (g\circ h)
    =\{u\in\ap\setmid u=v\circ h~\hbox{and}~v\in (f)\cap (g)\}$,  
since all common left
multiples of $f\circ h$ and $g\circ h$ are also common multiples
of $f$ and $g$, composed with $h$.  Since $(f)\cap (g)$ has
generator $f\jn g$, the lemma follows.
\item[(v)] From the fact that
$(g)\supseteq (f)$ and $(h)\supseteq (f)$ 
it follows that $(g)\cap (h)\supseteq (f)$ and therefore
that $g\jn h\apdivides f$.
\QED
\end{list}
\end{proof}

The existence of a join does not give a construction for it.  The standard
commutative construction of the product divided by the greatest common
divisor is not appropriate in a non-commutative ring.  
However, an extension to the Euclidean scheme will provide a more
concrete representation of the join.    We first require the following
theorem.

\medskip
\begin{theorem}
Let $f,g,h\in\ap$.  If $f\apequiv g \bmod h$ then 
\[
f\jn h=a( (g\jn h)\apdiv g)\circ f,
\]
where $a\in F$ is such that the join is monic.
\end{theorem}
\begin{proof}
We know $ g\jn h=u\circ g$ for some $u\in\ap$.  From the assumptions
$f=g+Q\circ h$ for some $Q\in\ap$ and $u\circ f=u\circ g+u\circ Q\circ h$.
Since $h\apdivides u\circ g$ and $h\apdivides u\circ Q\circ h$, we know
$h\apdivides u\circ f$.  Because $f\apdivides u\circ f$ as well,
$h\jn f\apdivides u\circ f$ by lemma 4.1 (v).  
We now show that in fact $h\jn f=a(u\circ f)$ where $a\in F$ is such that
$a(u\circ f)$ is monic.
Suppose $ h\jn f=v\circ f$ for some $v\in\ap$. 
Then $h\jn f = v\circ f=v\circ g+v\circ Q\circ h$ and
since $h\apdivides v\circ f$ and $h\apdivides v\circ Q\circ h$,
it follows that $h\apdivides v\circ g$ as well.  
We know $g\apdivides v\circ g$, so
$g\jn h\apdivides v\circ g$.  Since $g\jn h=u\circ g$, $\expn v\geq \expn u$.
Therefore $h\jn f=a(u\circ f)$.  By definition $u=(g\jn h)\apdiv g$,
so $h\jn f=a((g\jn h)\apdiv g)\circ f$.
\QED
\end{proof}

\noindent 
The join of $f_1$ and $f_2$ can now be written as follows.

\medskip
\begin{theorem}
\[
f_1 \jn f_2 =b(\cdots(f_{n-1}\apdiv f_n)\circ f_{n-2})\apdiv f_{n-1})\circ 
    \cdots\circ f_3)\apdiv f_4)\circ f_2)\apdiv f_3)\circ f_1
\]
for some $b\in F$ chosen to make the join monic
(the alternation of $\circ$ and $\apdiv$ is similar to the alternation
of $+$ and $\cdot\;$ in Horner's rule, and the difference between
successive indices (from left to right)
is $+1,-2,+1,-2,+1,\ldots$ for each of the $2n-1$ terms).
\end{theorem}
\begin{proof}
In the Euclidean scheme, $f_i\apequiv f_{i+2}\bmod f_{i+1}$
for $1\leq i\leq n-1$ (with $f_{n+1}=0$), where $n$ is the length
of the sequence of $f_i$'s in the Euclidean scheme.
Also note that $f_n\apdivides f_{n-1}$.
From theorem 4.2 this implies that 
\[
f_i \jn f_{i+1}=a((f_{i+1}\jn f_{i+2})\apdiv f_{i+2})\circ f_i.
\]
for some $a\in F$.
We proceed by induction on $n$.

\noindent
If $n=2$ then $f_1\jn f_2=f_1$ and the theorem holds immediately.

\noindent
Now assume that the theorem holds for Euclidean schemes of length
less than $n$.
If the Euclidean scheme has length
$n$, then
\[
f_1\jn f_2 = a_1(\lbrack f_2\jn f_3\rbrack\apdiv f_3)\circ f_1,
\]
and by induction,
\[
\eqalign{
f_1\jn f_2 = & a_1(\lbrack a_2(\cdots(f_{n-1}\apdiv f_n)\circ f_{n-2})\apdiv
                        \cdots \apdiv f_5)\circ f_3)\apdiv f_4)\circ f_2
                        \rbrack\apdiv f_3)\circ f_1\cr
           = & b(\cdots(f_{n-1}\apdiv f_n)\circ f_{n-2})\apdiv
                        \cdots \apdiv f_5)\circ f_3)\apdiv f_4)\circ f_2
                        )\apdiv f_3)\circ f_1\cr
}
\]
for appropriate $a_1,a_2,b\in F$, and the theorem follows.
\QED
\end{proof}

\noindent
Theorem 4.3 also allows us to calculate the exponent of the  join.

\pagebreak
\begin{theorem}
$\expn (f_1 \jn f_2)=\expn f_1+\expn f_2-\expn (f_1 \mt f_2)$
\end{theorem}
\begin{proof}
Using the simple fact that $\expn f\circ g=\expn f+\expn g$ and 
$\expn f\apdiv g=\expn f-\expn g$, for $f,g\in\ap$,
a quick examination of the formula for join given in theorem 4.3 reveals that
\[
\eqalign{
\expn f_1\jn f_2= & \expn f_1+\expn f_2-\expn f_n\cr
              = & \expn f_1+\expn f_2-\expn (f_1 \mt f_2)\cr
}
\]
and the theorem is proved.
\QED
\end{proof}

\noindent
Continuing with the previous example,
\[
\eqalign{
f_1 \jn f_2 = & a((f_2\apdiv f_3)\circ f_1)\cr
                     = & a(((2x^3+x)\circ(2x^3+x))\apdiv (2x^3+x))\circ f_1\cr
                     = & a(2x^3+x)\circ(x^{27}+2x^9+x^3+2x)\cr
                     = & a(2x^{81}+2x^{27}+x^9+2x^3+2x)\cr
                     = & x^{81}+x^{27}+2x^9+x^3+x.\cr
}
\]
For verification we check that indeed
\[
\eqalign{
x^{81}+x^{27}+2x^9+x^3+x = & (x^9+x)\circ f_2\cr
                         = & (x^3+2x)\circ f_1.\cr
}
\]

If $f,g,h\in\ap$ with $g\neq 0$ and $f=g\circ h$, then $h$ is a multiplicative
factor as well as a right composition factor of $f$.  Thus, if
\[
f=Q\circ g+R
\]
where $Q,R\in\ap$ and $\expn R<\expn g$, then
\[
\eqalign{
f-R = & Q\circ g\cr
    = & Q^\prime g,\cr
}
\]
where $Q^\prime= (f-R)/g\in F[x]$.
Therefore $f=Q^\prime g+R$ and usual multiplicative 
division in $F[x]$ yields the same
remainder as compositional division in $\ap$.  This means that the 
right-Euclidean algorithm for $\ap$ just described generates the same
sequence of $f_i$'s as the usual multiplicative Euclidean algorithm
(though obviously a different sequence of $Q_i$'s) and we have the
following theorem.

\medskip
\begin{theorem}
If $f_1,f_2\in\ap$, then $f_1 \mt f_2$ is equal to the
usual multiplicative greatest common divisor of $f_1$ and $f_2$.
\end{theorem}

We can speak of $f_1$ and $f_2$ in $\ap$ as being {\it composition-coprime}
if $f_1 \mt f_2 =x$, 
and this is equivalent to saying that the  usual, multiplicative, 
greatest common divisor of $f_1$ and $f_2$ is $x$.

\subsection{The Structure of the Set of Decompositions}

The set of all distinct complete rational normal decompositions
of a given additive polynomial has a very strong internal structure.
Ore[1933a] develops this structure in the general context of
non-commutative left-Euclidean polynomial rings.

The central concept of Ore's theory is that of {\it transformation}.
Let $f,g\in\ap$ be monic. The monic polynomial
\[
g\trans f = (g \jn f)\apdiv g~\in\ap
\]
is called the transformation of $f$ by $g$.  By theorem 4.4, we determine
that
\[
\eqalign{
\expn(g\trans f) = & \expn (g \jn f) - \expn g \cr
                 = & \expn g + \expn f -\expn (g \mt f) -\expn g\cr
                 = & \expn f - \expn (g \mt f).\cr
}
\]
Obviously, if $f$ and $g$ are composition-coprime then
$\expn (g\trans f)=\expn f$ (though $g\trans f$ certainly does not have
to equal $f$).

The properties of transformation will be developed in the following
few theorems.  There does not seem to be an easy technique relating these
properties to the familiar multiplicative identities, say over the
integers.  Once might liken meet to integer greatest common divisor (gcd) and
join to least common multiple (lcm).  In this case transformation becomes
lcm divided by gcd.  But this is also a commutative construction, which is not
the case for transformation in the additive polynomials.

\medskip
\begin{theorem}
Let $f,g,h\in\ap$ be monic.
If $f\apequiv g\bmod h$ then $f\trans h= g\trans h$.
\end{theorem}
\begin{proof}
By theorem 4.2, $f\jn h=((g\jn h)\apdiv g)\circ f$. Dividing
both sides on the right by $f$, we get
$(f\jn h)\apdiv f=(g\jn h)\apdiv g$ (the multiplying
constant $a\in F$ from theorem 4.2 is one since $f,g$ and $h$ 
are assumed to be monic). 
Directly, we have that $f\trans h=g\trans h$.
\QED
\end{proof}

\medskip
\begin{theorem}
Let $f,g,h\in\ap$ be monic. If $h\apdivides f\circ g$ then
\begin{list}{}{\itemsep=0pt\parsep=1pt}
\item[(i)]$(g\trans h)\apdivides f$, and
\item[(ii)] if $f$ is indecomposable, $g\mt h=x$, and $h\neq x$,
then $g\trans h =f$.
\end{list}
\end{theorem}
\begin{proof}
\begin{list}{}{\itemsep=1pt\parsep=0pt}
\item[(i)]
The polynomials $g$ and $h$ are both right factors of $f\circ g$,
so there exists a $u\in\ap$ such that $f\circ g=u\circ(g\jn h)$.
Thus 
\[
\eqalign{
f = & (u\circ(g\jn h))\apdiv g\cr
  = & u\circ (g\trans h).\cr
}
\]
\item[(ii)]
As $f$ is indecomposable and $g \mt h=x$, we know $\expn(g\trans h)=\expn h$.
From (i), $(g\trans h)\apdivides f$ and since $\expn h=\expn (g\trans h)>0$ 
and $f$ is
indecomposable, $(g\trans h)= f$. \QED
\end{list}
\end{proof}

Two monic additive polynomials $f,g\in\ap$ are said 
to be {\it similar} if there
exists a $u\in\ap$  composition-coprime with $g$
such that $f=u\trans g$.  To denote similarity we write $f\sim g$.
Note that if $f$ and $g$ are similar then $\expn f=\expn g$.
We will  show that similarity is an equivalence relation.  First, we
must prove a preliminary lemma.

\medskip
\begin{lemma}
Let $f,g,h\in\ap$ be monic.
Then $(g\circ h)\trans f=g\trans(h\trans f)$.
\end{lemma}
\begin{proof}
\[
\eqalign{
(g\circ h)\trans f 
= & ((g\circ h)\jn f)\apdiv (g\circ h)\cr
= & ((g\circ h) \jn h \jn f) \apdiv (g\circ h)\cr
= & ((g\circ h) \jn (h \jn f)) \apdiv h)\apdiv g\cr 
= & (g \jn ((h\jn f)\apdiv h))\apdiv g\cr
= & g\trans (h\trans f).\qquad\qquad\mbox{\QED}\cr
}
\]
\end{proof}

\pagebreak
\begin{theorem}
Similarity is an equivalence relation.
\end{theorem}
\begin{proof}
Let $f,g,h\in\ap$ be monic.
\begin{list}{}{\itemsep=1pt\parsep=0pt}
\item[(i)] 
Similarity is reflexive since $x\trans f=f$.
\item[(ii)]  
Assume $f\sim g$, so that $f=u\trans g$ for some $u\in \ap$
such that $u \mt g=x$.  As $u$ and $g$ are composition-coprime,
there exist $Q,v\in\ap$ such that 
\[
v\circ u+Q\circ g=x.
\]
Therefore $v\circ u\apequiv x\bmod g$.  We have
\[
\eqalign{
g= x\trans g
= & (v\circ u)\trans g~~~~&\hbox{by theorem 4.6}\cr
= & v\trans (u\trans g)&\hbox{by lemma 4.8}\cr
= & v\trans f,\cr
}
\]
and $g\sim f$, so similarity is symmetric.
\item[(iii)] 
Assume $f\sim g$ and $g\sim h$.  Then there exist $u,v\in \ap$
such that $u \mt g=x$, $f=u\trans g$, 
$v \mt h=x$, and $g=v\trans h$.  By lemma 4.8,
\[
\eqalign{
f= & u\trans g\cr
 = & u\trans(v\trans h)\cr
 = & (u\circ v)\trans h.\cr
}
\]
Because $h$ and $f$ have the same exponent, $(u\circ v)\mt h=x$ and
$h\sim f$.  Thus similarity is transitive.
\end{list}

\noindent
By (i), (ii), and (iii) above, similarity is an equivalence
relation.
\QED
\end{proof}

An interesting case is that of the additive polynomial $x^p$, which has the
following property.

\medskip
\begin{lemma}
The only additive polynomial similar to $x^p\in\ap$ is $x^p$.
\end{lemma}
\begin{proof}
Let $u\in\ap$ be monic and composition-coprime with $x^p$.  
Thus, $u$ is simple ($u$ is monic and $u(0)=0$).  
Since $u\jn x^p=w\circ x^p$ for some $w\in\ap$,
$u\jn x^p$ is not simple.  We also know that $u\jn x^p=v\circ u$
for some $v=x^p+ax\in\ap$ for some $a\in F$.  
As $u$ is simple and $v\circ u$ is
not simple, $v=x^p$.  Therefore $u\jn x^p=x^p\circ u$ and
$u\trans x^p=x^p$.  
\QED
\end{proof}

\noindent
A further property of transformation is that
the transformation of a join is simply the transformation of its 
components.

\medskip
\begin{theorem}
Let $f,g,h\in\ap$ be monic.  Then
$h\trans (f \jn g)=(h\trans f) \jn (h\trans g)$.
\end{theorem}
\begin{proof}
We know
\[
\eqalign{
(h\trans (f \jn g))\circ h
= & ((h \jn (f \jn g))\apdiv h)\circ h\cr
= & (h \jn (f \jn g) )\cr
= & (h \jn f) \jn (h \jn g)\cr
= & (((h\jn f)\apdiv h) \jn ((h \jn g)\apdiv h))\circ h.\cr
}
\]
Dividing on the right by $h$ we get
\[
h\trans (f \jn g) = (h\trans f) \jn (h\trans g).\qquad\mbox{\QED}
\]
\end{proof}

Transformation will later be used to characterise the different
decompositions of a given additive polynomial.  It will be useful to
know the effect of transformation on a composition of additive polynomials.

\medskip
\begin{theorem}
Let $f,g,h\in\ap$ be monic.
Then
$h\trans (f\circ g)= ((g\trans h)\trans f)\circ (h\trans g)$.
\end{theorem}
\begin{proof}
We know that
\[
\eqalign{
h\trans (f\circ g)
= & h\trans ((f\circ g)\jn g)\cr
= & (h\trans (f\circ g)) \jn (h\trans g)\qquad \hbox{(by theorem 4.11)}\cr
= & Q\circ (h\trans g)\cr
}
\]
for some $Q\in\ap$.
This implies $Q\circ (h\jn g)=h\jn (f\circ g)$ and
\[
\eqalign{
Q\circ (g\trans h)
= & (h \jn (f\circ g))\apdiv g\cr
= & (g \jn (h \jn (f\circ g)))\apdiv g\cr
= & g\trans (h \jn (f\circ g))\cr
= & (g\trans h) \jn  (g\trans (f\circ g))\cr
= & (g\trans h) \jn  ((g \jn (f\circ g))\apdiv g)\cr
= & (g\trans h) \jn f.\cr
}
\]
Therefore, $Q=(g\trans h)\trans f$ and the theorem follows.
\QED
\end{proof}

This theorem can be easily
extended to consider the transformation of a composition of
many polynomials.

\medskip
\begin{theorem}
Let $f\in\ap$ be monic.
Assume $f=f_m\circ f_{m-1}\circ\cdots\circ f_1$ where $f_i\in\ap$ are 
monic for
$1\leq i\leq m$. 
Let $h\in\ap$ be monic and composition-coprime with $f$.  If $h_i\in\ap$ 
is defined by
\[
h_i=\biggl\lbrace
\eqalign{
&(f_{i-1}\circ f_{i-2}\circ\cdots\circ f_1)\trans h & ~~~\hbox{for}~i>1,\cr
&h \hfill & ~\hbox{for}~i=1,\cr
}
\]
for $1\leq i\leq m$ then
\[
h\trans f= \fbar_m\circ\fbar_{m-1}\circ\cdots\circ\fbar_1,
\]
where $\fbar_i=h_i\trans f_i$.
\end{theorem}
\begin{proof}
We proceed by induction on $m$.
If $m=1$ then $h\trans f_1=\fbar_1$.  Assume the theorem
is true if the number of factors is less than $m$ and that $m>1$.
From theorem 4.12,
\[
\eqalign{
h\trans f
= & h\trans((f_m\circ f_{m-1}\cdots\circ f_2)\circ f_1)\cr
= & ((f_1\trans h)\trans (f_m\circ\cdots\circ f_2))\circ (h\trans f_1).\cr
}
\]
Since $h$ and $f$ are composition-coprime, $\expn h\trans f=\expn f$.
Therefore, by computing the exponents of each side of the above
equation, we have
\[ 
\expn (f_m\circ\cdots\circ f_2)=
\expn ((f_1\trans h)\trans (f_m\circ\cdots\circ f_2)),
\]
and $(f_1\trans h)$ and $f_m\circ\cdots\circ f_2$ must be composition-coprime.
By induction, 
\[
(f_1\trans h)\trans (f_m\circ\cdots f_2)
=\fbar_m\circ\fbar_{m-1}\circ\cdots\circ\fbar_2
\]
with $\fbar_i=\hbar_i\trans f_i$ for $2\leq i\leq m$
where $\hbar_i$ is defined by
\[
\eqalign{
\hbar_i= & \Biggl\lbrace
{\eqalign{
&(f_{i-1}\circ f_{i-2}\circ\cdots\circ f_2)\trans (f_1\trans h) 
& ~\hbox{for}~i>2\cr
&f_1\trans h \hfill & ~\hbox{for}~i=2\cr
}}\cr
= & (f_{i-1}\circ\cdots\circ f_2)\trans h\qquad\hbox{(by lemma 4.8)},\cr
}
\]
and the theorem follows.
\QED
\end{proof}

The above theorems consider the transformations of arbitrary decompositions.
What are the effects of transformation on complete decompositions?  
We first need to know the relationship between the
decompositions of similar additive polynomials.

\medskip
\begin{theorem}
If $f,g\in\ap$, 
$f\sim g$ and $f$ is indecomposable, then $g$ is indecomposable.
\end{theorem}
\begin{proof}
Assume $f=u\trans g$ for some $u\in\ap$ composition-coprime with $g$.  
Suppose that $g=g_2\circ g_1$, 
where $\expn g_2>0$ and $\expn g_1>0$ (so $g$ is
decomposable).  By theorem
4.12, $f=((g_1\trans u)\trans g_2)\circ (u\trans g_1)$.  The polynomials
$u$ and $g_1$ are composition-coprime because $u$ and $g$ are 
composition-coprime
and $g_1\apdivides g$.  It follows that $u\trans g_1 \sim g_1$
and $\expn u\trans g_1=\expn g_1$.  Therefore 
$\expn ((g_1\trans u)\trans g_2)=\expn g_2$
and $f$ is decomposable, a contradiction.  Therefore
$g$ is indecomposable.
\QED
\end{proof}

\noindent
From 4.13 and 4.14 we immediately get the following theorem.

\medskip
\begin{theorem}
If $(f,(f_m,\ldots,f_1))\in\capdecall$ and $g\sim f$ (say $g=u\trans f$),
then there exists $(g,(g_m,\ldots,g_1))\in cAPDEC^F_*$ where 
$g_i\sim f_i$ for $1\leq i\leq m$ (specifically,
$g_1=u\trans f_1$ for some $u\in\ap$
and $g_i=((f_{i-1}\circ\cdots\circ f_1)\trans u)\trans f_i$ for
$2\leq i\leq m$).
\end{theorem}
\begin{proof}
By theorem 4.13 we transform the composition, giving a decomposition
of $g$.  The fact that this is a complete decomposition follows from
4.14.
\QED
\end{proof}

Transformation and similarity can be used to completely characterise the
relationship between decompositions.  Let $f,g\in\ap$ be monic.  If there
exists a monic $\fbar\in\ap$ such that $\fbar\sim f$ and 
$f=g\trans\fbar$, we say $f$ and $g$ are {\it transmutable}
or that they {\it transmute}.
The additive polynomial $\fbar$ is called a transmutation of $f$ by $g$.
In this case,
\[
\eqalign{
f\circ g= & (g\trans\fbar)\circ g\cr
        = & ((g \jn \fbar)\apdiv g)\circ g\cr
        = & g \jn \fbar\cr
        = & ((\fbar \jn g)\apdiv \fbar)\circ\fbar\cr
        = & (\fbar\trans g)\circ \fbar\cr
        = & \gbar\circ\fbar\cr
}
\]
where $\gbar=\fbar\trans g\in\ap$.  Because $f\sim\fbar$, $f$ and $\fbar$
have the same exponent and so $\gbar$ and $g$ also have the same exponent.
We know
\[
\eqalign{
\expn(\fbar\jn g) = & \expn\fbar+\expn g-\expn(\fbar\mt g)\cr
                          = & \expn (f\circ g)\cr
                          = & \expn (\fbar\circ g)\cr
                          = & \expn \fbar + \expn g.\cr
}
\]
Thus $\expn \fbar \mt g=0$ and $\gbar\sim g$.  

\noindent
There is no reason why there cannot exist an $\ftilde\in\ap$ such that
$\ftilde\sim f$, $\ftilde\neq\fbar$ and $f=g\trans\ftilde$.  The transmutation
of $f$ by $g$
is not unique.  Consider the following example over an arbitrary
field $F$ of characteristic $p$.
\[
\eqalign{
f=x^p+ax~~~~~~ & a\in F\cr
g=x^p+bx~~~~~~ & b\in F\cr
}
\]
Then $f\circ g=x^{p^2}+(a+b^p)x^p+abx$.  Assume $f\circ g=\gbar\circ\fbar$
where
\[
\eqalign{
\fbar=x^p+\abar x~~~~~ & \abar\in F\cr
\gbar=x^p+\bbar x~~~~~ & \bbar\in F\cr
}
\]
This implies $\bbar+\abar^p=a+b^p$ and $ab=\abar\bbar$.  Thus
\[
\eqalign{
0= &ab-\abar\bbar\cr
 = &ab-\abar(-\abar^p+a+b^p)\cr
 = &ab+\abar^{p+1}-a\abar-b^p\abar\cr
 = &a(b-\abar)-\abar(b-\abar)^p\cr
 = &(b-\abar)(a-\abar(b-\abar)^{p-1}),\cr
}
\]
and either $f=\gbar$ and $\fbar=g$, or
$\abar$ is a root of $\varphi=a-x(b-x)^{p-1}\in F[x]$~ 
($\gbar$ is uniquely determined as $\gbar =(f\circ g)\apdiv\fbar)$).
Noting that
\[
\eqalign{
f= & (\gbar\circ\fbar)\apdiv g\cr
 = & (g\jn (\gbar\circ\fbar))\apdiv g &~~~~
           ~~~~~~\hbox{since}~ g\apdivides \gbar\circ\fbar\cr
 = & g\trans (\gbar\circ\fbar)\cr
 = & g\trans ((\gbar\circ\fbar)\jn\fbar)\cr
 = & g\trans (\gbar\circ\fbar)\jn (g\trans\fbar)\cr
 = & f\jn (g\trans\fbar),\cr
}
\]
and $\expn f=\expn g\trans \fbar=1$, it follows that $f=g\trans\fbar$,
so $f\sim\fbar$ and the transmutation of $f$ by $g$ is $\fbar$.
Since the argument can be reversed, it implies that the polynomial 
$f$ can transmute by $g$ in up to $p$ different
ways, depending upon the roots of $\varphi$ in $F$.

A point worth noting is that $x^p$ does not transmute by $x^p$.  By lemma
4.10, the only polynomial similar to $x^p$ is $x^p$. If $x^p$ did transmute
by $x^p$, then $x^p=x^p\trans x^p=(x^p\jn x^p)\apdiv x^p=x$, a contradiction.

The set of all complete decompositions of $f\in\ap$ can be given structure 
using transmutation and similarity.  
Let $(f,(f_m,f_{m-1},\ldots,f_1))\in\capdecall$.
If $f_i$ and $f_{i-1}\circ f_{i-2}\circ\cdots\circ f_\ell$ transmute 
for some $i,\ell\in\nats$ with $m\geq i>\ell\geq 1$ then we
get another complete decomposition of $f$.  As in theorem 4.13, this is
\[
(f,(f_m,f_{m-1},\ldots,f_{i+1},\fbar_{i-1},\ldots,\fbar_\ell,
\fbar_i,f_{\ell-1},\ldots,f_1))\in\capdecall.
\]
where $\fbar_j\sim f_j$ for $\ell\leq j\leq i$.
We say these two decompositions are {\it single-transmutation equivalent}.
Letting $(f,(f_m^{(0)},\ldots,f_1^{(0)}))\in\capdecall$, if there
is a sequence
\[
\eqalign{
(f_m^{(0)},\ldots,f_1^{(0)})\cr
(f_m^{(1)},\ldots,f_1^{(1)})\cr
\vdots~~~~~~~~~~\cr
(f_m^{(t)},\ldots,f_1^{(t)})\cr
}
\]
where $(f,(f_m^{(i)},\ldots,f_1^{(i)}))\in\capdecall$ for $1\leq i\leq t$, and 
$(f_m^{(i)},\ldots,f_1^{(i)})$ and $(f_m^{(i+1)},\ldots,f_1^{(i+1)})$
are single-transmutation equivalent for $1\leq i<t$, we say that
$f_m^{(0)},\ldots,f_1^{(0)}$ and $f_m^{(t)},\ldots,f_1^{(t)}$
are {\it transmutation equivalent}.
Transmutation equivalence is the reflexive 
transitive closure of single-transmutation equivalence.

\medskip
\begin{theorem}
All complete rational normal decompositions of a monic 
$f\in\ap$ are transmutation
equivalent.
\end{theorem}
\begin{proof}
Let $(f,(f_m,\ldots,f_1))=(f,(f_m^{(0)},\ldots,f_1^{(0)}))$
and $(f,(g_r,\ldots,g_1))$ for $r,m\in\nats$ be
complete rational normal decompositions of $f$.
We prove the theorem by induction on $m$.
If $m=1$ then $f$ is indecomposable and $f_1=g_1$, so the statement
is true.

Assume the theorem is true for complete decompositions of length
less than $m$.
Let $k\in\nats$ be the smallest number such that 
$g_1\apdivides (f_k\circ f_{k-1}\circ\cdots\circ f_1)$.  If $k=1$ then
$f_1=g_1$ (they are both indecomposable), 
and by induction $(f\apdiv g_1,(f_m,\ldots,f_2))$ and
$(f\apdiv g_1,(g_r,\ldots,g_2))\in\capdecall$ are transmutation equivalent.
Therefore $(f,(f_m,\ldots,f_1))$ and $(f,(g_r,\ldots,g_1))$ are transmutation 
equivalent.

If $k>1$ then
$f_{k-1}\circ f_{k-2}\circ\cdots\circ f_1$ and $g_1$ are 
composition-coprime.  By theorem 4.7(ii), 
\[
f_k=(f_{k-1}\circ f_{k-2}\circ\cdots\circ f_1)\trans g_1,
\]
$f_k\sim g_1$, and $f_k$ and $(f_{k-1}\circ\cdots\circ f_1)$ are
transmutable.  Therefore
\[
\eqalign{
f_k\circ f_{k-1}\circ\cdots\circ f_1
= & (g_1\trans (f_{k-1}\circ\cdots\circ f_1))\circ g_1\cr
= & \fbar_{k-1}\circ \fbar_{k-2} \circ \cdots \circ \fbar_1\circ g_1\cr
}
\]
by theorem 4.13, where $f_i\sim\fbar_i$ for $1\leq i\leq k-1$.  
Thus 
\[
(f,(f_m,f_{m-1},\ldots,f_{k+1},\fbar_{k-1},\ldots,\fbar_1,g_1))
\]
and  $(f,(f_m,\ldots,f_1))$ are single-transmutation equivalent.
Also, by the inductive hypothesis,
$(f\apdiv g_1,(f_m,f_{m-1},\ldots,f_{k+1},\fbar_{k-1},\ldots,\fbar_1))$
and $(f\apdiv g_1,(g_r,\ldots,g_2))$ are transmutation equivalent.
Therefore $(f,(f_m,\ldots,f_1))$ and $(f,(g_r,\ldots,g_1))$
are transmutation equivalent and the theorem follows.
\QED
\end{proof}

Any two single-transmutation equivalent decompositions have the
same number of indecomposable factors in any complete decomposition,
and these factors are similar in pairs.
Since similarity is transitive, we immediately get the following
corollary.

\medskip
\begin{corollary}
Any two complete decompositions of $f\in\ap$ have the same number of
factors and if $(f,(f_m,f_{m-1},\ldots,f_1))$,
$(f,(g_m,g_{m-1},\ldots,g_1))\in\capdecall$ for some $m>0$, there
exists a permutation $\sigma$ of $\{1,\ldots,m\}$ such that
$g_i\sim f_{\sigma_i}$ for $1\leq i\leq m$.
\end{corollary}

\subsection{Completely Reducible Additive Polynomials}

A monic additive polynomial is said to be {\it completely reducible} if it is
the join of a set of indecomposable additive polynomials.  Completely
reducible additive polynomials have a number of nice properties which
we will examine mathematically and algorithmically.

\medskip
\begin{lemma}
A completely reducible polynomial $f\in\ap$ can be represented in the form
\[
f=h_r \jn h_{r-1} \jn \cdots \jn h_1
\]
where, for $1\leq i\leq r$, $h_i$ is indecomposable and no one of the
$h_i$'s is a right composition factor of the join of the others.
This is called an {\it indecomposable basis} of $f$.
\end{lemma}
\begin{proof}
Let $f\in\ap$ be completely reducible and let 
$u_1,u_2,\ldots,u_m\in \ap$ be the indecomposable right
factors of $f$.  We know that $f$ is the join of these right factors
by the definition of completely reducible, and that there is a finite
number of them since they are all multiplicative divisors of $f$.
Consider the following method for determining an indecomposable
basis for $f$.

\medskip
{\tt\obeylines
~~~~1)~Let $T:=\emptyset$
~~~~2)~Let $g^{(0)}:= x$
~~~~3)~For $i$ from 1 to $m$
~~~~~~~3.1) if $u_i\mt g^{(i-1)}=x$ then
~~~~~~~~~~~~~~~3.1.1) let $g^{(i)}:=g^{(i-1)}\jn u_i$
~~~~~~~~~~~~~~~3.1.2) let $T:=T\cup\{u_i\}$
~~~~~~~~~~~~else
~~~~~~~~~~~~~~~3.1.3) let $g^{(i)}:=g^{(i-1)}$
~~~~~~~3.2) if $g^{(i)}=f$, then quit, returning $T$
}

\medskip
At step 3.1, we know that if 
$u_i\mt g^{(i-1)}\neq x$ then $u_i\apdivides g^{(i-1)}$ for $i\geq 1$
because $u_i$ is indecomposable.  
In this case it will not change the join $g^{(i-1)}$ and is
redundant.
Because $f$ is the join of all its indecomposable right factors,
$g^{(k)}=f$ for some $k\leq m$.
From the construction, the exponent of the join of the polynomials 
in $T$ is the 
sum of the exponents of these polynomials.  By theorem 4.4, therefore,
any one polynomial in $T$ is composition-coprime with the join of the others
in $T$.
\QED
\end{proof}

\noindent
Note that we can choose any indecomposable right factor $u_1$ we want in
the above procedure. 

A polynomial $f\in\ap$ is said to be {\it completely transmutable}
if in any complete decomposition, any two adjacent indecomposable
factors are transmutable.

\pagebreak
\medskip
\begin{theorem}
An additive polynomial is completely reducible if and only if 
it is completely transmutable.
\end{theorem}
\begin{proof}
We first show that if $f\in\ap$ is completely reducible then it is completely
transmutable.  We proceed by induction on the number of indecomposable
factors in a  complete decomposition of $f$.  
Assume $f=g_1\in\ap$, where $g_1$ is indecomposable.
Then $f$ is completely transmutable.
Now, assume the statement  is true if $f$ has less than $m$
indecomposable factors in any complete decomposition.  Let
$f=g_m\circ g_{m-1}\circ\cdots\circ g_1=\gbar\circ g_1$ where
$g_i\in\ap$ are indecomposable for $1\leq i\leq m$ and $m>2$, and
$\gbar=f\apdiv g_1\in\ap$.  As $f$ is completely reducible,
it has an indecomposable basis $\{g_1,h_2,h_3,\ldots,h_\ell\}$
with $h_i\in\ap$ indecomposable 
for $2\leq i\leq \ell$.  We get
\[
\eqalign{
\gbar = & ((h_\ell \jn \cdots \jn h_2) \jn g_1)\apdiv g_1\cr
      = & g_1\trans (h_\ell \jn \cdots \jn h_2)\cr
      = & (g_1\trans h_\ell) \jn (g_1 \trans h_{\ell-1}) 
          \jn \cdots \jn (g_1\trans h_2)\cr
}
\]
and $\gbar\sim (h_\ell \jn \cdots \jn h_1)$. Thus, $\gbar$  is completely 
reducible and, by the inductive assumption, completely transmutable.
We have shown the leftmost $m-1$ factors of any complete
decomposition of $f$ are completely transmutable.
Now we need only show that $g_1$ and $g_2$ are transmutable. We know 
$g_2\apdivides \gbar=g_1\trans (h_\ell \jn \cdots \jn h_2)$.
By theorem 4.15, all 
complete decompositions of $g_1\trans (h_\ell \jn \cdots \jn h_2)$
are simply decompositions of $h_\ell \jn \cdots \jn h_2$ transformed
by $g_1$ ($g_1$ and $(h_\ell \jn \cdots \jn h_2)$ are composition-coprime).
Therefore, $g_2=g_1\trans u$ for some $u\in\ap$ similar to $g_2$, and 
$g_1$ and $g_2$ are transmutable.  Thus, any completely
reducible additive polynomial is completely transmutable.

We now show that if $f\in\ap$ is completely transmutable then $f$ is
completely reducible.  Once again we prove this by induction on the
number of indecomposable factors in a complete decomposition of $f$.  
If $f=g_1$, where
$g_1\in\ap$ is indecomposable, then $f$ is obviously
completely reducible.  
Assume the statement holds if $f$ has fewer than $m$ 
factors in a complete decomposition.
Then assume
$f = g_m\circ g_{m-1}\circ\cdots\circ g_1$
where $g_i\in\ap$ for $1\leq i\leq m$.  Also,  let 
$\gbar=g_m\circ g_{m-1}\circ\cdots\circ g_2$.
Since $\gbar$ is completely transmutable, it is completely
reducible by the inductive assumption, so
$\gbar=h_\ell \jn \cdots \jn h_2$ where $\ell\in\nats$ is greater than two and
$h_2,\ldots,h_\ell\in\ap$ are indecomposable.
Each of the $h_i$ are indecomposable right factors of $\gbar$ and
because $f$ is completely transmutable,
each of the $h_i$'s can be transmuted with $g_1$.
Thus, $h_i=g_1\trans \hbar_i$ for some $\hbar_i\in\ap$, $\hbar_i\sim h_i$
for $1\leq i\leq \ell$.  
Therefore
\[
\eqalign{
f=\gbar\circ g_1
= & (h_\ell \jn h_{\ell-1}\jn\cdots \jn h_2)\circ g_1\cr
= & ((g_1\trans \hbar_\ell)\jn (g_1\trans \hbar_{\ell-1})\jn
    \cdots\jn (g_1\trans \hbar_2))\circ g_1\cr
= & \hbar_\ell\jn\hbar_{\ell-1}\jn\cdots\jn\hbar_2\jn g_1,\cr
}
\]
and since the $\hbar_i$ are indecomposable for $2\leq i\leq\ell$ and
$g_1$ is indecomposable, $f$ is completely reducible.
\QED
\end{proof}

Note that since $x^p$ does not transmute with itself, this theorem implies
that any completely reducible polynomial can have at most one composition
factor $x^p$ in an arbitrary complete decomposition.

\noindent
A strong relationship exists between the composition factors of an arbitrary
complete decomposition and an arbitrary indecomposable basis.

\medskip
\begin{theorem}
Let $f\in\ap$ be completely reducible, 
$(f,(f_m,f_{m-1},\ldots,f_1))\in\capdecall$, and $h_1,\ldots,h_\ell$ be
an indecomposable basis for $f$.  Then $m=\ell$ and there exists a
permutation $\sigma$ of $\{1,\ldots,m\}$ such that $h_i\sim f_{\sigma_i}$
for $1\leq i\leq m$.
\end{theorem}
\begin{proof}
We proceed by induction on $m$.  If $m=1$, then $f$ is indecomposable,
and $\ell=1$ and $\sigma$ is the identity permutation.  Now assume the
hypothesis is true for all complete decompositions of length less
than $m$.  Let $f\in\ap$ be completely reducible,
$(f,(f_m,f_{m-1},\ldots,f_1))\in\capdecall$, and $h_1,\ldots,h_\ell$
be an indecomposable basis for $f$.
Since $h_1$ is an indecomposable right factor of $f$, there
exists a decomposition $(f,(f_m',f_{m-1}',\ldots,f_2',h_1))\in\capdecall$
and by corollary $4.17$ a permutation $\tau$ of $\{1,\ldots,m\}$ such that
$f_i'\sim f_{\tau_i}$ for $1\leq i\leq m$.  Now,
$(f\apdiv h_1,(f_m',f_{m-1}',\ldots,f_2'))\in\capdecall$ and
by the inductive assumption,
\[
\eqalign{
f\apdiv h_1 
= & h_1\jn(h_\ell\jn h_{\ell-1}\jn\cdots\jn h_2)\apdiv h_1\cr
= & h_1\trans (h_\ell\jn h_{\ell-1}\jn\cdots\jn h_2)\cr
= & (h_1\trans h_\ell) \jn (h_1\trans h_{\ell-1})\jn
        \cdots\jn (h_1\trans h_2),\cr
}
\]
giving an indecomposable basis for $f\apdiv h_1$.
By the inductive hypothesis $\ell-1=m-1$ so $m=\ell$ and there exists a
permutation $\mu$ of $\{2,\ldots,m\}$ such that $h_i\sim f_{\mu_i}'$
for $2\leq i\leq m$.  Extending this to a permutation $\mubar$ of
$\{1,\ldots,m\}$ by letting $\mubar_1=1$ we find that
$\sigma=\tau\mubar$ has the property that 
$h_i\sim f_{\mubar_i}'\sim f_{\sigma_i}$ for $1\leq i\leq m$.
\QED
\end{proof}

\subsection{The Uniqueness of Transmutation}

A question which will concern us algorithmically is that of the
uniqueness of transmutation.  We can characterise how additive
polynomials transmute in terms of the similar factors in an
arbitrary complete decomposition.  Let $g\in\ap$ be monic of
exponent $\nu$ with complete decomposition 
$(g,(g_m,g_{m-1},\ldots,g_1))\in\capdecall$.  Let $f\in\ap$ be monic
and indecomposable.
Assume that $f$ transmutes by $g$ in two distinct ways, say
$f=g\trans\fbar$ for $\fbar\in\ap$, $\fbar\sim f$, and $f=g\trans\ftilde$
for $\ftilde\in\ap$, $\ftilde\sim f$, and $\ftilde\neq \fbar$.  Then
$f\circ g$ has complete decompositions 
$(f\circ g,(\gbar_m,\ldots,\gbar_1,\fbar))$ where $\gbar_i\in\ap$ and
$\gbar_i\sim g_i$ for $1\leq i\leq m$ and
$(f\circ g,(\gtilde_m,\ldots,\gtilde_1,\ftilde))$ where $\gtilde_i\in\ap$ and
$\gtilde_i\sim g_i$ for $1\leq i\leq m$.  We know $\ftilde\apdivides f\circ g$.
Let $k\in\nats$ be the smallest number such that 
$\ftilde\apdivides \gbar_k\circ\gbar_{k-1}\circ\cdots\circ \gbar_1\circ\fbar$
($\ftilde$ does not divide $\fbar$).  Then $\ftilde$ and
$\gbar_{k-1}\circ \gbar_{k-2}\circ\cdots\circ \gbar_1\circ\fbar$ are 
composition-coprime and by theorem 4.7,
$\gbar_k=(\gbar_{k-1}\circ\cdots\circ\gbar_1)\trans\ftilde$
and $\gbar_k\sim\ftilde$.  We have the following theorem:

\medskip
\begin{theorem}
Let $f,g\in\ap$ be monic with 
$(g,(g_m,\ldots,g_1))\in\capdecall$ and $f$ indecomposable. If
$f$ and $g$ transmute in two or more distinct ways, then $f\sim g_i$
for some $i$ such that  $1\leq i\leq m$.
\end{theorem}

We can further characterise when non-unique transmutations occur by 
showing the following theorem about transmutations in general.

\medskip
\begin{theorem}
Let $f,g,h\in\ap$ be monic.  If $f$ transmutes by 
$g\circ h$ with transmutation $\fbar$, then $f$ and $g$ transmute, and if
a transmutation of $f$ by $g$ is $\ftilde$, then
$\ftilde$ and $h$ transmute as well.
\end{theorem}
\begin{proof}
If $f= (g\circ h)\trans \fbar$, then $f=g\trans (h\trans\fbar)$ by
lemma 4.8.  Since $\fbar$ and $g\circ h$ are composition-coprime, $\fbar$
and $h$ are composition-coprime.  Let $\ftilde=h\trans\fbar$. Then
$\ftilde$ transmutes with $g$ since $f=g\trans\ftilde$.  Furthermore, because
$\ftilde=h\trans\fbar$, $\ftilde$ and $h$ transmute.
\QED
\end{proof}

\noindent
This theorem can be extended to the case when $f$ transmutes by
$h_m\circ h_{m-1}\circ\cdots\circ h_1$.

\medskip
\begin{theorem}
Let $f,h_i\in\ap$ for $1\leq i\leq m$.  If $f$ transmutes by
$h=h_m\circ h_{m-1}\circ\cdots\circ h_1$, then 
\begin{list}{}{\itemsep=1pt\parsep=0pt}
\item[(i)] $f$ transmutes by $h_m$, with transmutation
$f^{(m)}\in\ap$, for some $f^{(m)}\sim f$, and
\item[(ii)] for $m\geq i>1$, $f^{(i)}$ transmutes by $h_{i-1}$
with transmutation $f^{(i-1)}\in\ap$
for some $f^{(i-1)}\sim f$.
\end{list}
\end{theorem}
\begin{proof}
We proceed by induction on $m$.  The base case, where $m=2$ follows
directly from theorem 4.22.  Assume the theorem holds if $h$ is given
as a composition of less than $m$ factors.  If $h$ is given
as a composition of $m$ factors then by theorem 4.22 $f$ transmutes
by $h_m\;$.  Let $f^{(m)}\in\ap$ be the transmutation of $f$ by $h_m$.
Also by theorem 4.22, $f^{(m)}$ transmutes with $h_{m-1}\circ\cdots\circ h_1$.
By the
inductive hypothesis, $f^{(m)}$ transmutes by $h_{m-1}$ with
some transmutation $f^{(m-1)}\sim f$ and for $m-1\geq i>1$,
$f^{(i)}$ transmutes by $h_{i-1}$ with some transmutation
$f^{(i-1)}\sim f^{(m)}\sim f$.
\QED
\end{proof}

If $f$ transmutes by $h$ in two distinct ways, then 
for an arbitrary decomposition
$(h,(h_m,h_{m-1},\ldots,h_1))\in cAPDEC^F_*$,
$f$ transmutes by each $h_i$ in turn for $m\geq i\geq 1$.
Which 
transmutation of $f$ by $h$ is obtained is determined entirely
by the transmutation of $f^{(i)}$ by $h_i$ for $1\leq i\leq m$.  Since
the transmutation of
$f^{(i)}$ by $h_i$ is unique if $f^{(i)}\nsim h_i$,
the transmutation of $f$ by $h$ is determined completely by the
transmutation of
$f^{(i)}$ by $h_i$ for $1\leq i\leq m$ when $f^{(i)}\sim h_i$.
With this in mind, we define an additive polynomial $f\in\ap$ to be
{\it similarity free} if in an arbitrary complete decomposition, no two
of the composition factors are similar.  In a similarity free additive
polynomial, all transmutations of the factors are unique.  

The previous theorem also allows us to strengthen theorem 4.16 of Ore's.
We say two complete decompositions 
$(f,(f_m,f_{m-1},\ldots,f_1))$ and $(f,(g_m,\ldots,g_1))$ in $cAPDEC^F_*$
are {\it single-indecomposable-transmutation
equivalent} if $f_i=g_i$ for $1\leq i\leq m$ or
there exists an $\ell\in\nats$ with $1\leq \ell<m$ such that
\[
(f,(g_m,\ldots,g_{\ell+1},g_{\ell},g_{\ell-1},g_{\ell-2},\ldots,g_1))
=(f,(f_m,\ldots,f_{\ell+1},\fbar_{\ell-1},\fbar_{\ell},f_{\ell-2},\ldots,f_1))
\]
where $\fbar_{\ell}\sim f_\ell$ is the transmutation of $f_\ell$
by $f_{\ell-1}$ and 
$\fbar_{\ell-1}=\fbar_\ell\trans f_{\ell-1}\sim f_{\ell-1}$.

We define {\it indecomposable-transmutation equivalence} as the
reflexive transitive closure of single-indecomposable-transmutation
equivalence.
Thus, two decompositions are indecomposable-transmutation
equivalent if one can be obtained from the other by a sequence
of transmutations of adjacent indecomposable factors.

\medskip
\begin{theorem}
Two complete decompositions are indecomposable-
transmutation equivalent
if and only if they are transmutation equivalent.
\end{theorem}
\begin{proof}
If two complete 
decompositions are indecomposable-transmutation equivalent, then
they are transmutation equivalent.
By theorem 4.23, any transmutation of an indecomposable additive
polynomial with a composition of indecomposable polynomials
is equivalent to a sequence of transmutations with each of the
indecomposable factors in turn.  Thus, single-transmutation equivalent
decompositions are indecomposable-transmutation equivalent. 
Since transmutation equivalence is just the reflexive transitive closure
of single-transmutation equivalence, transmutation equivalent
decompositions must be indecomposable-transmutation equivalent.
\QED
\end{proof}

\noindent
As an immediate  corollary we get a stronger version of theorem 4.16.

\medskip
\begin{corollary}
All complete decompositions of an additive polynomial $f\in\ap$
in $\capdecall$ are indecomposable-transmutation equivalent. 
\end{corollary}

\noindent
From now on we will simply say two complete decompositions are
transmutation equivalent to mean
indecomposable-transmutation equivalent.

\subsection{The Number of Complete Decompositions}

Using the methods from chapter 3 as well as the material from
this chapter, we can now prove an upper bound on the number of complete
decompositions of a (not necessarily simple) additive polynomial.

Let $f\in\ap$ be monic of degree $n=p^\nu$.  
Then $f=g\circ x^{p^\ell}$ where
$g\in\ap$ is simple and $\ell\geq 0$.  From theorem 3.10, we know
that $g$ has at most $n^{\mu\log n}$ complete decompositions
in $cSAPDEC^F_*$ where $\mu=(2\log p)^{-1}$.  For each decomposition
$(g,(g_m,g_{m-1},\ldots,g_1))\in cSAPDEC^F_*$, $f$ has a decomposition
\[
(f,(g_m,g_{m-1},\ldots,g_1, \overbrace{x^p,x^p,\ldots,x^p}^{\ell~{\rm times}}))
\in\capdecall.
\]
Without changing the order of the $g_i$'s for $m\geq i\geq 1$, 
and allowing for transformations into similar factors,
we can distribute the indecomposable factors
$x^p$ throughout the decomposition of $f$.  
There are up to
\[
{{m+\ell}\choose \ell}\leq {\nu\choose\ell}
\leq 2^\nu\leq n 
\] 
such distributions.  
We know these are all the decompositions because all decompositions of
additive polynomials are transmutation equivalent.
Because there are $n^{\mu\log n}$ complete decompositions
of $g\in cSAPDEC^F_*$, there are at most $n^{1+\mu\log n}$
complete decompositions of $f\in\capdecall$.
We have shown the following generalisation of lemma 3.10.

\medskip
\begin{theorem}
If $f\in\ap$ has degree $n$, then $f$ has at most $n^{1+\mu\log n}$
decompositions in $\capdecall$.
\end{theorem}

Note that in the case of a perfect field $F$, for any $u\in \ap$,
we know $u\circ x^p=x^p\circ u^{1\over p}=x^p\circ\ubar$
where $\ubar=u^{1\over p}\circ x^p\in\ap$.  In this case there are,
therefore, exactly ${{m+\ell}\choose \ell}$ times as many 
complete decomposition of $f$ than of $g$ in $cAPDEC^F_*$.


\newpage
\section{Decomposing Additive Polynomials}

\subsection{The Model of Computation}

The model of computation used in this chapter 
is the ``arithmetic Boolean circuit'' 
as described in chapter 2, section A.
Once again, let $\sf(n)$ be the number of field operations required to factor
an arbitrary univariate polynomial $f\in F[x]$ of degree $n$ into
irreducible factors (where $F$ is a field of characteristic $p$).  
In this chapter, $\sf(n)$ is assumed to be polynomially bounded.
It is also assumed to satisfy the property that for $p,\nu\in\nats$,
\[
\sum_{0\leq i\leq\nu} \sf(p^i)= O(\sf(p^\nu)).
\]

\noindent
The following theorem will also be useful in the analysis of some of our
algorithms.

\medskip
\begin{lemma}
If $p,\nu,d\in\nats$ with $p\geq 2$ and $\nu\geq 1$, then
\[
\sum_{1\leq i\leq \nu} i^dp^i\leq 3\nu^dp^\nu.
\]
\end{lemma}
\begin{proof}
We proceed by induction on $\nu$.  If $\nu=1$ then the theorem is trivially
true.  Assume it is true for $\nu<k$. Then
\[
\eqalign{
\sum_{1\leq i\leq k} i^dp^i
=    & \sum_{1\leq i\leq k-1}i^dp^i + k^dp^k\cr
\leq & 3(k-1)^dp^{k-1}+k^dp^k\cr
\leq & (3/2)k^dp^k+k^dp^k\cr
\leq & 3k^dp^k,\cr
}
\]
and the theorem holds for all $\nu\geq 1$.
\QED
\end{proof}

\subsection{The Cost of Basic Operations in $\protect\ap$.}

\noindent
Let $f,g\in\ap$ be of exponents $\nu$ and $\rho$ respectively,
and max$(\nu,\rho)\leq \delta$. The following
analyses of the basic operations in $\ap$ are probably not optimal, but will be
sufficient for our purposes.

\medskip
\begin{lemma}
(Composition)  Computing $f\circ g$ requires at most $O(\delta^2\log p)$
field operations.
\end{lemma}
\begin{proof}
Each coefficient of $g$ must be raised to the $p^i$th power for 
$0\leq i\leq\nu\leq \delta$.  
This requires $O(\rho \delta\log p)= O(\delta^2\log p)$
field operations.~\QED
\end{proof}

\medskip
\begin{lemma}
(Division with remainder)  If $g\neq 0$, 
computing $Q,R\in\ap$ such that $f=Q\circ g+R$
and $\expn R<\expn g$ requires $O(\delta^2\log p)$ field operations.
\end{lemma}
\begin{proof}
The cost of computing right division with remainder is dominated by the
cost of raising $g$ (and hence each coefficient of $g$) to the
$p^i$th power for $0\leq i\leq \nu-\rho$.  This requires
$O(\rho(\nu-\rho)\log p)=O(\delta^2\log p)$ field operations.~\QED
\end{proof}

\medskip
\begin{lemma}
(Meet)  Computing $f \mt g$ requires $O(\delta^3\log p)$ field operations.
\end{lemma}
\begin{proof}
In the Euclidean scheme described in the previous section, each step
involves right division with remainder of additive polynomials with
exponent at most $\delta$.  There are at most $\delta$ steps.  Therefore
we can compute the meet of $f$ and $g$ with $O(\delta^3\log p)$ divisions.
\QED
\end{proof}

\medskip
\begin{lemma}
(Join)  Computing $f \jn g$ requires $O(\delta^3\log p)$ field operations.
\end{lemma}
\begin{proof}
Using the formula of theorem 4.3, we must first compute the $f_i$'s of the
Euclidean scheme.  This requires $O(\nu^3\log p)$ field operations by the
previous lemma.  Computing the join then 
requires at most $\delta$ divisions and
$\delta$ compositions of polynomials with 
exponents not exceeding $2\delta$.  Thus,
computing the join requires
$O(\delta^3\log p)+\delta O((2\delta)^2\log p)= 
 O(\delta^3\log p)$ field operations.~\QED
\end{proof}

\medskip
\begin{lemma}
(Transformation) Computing $f\trans g$ requires $O(\delta^3\log p)$ field
operations.
\end{lemma}
\begin{proof}
By definition $f\trans g=(f \jn g)\apdiv f$, and the number of field 
operations involved is dominated by the number of field operations
required to compute the join, which is $O(\delta^3\log p)$.~\QED
\end{proof}

\subsection{The Minimal Additive Multiple}

Let $f$ be an arbitrary monic polynomial in $F[x]$.  A concept which will prove
extremely useful when dealing computationally with additive polynomials
is that of the {\it minimal additive multiple} $\fhat\in\ap$ of $f$.  This is
the monic additive polynomial of smallest exponent such that $\fhat$ is
a multiple of $f$.  The idea of a minimal additive multiple first
appears in Ore[1933b].

If $f=0$, then $\fhat=f=0\in\ap$.  If $f\in\monicpoly$ does not equal
zero, the following algorithm computes the minimal additive multiple
$\fhat$ of $f$.

\medskip
{\tt \obeylines
MinAddMult : $\monicpoly\rightarrow\ap$
~~~~Input:~~-~$f\in\monicpoly$ of degree $n\geq 1$.
~~~~Output:~-~$\fhat\in\ap$, the minimal additive multiple of $f$.

~~~~1) For $i$ from $0$ to $n$,
~~~~~~~~1.1) compute $h_i\equiv x^{p^i}\bmod f$
~~~~~~~~~~~~~where $h_i\in F[x]$ and $\deg h_i<\deg f$.
~~~~2) Let $k\in\nats$ be the smallest number with $0\leq k\leq n$
~~~~~~~such that there exists $\alpha_0,\alpha_1,\ldots,\alpha_{k-1}\in F$ 
~~~~~~~such that $h_k=\sum_{0\leq j<k}\alpha_jh_j$.
~~~~3) Return $\fhat=x^{p^k}-\sum_{0\leq j<k}\alpha_jx^{p^j}$.
}
\medskip

\noindent
We know $\fhat$ is a multiple of $f$ because
\[
\eqalign{
\fhat = & x^{p^k}-\sum_{0\leq j<k}\alpha_jx^{p^j}\cr
\equiv & (h_k+\sum_{0\leq j<k} \alpha_jh_j) \bmod f\cr
\equiv & 0 \bmod f.\cr
}
\]
The existence of any additive multiple of $f$ with exponent $\ell<k$ would
imply $h_0,\ldots,h_\ell$ are linearly dependent, which
is false.  Thus $\fhat$ is the minimal additive multiple of $f$.
We know a solution always exists since $n+1$ vectors in $F^n$ must
be linearly dependent.

The number of field operations to compute $h_j$ 
for $0\leq j\leq n$ is $O(nM(n)\log p)$.  
The determination of $k$ can be done by a modified Gaussian elimination
on the $n\times n$ matrix $H$ where $H_{ij}$ is the coefficient
of $x^i$ in $h_j$ for $0\leq i,j<n$.  We proceed in stages from $0$ to $n-1$.
Let $H^{(0)}=H$.
At stage $\ell$ (with $0\leq\ell<n$) we perform Guassian elimination on
rows zero through $\ell$ of $H^{(k)}$ obtaining $H^{(k+1)}$
(leaving rows $\ell+1$ through $n-1$ unchanged).  If, at the end of stage
$\ell$, row
$\ell$ of $H^{(\ell+1)}$ has all entries zero, then rows zero through
$\ell$ of $H$ are linearly dependent and we can return $k=\ell$.  
At each stage of this elimination we only perform a row 
operation on row $\ell$,
so each stage requires $O(n^2)$ field operations over $F$.
The complete procedure then requires $O(n^3)$ field operations.
Given $k$, it is simple linear algebra to find 
$\alpha_0,\alpha_1,\ldots,\alpha_{k-1}$ such that 
$h_k=\sum_{0\leq j<k}\alpha_jh_j$.  This also require $O(n^3)$ field 
operations.  We get the following theorem:

\medskip
\begin{theorem}
Let $f\in F[x]$ be monic 
of degree $n$.  The minimal additive multiple $\fhat\in\ap$
of $f$ can be determined in $O(n^3)$ field operations.
\end{theorem}

If $\ftilde$ is also an additive multiple of $f$, then by theorem 4.5, 
$h=\fhat \mt \ftilde$ is equal to the multiplicative greatest
common divisor of $\fhat$ and $\ftilde$.  Thus  $f$ divides 
$\fhat \mt \ftilde$ and this is an additive multiple of $f$.  But $\fhat$
is the minimal additive multiple of $f$ so $\fhat=\fhat \mt \ftilde$
and $\fhat\apdivides\ftilde$.  We have shown the following:

\medskip
\begin{theorem}
If $\fhat\in\ap$ is the minimal additive multiple of $f\in F[x]$ and
$\ftilde$ is a monic additive 
multiple of $f$, then $\fhat\apdivides\ftilde$.
\end{theorem}

Another characterisation of the minimal additive multiple of $f\in\monicpoly$
can be obtained by looking at the roots of $f$ in its splitting field $K$.
Assume $f$ is squarefree 
and has roots $\lbrace\theta_1,\ldots,\theta_n\rbrace$
and minimal additive multiple $h$. Then $f$
must have an additive multiple $g\in\bbA_K$ such that
\[
g= ((x^p-\theta_1^{p-1}x) \jn (x^p-\theta_2^{p-1}x)
\jn \cdots \jn (x^p-\theta_n^{p-1}x))\in\bbA_K.
\]
The polynomial $g$ is an additive multiple of $f$ because all roots of $f$ are
roots of $g$, and $g$ is additive.
Also, $g\apdivides h$ because  for each root $\theta_i$ of $f$, $k\theta_i$
must be a root of $h$ for each $k\in\ints_p$ and $1\leq i\leq n$.  The 
one dimensional vector space
$V_i=\lbrace k\theta_i\setmid k\in\ints_p\rbrace$ 
is a subspace of the kernel of $h$,
and the polynomial $g_i=x^p-\theta_i^{p-1}x$ has $V_i$ as its kernel.  So
$g_i\apdivides h$ for $1\leq i\leq n$ 
and therefore $g=(g_1\jn g_2\jn\cdots\jn g_n)\apdivides h$ as well
(see lemma 4.1(v)).
The coefficients
of $g$ are symmetric functions (over $F$) of the $\theta_i$'s
for $1\leq i\leq n$.  
Each automorphism of $K$ relative to 
$F$ carries the set of roots of $f$ into itself.
Thus, each automorphism leaves the coefficients of $g$ fixed, and
these coefficients must therefore be in $F$.
It follows that
$g\in F[x]$, and since $g$ right divides $h$ (the minimal additive multiple
of $f$), $g$ is equal to $h$.
This also
means that the exponent of the minimal additive multiple $\fhat$ is exactly
the dimension of the linear span of the roots of $f$ considered as a 
$\ints_p$ vector space in $K$.

An interesting case is when $f$ is a 
normal polynomial in $\ints_p[x]$ (a normal
polynomial is an irreducible polynomial such that its roots 
[in some fixed algebraic closure of $\ints_p$] form a basis over
$\ints_p$ for its splitting field).  The dimension of the $\ints_p$ vector
space spanned by the roots of $f$ is therefore the degree of $f$.  
It follows that the
normal polynomials of degree $n$ are exactly those irreducible polynomials
whose minimal additive multiples have exponent $n$.

\subsection{Complete Rational Decomposition of \mbox{Additive Polynomials}}

Assume the field $F$ supports a polynomial factorisation algorithm.
The preceding method for finding minimal additive multiples can be used
to build a polynomial time algorithm for finding rational complete normal
decompositions of additive polynomials in a polynomial number of
field operations.

First we present an algorithm for finding the set of indecomposable right
composition factors of $f\in\ap$.

\medskip
{\tt\obeylines
FindIndecRightFactors: $\ap\rightarrow\pwr(\ap)$

~~~~Input:~~-~$f\in\ap$, a monic additive polynomial.
~~~~Output:~-~$H=\lbrace h_1,\ldots,h_\ell\rbrace$, the set of indecomposable
~~~~~~~~~~~~~~right composition factors of $f$.

~~~~1) Factor $f$ such that $f=x^{e_0}h_1^{e_1}h_2^{e_2}\cdots h_m^{e_m}$ 
~~~~~~~where $h_i\in F[x]$ are distinct, monic and irreducible and
~~~~~~~$e_i\in\posnats$ for $0\leq i\leq m$.
~~~~2) Let $J:=\lbrace \hhat\setmid \hhat~\hbox{is the minimal additive multiple of}$
~~~~~~~~~~~~~~~~~~~~$h_i$ for some $i$ such that $1\leq i\leq\ell\rbrace$.
~~~~~~~Assume $J=\lbrace g_1,\ldots,g_\ell\rbrace$ for some $\ell\in\nats$
~~~~~~~and is indexed such that if $i<j$ then $\expn g_i\leq\expn g_j$
~~~~~~~for $1\leq i\leq\ell$.
~~~~3) For $1\leq i<j\leq\ell$, if $g_i\apdivides g_j$, mark $g_j$. 
~~~~4) Let $H=\lbrace g\in J\setmid g~\hbox{not marked in step 3}\rbrace$.
~~~~5) Return $H$.
}
\medskip

To show correctness 
we must prove that $g\in H$ if and only if $g$ is an indecomposable
right composition factor of $f$.
If $g\in\ap$ is an indecomposable right
factor of $f$, then, since $g$ is also a multiplicative factor of $f$, 
each irreducible multiplicative factor $h\in F[x]$ of $g$
is an irreducible multiplicative factor of $f$.  We know that the
minimal additive multiple $\hhat\in\ap$ of any such $h$ 
right divides $g$ by theorem 5.8, and
as $g$ is indecomposable, $\hhat=g$. 
Therefore $g$ will never be marked in step 3 and $g\in H$.  
Assume, on the other hand, that $g\in H$.  Suppose that 
$g$ is decomposable and $h$ is an
indecomposable right composition factor of $g$.  
Then $h$ is an indecomposable right factor of $f$ and $h\in H$ as shown above.
In step 3, $g$ would be marked
as decomposable  and would not be in $H$, a contradiction.
Therefore, each $g\in H$ is indecomposable and
the algorithm works correctly.

We now analyse the number of field operations required by the procedure
{\tt FindIndecRightFactors}.  For $f\in\ap$ of degree $n=p^\nu$
consider computing ${\tt FindIndec\-} 
{\tt RightFactors}(f)$.
The factorisation in step 1 requires $O(\sf(n))$ field operations.
In step 2 we find additive multiples of each of the indecomposable
right factors of $f$.  The worst case occurs  when there
is one factor of degree $n-1$.  
Thus, step 2 requires $O(n^3)$ field operations.
The number of operations required in the remaining steps is 
dominated by the requirements of steps 1 and 2, so we have the
following:

\medskip
\begin{lemma}
Given $f\in F[x]$ of degree $n$ we can compute all the
indecomposable right factors of $f$ in $O(\sf(n)+n^3)$
field operations.
\end{lemma}

Now consider the following algorithm for generating a complete
decomposition of $f$ in $\capdecall$.

\medskip
{\tt \obeylines
CompleteDecomposition:~$\ap\rightarrow\capdecall$

~~~~Input:~~-~$f\in\ap$, a monic additive polynomial.
~~~~Output:~-~a complete decomposition of $f$ in $\capdecall$.

~~~~1) Using FindRightIndecFactors, find the set $H$ of 
~~~~~~~indecomposable right factors of $f$.  Assume
~~~~~~~$H=\lbrace h_1,\ldots,h_\ell\rbrace$.
~~~~2) If $h_1=f$
~~~~~~~~~~then $f$ is indecomposable, Return $(f,(f))$.
~~~~~~~~~~else 
~~~~~~~~~~~~~2.1) Let $D:=$CompleteDecomposition$(f\apdiv h_1)$.
~~~~~~~~~~~~~~~~~~We know $D=(f\apdiv h_1,(u_t,\ldots,u_1))\in\capdecall$
~~~~~~~~~~~~~~~~~~for some $t\in\posnats$.
~~~~~~~~~~~~~2.2) Return $(f,(u_t,\ldots,u_1,h_1))$.
}
\medskip

At each recursive stage of the algorithm we simply determine one
indecomposable right factor $h_1$ of $f$.  We then proceed recursively
to find a complete decomposition of  $f\apdiv h_1$.  As $f$ has
exponent $\nu$, and each indecomposable right factor has exponent
at least one, there can be at most $\nu$ recursive stages.
We now analyse the number of field operations required to decompose a
polynomial $f\in\ap$ of degree $n=p^\nu$.
The worst case occurs when
a $p-linear$ (exponent one) 
right composition factor occurs at each recursive stage
$i$, for $1\leq i\leq \nu$.
In this case,
at stage $i$ we must call
{\tt FindRightIndecFactors} on a degree $p^i$ polynomial,
requiring $O(\sf(p^i)+p^{3i})$ field operations.
Thus, the total number of field operations required to find
one complete decomposition is	
\[
\eqalign{
\sum_{0\leq i\leq\nu} (\sf(p^i)+O(p^{3i}))
=    & O(\sf(p^\nu)+p^{3\nu})\cr
=    & O(\sf(n)+n^3).\cr
}
\]

\medskip
\begin{theorem}
Given an additive polynomial $f\in\ap$ of degree $n$ 
we can determine a complete
decomposition of $f$ in $\capdecall$ in $O(\sf(n)+n^3)$ field
operations.
\end{theorem}

\medskip
\begin{corollary}
Rational indecomposability of an additive polynomial of degree $n$ can
be determined in $O(\sf(n)+n^3)$ field operations.
\end{corollary}

\subsection{General Rational Decomposition of \mbox{Additive Polynomials}}

Let $f\in\ap$ be of degree $n$ and let $\wp=\ordfact$ be an ordered
factorisation of $n$.
The fact that we can obtain a complete decomposition of an $f\in\capdecall$ 
in a polynomial number of field 
operations in $n$ does not mean that we can determine the existence
of a decomposition of $f$
in $APDEC^F_\wp$, and find one if it exists, in polynomial time.
We can look at the set of all complete decompositions and check 
if the composition
factors of any of them can be ``grouped'' according to the desired ordered
factorisation $\wp$. More generally,
a length $d$
ordered factorisation $\wq=(s_d,s_{d-1},\ldots,s_1)$  of $n\in\nats$
is said to be a {\it refinement} of a length $m\leq d$ ordered
factorisation $\wp=\ordfact$ if there exists a non-decreasing, onto
map $\varphi:\{1,\ldots,d\}\rightarrow\{1,\ldots,m\}$ such that for 
$1\leq j\leq m$,
\[
\prod_{{1\leq i\leq d}\atop{\varphi(i)=j}}s_i=r_j.
\]
This is simply saying that the $d$-tuple $\wq$ can be divided into
$m$ contiguous pieces, with the elements of piece $j$ having product
$r_j$, for $1\leq j\leq m$.
One approach to finding decompositions of $f$ with a given ordered
factorisation is to 
generate the set of all complete decompositions of
$f$ and check if any of the ordered factorisations associated
with these decompositions are a refinement of $\wp$.
 
We now present an algorithm for generating all the
complete decompositions of an additive polynomial.

\pagebreak
{\tt \obeylines
AllCompleteDecomposition:~$\ap\rightarrow\pwr(\capdecall)$

~~~~Input:~~-~$f\in\ap$, a monic additive polynomial.
~~~~Output:~-~the set of all complete decompositions of $f$ 
~~~~~~~~~~~~~~in $\capdecall$.

~~~~1) Using FindRightIndecFactors, find the set $H$ of 
~~~~~~~indecomposable right factors of $f$.  Assume
~~~~~~~$H=\lbrace h_1,\ldots,h_\ell\rbrace$.
~~~~2) If $f=h_1$
~~~~~~~~~~Then $f$ is indecomposable, Return $(f,(f))$.
~~~~~~~~~~Else 
~~~~~~~~~~~~~2.1) Let $T:=\emptyset$.
~~~~~~~~~~~~~2.2) For $i$ from 1 to $\ell$
~~~~~~~~~~~~~~~~~~2.2.1) Let $D^{(i)}:=$\mbox{CompleteDecomposition$(f\apdiv h_i)$,}
~~~~~~~~~~~~~~~~~~~~~~~~~the set of all complete decompositions 
~~~~~~~~~~~~~~~~~~~~~~~~~of $(f\apdiv h_i)$ in $cAPDEC^F_*$.
~~~~~~~~~~~~~~~~~~2.2.2) For each decomposition 
~~~~~~~~~~~~~~~~~~~~~~~~~$(f\apdiv h_i,(u_\ell,\ldots,u_1))\in D^{(i)}$,
~~~~~~~~~~~~~~~~~~~~~~~~~add $(f,(u_\ell,\ldots,u_1,h_i))\in\capdecall$
~~~~~~~~~~~~~~~~~~~~~~~~~to $T$.
~~~~~~~~~~~~~2.3) Return T.
}
\medskip

Correctness is easy to verify.  At each recursive stage we simply
find the set of all indecomposable right factors $H$ of $f$ and for
each $h\in H$ we recursively find the complete decompositions of
$f\apdiv h$.  All complete decompositions are found
and, since we choose a different member of $H$ in each step 2.2.1,
each decomposition added to $T$ is distinct.

We analyse the cost of the algorithm by first finding the cost of
computing one complete decomposition. We then use the bounds
developed in chapter 3 and 4 on the number of complete decompositions
to get bounds on the cost of computing all complete decompositions.
As with {\tt CompleteDecomposition}, the worst case occurs when
a $p-linear$ right composition factor occurs at each recursive stage
$i$, for $1\leq i\leq \nu$.
In this case,
at stage $i$, we must call
{\tt FindRightIndecFactors} on a degree $p^i$ polynomial,
requiring $O(\sf(p^i)+p^{3i})$ field operations.
Thus the total number of field operations required to find
one complete decomposition is	
\[
\eqalign{
\sum_{0\leq i\leq\nu} \sf(p^i)+O(p^{3i})
=    & O(\sf(p^\nu)+p^{3\nu})\cr
=    & O(\sf(n)+n^3).\cr
}
\]

\medskip
\begin{theorem}
If $f\in\ap$ is monic of degree $n$ and $t\in\nats$,
then we can 
determine if there exist $t$ decompositions of $f$ in $cAPDEC^F_*$,
and if so find them,
with $O(t(\sf(n)+n^3))$ field operations.
\end{theorem}

By theorem 4.26, the total number of complete normal decompositions
of $f$ in $\capdecall$ is $n^{O(\log n)}$, so we can compute all complete
rational decompositions of an arbitrary additive polynomial in a 
quasi-polynomial number of field operations.

Let $\wp$ be a given ordered factorisation of $n$.
As we generate each complete decomposition of $f$, we can check if the
ordered factorisation associated with it is a refinement of $\wp$.
The number of operations required to do this is dominated by the
other steps in the algorithm.
Thus, the number of operations required to find all decompositions
of $f$ in $DEC^F_\wp$ is of the same order as the number required
to generate all complete decompositions.

\medskip
\begin{corollary}
If $f\in\ap$ is monic of degree $n$, and $\wp$ is an ordered
factorisation of $n$, then all decompositions of $f$ in $APDEC^F_\wp$
can be found in $n^{O(\log n)}$ field operations.
\end{corollary}

Note that this algorithm requires a comparable number of operations to
those of Kozen and Landau[1986] for separable irreducible polynomials.

\subsection{General Decomposition of Completely Reducible Additive Polynomials}

We now consider computing decompositions of a completely reducible
additive polynomial $f\in\ap$ of degree $n$ 
corresponding to a given ordered factorisation
$\wp$ of $n$. 
We will see that the decomposition problem for completely reducible 
additive polynomials can be computed in a polynomial number of field operations
in the input degree.
We proceed by constructing
an indecomposable basis for $f$ (see chapter 4, section D) and combine
the basis components appropriately to determine if an appropriate
decomposition exists, and if so, find it.

We now describe an efficient way of computing an indecomposable basis for a
given completely reducible additive polynomial.
The procedure strongly resembles the one described in the proof
of lemma 4.18.

\medskip
{\tt\obeylines
IndecBasis: $\ap\rightarrow\ap^*$

~~~~Input:~~ $f\in\ap$, completely reducible of degree $n=p^\nu$.
~~~~Output:~ $u_1,u_2,\ldots,u_d\in\ap$, an indecomposable basis for $f$.

~~~~1) Using FindIndecRightFactors, find the set 
~~~~~~~$R=\{v_1,v_2,\ldots,v_\ell\}$ of indecomposable right factors of $f$.
~~~~2) Let $j:=0$.
~~~~3) Let $g^{(0)}:=x$.
~~~~4) For $i$ from $1$ to $\ell$ do
~~~~~~~~4.1) If $v_i\mt g^{(i-1)}=x$ then
~~~~~~~~~~~~~~4.1.1) Let $g^{(i)}:=g^{(i-1)}\jn v_i$.
~~~~~~~~~~~~~~4.1.2) Let $j:=j+1$.
~~~~~~~~~~~~~~4.1.3) Let $u_j:=v_i$.
~~~~~~~~~~~~~Else
~~~~~~~~~~~~~~4.1.4) Let $g^{(i)}:=g^{(i-1)}$.
~~~~~~~~4.2) If $g^{(i)}=f$ then quit,
~~~~~~~~~~~~~~~~returning $u_1,\ldots,u_j$ as an indecomposable basis.
}
\medskip

Since we know $f$ is completely reducible, $f$ is, by definition, the join
of its indecomposable right factors $v_1,\ldots,v_\ell$.  The algorithm simply
looks at each indecomposable right factor in turn.  At step 4.1, either
$v_i\mt g^{(i-1)}=x$ or $v_i\jn g^{(i-1)}=g^{(i-1)}$.  Only in the first
case does $v_i$ contribute anything to the join of all the right factors,
and, in this case, $\expn(g^{(i)})=\expn(g^{(i-1)})+\expn(v_i)$. 
The set of all such contributing right factors clearly forms an indecomposable
basis for $f$.  The cost of this algorithm is dominated by the cost of
finding the set of indecomposable right factors of $f$.
We have shown the following:

\medskip
\begin{lemma}
Let $f\in\ap$ be completely reducible of degree $n$.  We can find an 
indecomposable basis for $f$ with $O(\sf(n)+n^3)$ field operations.
\end{lemma}

Let $f\in\ap$ be completely reducible of degree $n=p^\nu$ and let 
$\wp=(p^{\rho_m},\ldots,p^{\rho_1})$ be an 
ordered factorisation of $n$.  
We now address the problem of finding a
decomposition of $f$ in $APDEC^F_\wp$.  It is true in general that there
exists a decomposition of $f$ in $APDEC^F_\wp$ if and only if there exists
an ordered factorisation 
$\wq=(p^{\sigma_d},p^{\sigma_{d-1}},\ldots,p^{\sigma_1})$ of $n$ which is
a refinement of $\wp$ such that $f$ has a complete decomposition in
$cAPDEC^F_\wq$.  Since all completely reducible additive polynomials are
completely transmutable, we need only determine if some permutation of the
ordered factorisation $\wq$ 
is a refinement of $\wp$.  By theorem 4.20, the composition
factors of an arbitrary complete decomposition and the members of an arbitrary
indecomposable basis of $f$ are similar in pairs.  
Assume $u_1,\ldots,u_d$ forms an indecomposable basis for $f$, where
$\deg u_i=p^{e_i}$ for $1\leq i\leq d$.  Then $f$ has a decomposition
in $APDEC^F_\wp$ if an only if some permutation of 
$\mu=(p^{e_d},p^{e_{d-1}},\ldots,p^{e_1})$ is a refinement of $\wp$.
This is equivalent to saying that some permutation of $\mu$ is a refinement
of some permutation of $\wp$ -- we do not need to consider the order 
of $\wp$ either.
In light of this, we denote an  {\it unordered factorisation} 
of $n$ of length $m$ as 
$[a_m,a_{m-1},\ldots,a_1]$ where $a_i\in \nats$ for $1\leq i\leq m$, 
$\prod_{1\leq i\leq m}a_i=n$ and for any
permutation $\tau$ of $\{1\ldots m\}$, 
$[a_{\tau_m},a_{\tau_{m-1}},\ldots,a_{\tau_1}]=[a_m,a_{m-1},\ldots,a_1]$.  
Such a
data structure can be easily managed computationally and the details will
be left to the reader (for instance, one could manage them as sorted
$m$ tuples).
Basic operations on an unordered factorisation of length $\ell$, 
such as assignment and equality test, will be assumed to require $\ell^{O(1)}$
field operations.
Let $\wpbar=[p^{\rho_m},p^{\rho_{m-1}},\ldots,p^{\rho_1}]$ 
and $\mubar=[p^{e_d},p^{e_{d-1}},\ldots,p^{e_1}]$ be the unordered
factorisations corresponding with $\wp$ and $\mu$ respectively.
A length $d$ unordered factorisation 
$\gamma=[p^{\sigma_d},p^{\sigma_{d-1}},\ldots,p^{\sigma_1}]$ 
is an {\it unordered refinement}
of $\wpbar$ if there is an onto map 
$\psi:\{1,\ldots,d\}\rightarrow\{1,\ldots m\}$
such that for $1\leq j\leq m$,
\[
\prod_{{1\leq i\leq d}\atop{\psi(i)=j}}\sigma_i=r_j.
\]
We proceed by generating the set $L$ of 
all length $m$ unordered factorisations $\lambda$
of $n$ which are unordered refinements of $\mubar$. For each 
$\lambda\in L$ we keep exactly one refinement $\psi_\lambda$ from
$\mubar$ to $\lambda$, ignoring other such 
refinements. We show that $L$ is in fact
small and can be computed in time polynomial in $n$.   Once $L$ is computed,
it is easy to check if $\wpbar$ is in $L$.  If it is, then $f$ has a 
decomposition in $APDEC^F_\wp$ and it is a simple matter to recover this
decomposition from the refinement.

We proceed by dynamic programming.  
We define the $d\times m$ array $S$ 
of sets of unordered factorisations as follows.  For $1\leq i\leq d$
and $1\leq j\leq m$, let $S_{ij}$ be the set of unordered factorisations
of length $j$ of $p^{d_i}$ (where $d_i=\sum_{1\leq k\leq i} e_i$)
which are unordered refinements of $[p^{e_i},p^{e_{i-1}},\ldots,p^{e_1}]$.
The following algorithm
exploits an easy recurrence to generate all of $S$.

\pagebreak
{\tt\obeylines
FindUnorderedFacts: $\nats^*\times\nats\rightarrow (\pwr(\nats^*))^*$

~~~~Input:~~-~$\mu=(p^{e_d},p^{e_{d-1}},\ldots,p^{e_1})$, an ordered
~~~~~~~~~~~~~~factorisation of $n$,
~~~~~~~~~~~~-~$m\in\nats$, an integer at most $d$.
~~~~Output:~-~$S$, a $d\times m$ array of sets of unordered 
~~~~~~~~~~~~~~factorisations of $n$ as described above.

~~~~1) $S_{11}:=(p^{e_1})$.
~~~~~~~For $i$ from $2$ to $d$ 
~~~~~~~~~~~For $j$ from $1$ to $m$ 
~~~~~~~~~~~~~~~2) $S_{ij}:=\emptyset$.
~~~~~~~~~~~~~~~3) For each unordered factorisation
~~~~~~~~~~~~~~~~~~$[p^{a_{j-1}},\ldots,p^{a_1}]\in S_{i-1,j-1}$
~~~~~~~~~~~~~~~~~~add $[p^{a_{j-1}},\ldots,p^{a_1},p^{e_i}]$ to $S_{ij}$.
~~~~~~~~~~~~~~~4) For each unordered factorisation 
~~~~~~~~~~~~~~~~~~$[p^{a_j},\ldots,p^{a_1}]\in S_{i-1,j}$ and %
for $1\leq k\leq j$
~~~~~~~~~~~~~~~~~~add $[p^{a_j},\ldots,p^{a_{k+1}},p^{a_k}p^{e_i},%
p^{a_{k-1}},\ldots,p^{a_1}]$ to $S_{ij}$.
}
\medskip

Certainly, at the conclusion $S_{dm}$ contains the desired set of unordered 
factorisations.  The number of unordered factorisations which
are unordered refinements of $\mubar$ is at most the number of additive
partitions $p(\nu)$ of $\nu$ (the exponents of $p$ in the
unordered factorisation give a partition of $\nu$).  
Hua[1982] (theorem 6.1) shows that
\[
\eqalign{
p(\nu)\leq & \nu^{3\lfloor \sqrt{\nu}\rfloor}\cr
      \leq & \nu^{3\sqrt{\nu}+3}\cr
      \leq & (2^{\log\nu})^{3\sqrt{\nu}+3}\cr
      \leq & 2^{6\sqrt{\nu}\log\nu}.\cr
}
\]
Thus the total algorithm can be completed in 
\[
\eqalign{
dm\nu^{O(1)} 2^{6\sqrt{\nu}\log\nu} 
= & \nu^{O(1)} 2^{6\sqrt{\nu}\log\nu}\cr
= & O(n)\cr
}
\]
field operations.  
By keeping the products $p^{a_k}p^{e_i}$ in step 4
in an ``unevaluated'' form (or, alternatively, keeping some record of
the multiplicands) for each $\lambda\in L$,
we can easily recover an explicit unordered refinement $\psi_\lambda$
from $\mubar$ to $\lambda$. 
By checking if $\wpbar$ is 
in $S_{dm}$, we can determine
if $\mubar$ is an unordered refinement of $\wpbar$, and, if it is,
actually determine the refinement $\psi$.

Assume $\wpbar\in S_{dm}$ and $\psi$ is an unordered refinement
from $\mubar$ to $\wpbar$.  For $1\leq j\leq m$, let
\[
h_j=\bigsqcup_{{1\leq i\leq d}\atop {\varphi(i)=j}} u_i.
\]
Then for $1\leq i\leq m$, $\deg h_i=r_i$, $h_1,h_2,\ldots,h_m$ are
pairwise composition coprime, and $h_1\jn\cdots\jn h_m=f$.  The following
simple procedure can be used to recover a 
decomposition of $f$ in $APDEC^F_\wp$.

\medskip
{\tt\obeylines
BasisToDec: $\ap^*\rightarrow APDEC^F_*$

~~~~Input:~~-~$h_1,\ldots,h_m\in\ap$ such that $f=h_1\jn h_2\jn \cdots\jn h_m$,
~~~~~~~~~~~~~~$h_i\mt h_j=x$ for $1\leq i<j\leq m$ 
~~~~~~~~~~~~~~and $\deg h_i=r_i$ for $1\leq i\leq m$.
~~~~Output:~-~$(f,(f_m,f_{m-1},\ldots,f_1))\in APDEC^F_\wp$ 
~~~~~~~~~~~~~~where $\wp=(r_m,r_{m-1},\ldots,r_1)$.

~~~~Let $g^{(0)}:=x$.
~~~~For $1\leq i\leq m$
~~~~~~~~Let $g^{(i)}:=g^{(i-1)}\jn h_i$.
~~~~~~~~Let $f_i:=g^{(i)}\apdiv g^{(i-1)}$.
~~~~Return $(f,(f_m,f_{m-1},\ldots,f_1))\in APDEC^F_\wp$.
}
\medskip

This procedure can certainly be completed in $O(n^3)$ field operations.  We
have now completed the description of a general decomposition
algorithm for completely reducible additive polynomials and have shown
the following theorem:

\medskip
\begin{theorem}
Given $f\in\ap$ completely reducible of degree $n$ and $\wp$ an ordered
factorisation of $n$, we can determine if $f$ has a decomposition in
$APDEC^F_\wp$, and if so find one, in $O(\sf(n)+n^3)$ field operations.
\end{theorem}

\subsection{Determining Transmutations of \mbox{Additive Polynomials}}

Another approach to finding decompositions of additive polynomials
is to find one complete decomposition and then, using the
relationship between decompositions (developed in chapter 4),
produce a decomposition into factors of the desired degrees.

To do this we must be able to determine if two polynomials
$f,g\in\ap$ are transmutable, and find the set 
\[
\lbrace (\gbar,\fbar)\in\ap\times\ap \setmid \fbar\sim f,~f=g\trans\fbar, 
\gbar=\fbar\trans g\rbrace,
\]
of possible
transmutations of $f$ by $g$.  
The following algorithm performs this task if $f$ is
indecomposable in a polynomial number of operations in the sum of the
degrees of $f$ and $g$.

\medskip
{\tt\obeylines
Transmutable: $\ap\times\ap\rightarrow\pwr(\ap\times\ap)$
~~~~Input:~~-~$f\in\ap$, monic and indecomposable, $g\in\ap$, monic.
~~~~Output:~-~$T=\lbrace (\gbar,\fbar)\in\ap\times\ap \setmid \fbar\sim f,f=g\trans\fbar, \gbar=\fbar\trans g\rbrace$.

~~~~1) Using FindIndecRightFactors, find the set $H\subseteq\ap$ of 
~~~~~~~indecomposable right factors of $f\circ g$.
~~~~2) Let $J:=\lbrace \fhat\in H \setmid \expn \fhat=\expn f\rbrace\subseteq H$.
~~~~3) Let $T:=\emptyset$.
~~~~4) For each $\fhat\in J$ 
~~~~~~~~4.1) Let $\ghat:=(f\circ g)\apdiv\fhat$.
~~~~~~~~4.2) If $\ghat=\fhat\trans g$ then let $T:= T\cup (\ghat,\hhat)$.
~~~~5) Return $T$.
}
\medskip

A transmutation of $f$ by $g$ will transform $f$ into a similar
polynomial $\fhat\in\ap$ which is a right factor of $f\circ g$. 
Therefore, we eliminate all the $\fhat\in H$ with exponents unequal
to that of $f$ in step 2.  Now, for any $\ghat,\fhat\in\ap$ such that
$f\circ g=\ghat\circ\fhat$ and $\ghat=\fhat\trans g$,
we know
\[
\eqalign{
f\circ g = & (\fhat\trans g)\circ \fhat\cr
         = & \fhat \jn g\cr
         = & ((g \jn \fhat)\apdiv g)\circ g\cr
         = & (g\trans \fhat)\circ g.\cr
}
\]
It follows that $f=g\trans\fhat$, and since $f$ and $\fhat$ have the same
exponent, $g$ and $\fhat$ are composition coprime and $f$ transmutes by $g$.

\medskip
\begin{theorem}
The set of all  transmutations of an indecomposable additive 
polynomial $f\in\ap$
by an arbitrary additive polynomial $g\in\ap$, where the degree of $f\circ g$
is $n=p^\nu$, can be computed in $O(\sf(n)+n^3)$ field operations.
\end{theorem}
\begin{proof}
Determining the set of indecomposable right factors in step 1 requires
$O(\sf(p^\nu)+p^{3\nu})=O(\sf(n)+n^3)$.
Step 4.1
require $O(n)$ exponent $\nu$ divisions, and $O(n\nu^2\log p)$ 
field operations.
Finally, step 4.2 requires $O(n\nu^3\log p)$ field operations.   Thus the
total number of field operations required is dominated by 
the number required for step 1 and is $O(\sf(n)+n^3)$.~\QED
\end{proof}

Suppose $g\in\ap$ in the above algorithm is given as a complete 
decomposition.  Then,
if $f$ transmute by $g$, $f=g\trans \fbar$ where $\fbar\in\ap$ and 
$\fbar\sim f$.  It follows that 
$f\circ g=\gbar\circ\fbar$ where $\gbar=\fbar\trans g$.
We would like to give the corresponding decomposition of $\gbar$.  By
theorem 4.13 we can compute the effect of transformation of a composition.
The following algorithm performs this task efficiently.

\medskip
{\tt\obeylines
TransformComposition: $\ap\times cAPDEC^F_*\rightarrow cAPDEC^F_*$

~~~~Input:~~-~$h\in\ap$, monic and indecomposable,
~~~~~~~~~~~~-~$(g,(g_m,g_{m-1},\ldots,g_1))\in\capdecall$.
~~~~Output:~-~$(\gbar,(\gbar_m,\gbar_{m-1},\ldots,\gbar_1))\in\capdecall$ where
~~~~~~~~~~~~~~$\gbar=h\trans g$ and $\gbar_i\in\ap$, $\gbar_i\sim g_i$ for $1\leq i\leq m$.

~~~~$h_1:=h$.
~~~~$\gbar_1:=h_1\trans g_1$.
~~~~For $2\leq i\leq m$
~~~~~~~~$h_i:=(g_{i-1}\circ g_{i-2}\circ \cdots\circ  g_1)\trans h$.
~~~~~~~~$\gbar_i:=h_i\trans g_i$.
~~~~Return $(h\trans g,(\gbar_m,\gbar_{m-1},\ldots,\gbar_1))$.
}
\medskip

\noindent
Correctness follows immediately as the algorithm is simply a direct 
application of theorem 4.13.
If $g\in\ap$ and $h\in\ap$ are of exponents $\rho$ and $\sigma$
respectively and $\delta=\max(\rho,\sigma)$, then computing $h_i\in\ap$
in the algorithm requires $O(\delta^3\log p)$ field operations for
each $i$ with $1\leq i\leq m$.  Computing $\gbar_i$ also requires
$O(\delta^3\log p)$ field operations for each $i$ with $1\leq i\leq m$.
We know that $m\leq\delta$, so we get the following theorem:

\medskip
\begin{theorem}
If $(g,(g_m,g_{m-1},\ldots,g_1))\in cAPDEC^F_*$ 
where $g\in\ap$ has exponent $\rho$, and
$h\in\ap$ has exponent $\sigma$, 
then we can transform the decomposition
of $g$ into a corresponding decomposition of $h\trans g$ in
$O(\delta^4\log p)$ field operations, where $\delta=\rho+\sigma$.
\end{theorem}

\subsection{Bidecomposition of Similarity Free \mbox{Additive Polynomials}}

We now describe an algorithm for finding a bidecomposition of
a similarity free additive polynomial $f\in\ap$ of degree $n=p^\nu$
corresponding to an ordered factorisation $\wp=(p^\rho,p^\sigma)$.  
We will see this can be done in a polynomial number of field operations
in the degree of the input polynomial.  Using the
algorithm {\tt CompleteDecomposition}, we can find a complete
decomposition $(f,(f_m,f_{m-1},\ldots,f_1))\in\capdecall$ with
$O(\sf(n)+n^3)$ field operations.  We proceed by looking at each
subset $S\subseteq \lbrace 1,\ldots,m\rbrace$ such
that $\sum_{i\in S}\expn f_i=\sigma$.  Assume $S$ has cardinality $t\in\nats$
where $t\geq 1$.
For each $S$ we determine
if there exists a decomposition $(f,(g_m,\ldots,g_1))\in\capdecall$ 
and a bijection $\varphi$ between $\lbrace 1,\ldots,t\rbrace$
and $S$ with $g_i\sim f_{\varphi(i)}$ for $1\leq i\leq t$.  
Say a decomposition with this property is {\it consistent} with $S$
and $(f,(f_m,f_{m-1},\ldots,f_1))$.
In other words, there is a decomposition such that the rightmost $t$
composition factors are similar  in pairs to the composition 
factors indexed by $S$.
Because
we are assuming $f$ is similarity free, all transmutations are unique
(see chapter 4, section E).

In the following algorithm we determine if a decomposition consistent with
a given $S$ of size $t$ and $(f,(f_m,f_{m-1},\ldots,f_1))\in\capdecall$ exists.
The algorithm proceeds in stages $\ell$, for $1\leq \ell\leq t$.  
Assume 
$(f_m^{(0)},f_{m-1}^{(0)},\ldots,f_1^{(0)})=(f_m,f_{m-1},\ldots,f_1)$.
At each stage
$\ell$ we transmute one of the factors of \linebreak
$(f,(f_m^{(\ell-1)},f_{m-1}^{(\ell-1)},\ldots,f_1^{(\ell-1)}))$
which is similar to a factor indexed in $S$ into the
the $\ell^{\rm th}$ composition factor position from the right, obtaining
a new decomposition $(f,(f_m^{(\ell)},f_{m-1}^{(\ell)},\ldots,f_1^{(\ell)}))$.
We keep track of where the factors of the original decomposition have
been transmuted to at the end of stage $\ell$ by means of an index vector
$c^{(\ell)}=(c_m^{(\ell)},\ldots,c_1^{(\ell)})$.  At this point,
$f_j^{(\ell)}\sim f_{c_j^{(\ell)}}$ for each $j$ such that $1\leq j\leq m$. 
The decomposition produced at the end of 
stage $\ell$ will have the property that for each 
$j\in\nats$ such that $1\leq j\leq \ell$, $c_j^{(\ell)}\in S$.  If at each
stage such a decomposition can be found, at stage $t$ we will have
a decomposition of $f$ consistent with $S$ and
$(f,(f_m,f_{m-1},\ldots,f_1))$.

\medskip
{\tt\obeylines
FactorsToRight: $\capdecall\times\pwr(\nats)\rightarrow\capdecall$
~~~~Input:~~-~$(f,(f_m,f_{m-1},\ldots,f_1))\in\capdecall$,
~~~~~~~~~~~~-~$S\subseteq\nats$ of cardinality $t\in\nats$.
~~~~Output:~-~$(f,(g_m,g_{m-1},\ldots,g_1))\in\capdecall$ consistent
~~~~~~~~~~~~~~with $(f,(f_m,f_{m-1},\ldots,f_1))$ and $S$ (if such a
~~~~~~~~~~~~~~decomposition exists).

~~~~1) Let $c^{(0)}:=(c_m^{(0)},c_{m-1}^{(0)},\ldots,c_1^{(0)}):=(m,m-1,\ldots,1)$.
~~~~2) Let $S^{(0)}:=S$.
~~~~3) For $\ell$ from $1$ to $t$
~~~~~~~3.1) For each $i\in S^{(\ell-1)}$
~~~~~~~~~~~~3.1.1) Let $k\in\nats$ be such that $c_k^{(\ell-1)}=i$.
~~~~~~~~~~~~3.1.2) Using Transmutable, determine if $f_k^{(\ell-1)}$
~~~~~~~~~~~~~~~~~~~transmutes by $f_{k-1}^{(\ell-1)}\circ\cdots\circ f_\ell^{(\ell-1)}$.
~~~~~~~~~~~~~~~~~~~If so, goto step 3.3.
~~~~~~~3.2) No transmutation found, quit.
~~~~~~~3.3) Using TransformComposition, for $\ell\leq j\leq k$
~~~~~~~~~~~~find $\fbar^{(\ell-1)}_j\sim f_j^{(\ell-1)}$ such that
~~~~~~~~~~~~$f_k^{(\ell-1)}\circ f_{k-1}^{(\ell-1)}\circ\cdots\circ f_\ell^{(\ell-1)}=\fbar^{(\ell-1)}_{k-1}\circ\fbar^{(\ell-1)}_{k-2}\circ\cdots\circ\fbar^{(\ell-1)}_\ell\circ\fbar^{(\ell-1)}_k$
~~~~~~~~~~~~(ie. compute the transmutation of $f_k^{(\ell-1)}$ by
~~~~~~~~~~~~$\fbar^{(\ell-1)}_{k-1}\circ\fbar^{(\ell-1)}_{k-2}\circ\cdots\circ\fbar^{(\ell-1)}_\ell)$.
~~~~~~~3.4) Let $(f_m^{(\ell)},\ldots,f_1^{(\ell)}):=(f_m^{(\ell-1)},\ldots,f_{k+1}^{(\ell-1)},\fbar^{(\ell-1)}_{k-1},\ldots,\fbar^{(\ell-1)}_\ell,$
~~~~~~~~~~~~~~~~~~~~~~~~~~~~~~~~~~~~~~~$\fbar^{(\ell-1)}_k,f^{(\ell-1)}_{\ell-1},\ldots,f_1^{(\ell-1)})$.
~~~~~~~3.5) Let $c^{(\ell)}:=(c_m^{(\ell-1)},\ldots,c_{k+1}^{(\ell-1)},c_{k-1}^{(\ell-1)},\ldots,c_\ell^{(\ell-1)},$
~~~~~~~~~~~~~~~~~~~~~~~~~~~~~~~~~~~~~~~$c_k^{(\ell-1)},c_{\ell-1}^{(\ell-1)},\ldots,c_1^{(\ell-1)})$.
~~~~~~~3.6) Let $S^{(\ell)}:=S^{(\ell-1)}-\lbrace i\rbrace$.
~~~~4) Return $(f,(f_m^{(t)},f_{m-1}^{(t)},\ldots,f_1^{(t)}))$.
}
\medskip  

We now show the correctness of the above algorithm.
If there exists no decomposition of $f$ consistent with
$(f,(f_m,f_{m-1},\ldots,f_1))$ and $S$, then {\tt FactorsToRight} will
obviously not find it.   

\medskip
\begin{lemma}
Let $S$ be a subset of $\{1,\ldots m\}$.
If there exists a decomposition of $f$
consistent with $(f,(f_m,f_{m-1},\ldots,f_1))$ and $S$,
{\tt FactorsToRight} will find one.
\end{lemma}
\begin{proof}
We prove this lemma by induction on $t$, the cardinality of $S$.  
For the basis step, $t=1$ and
$S=\lbrace i\rbrace$ for some $i$ such that $1\leq i\leq m$.
Assume that $(f,(g_m,g_{m-1},\ldots,g_1))\in\capdecall$ is a decomposition
of $f$ consistent with $(f,(f_m,f_{m-1},\ldots,f_1))$ and $S$.
In step 3.1.1, $k=i$.  We know that $g_1\apdivides f$.  Let $j\in\nats$
be the smallest number such that 
$g_1\apdivides f_j^{(0)}\circ\cdots\circ f_1^{(0)}$.
The polynomials $g_1$ and $f_{j-1}^{(0)}\circ\cdots\circ f_1^{(0)}$ 
are composition-coprime,
so by theorem 4.7, 
$f_j^{(0)}=(f_{j-1}^{(0)}\circ\cdots\circ f_1^{(0)})\trans g_1$, 
$f_j^{(0)}\sim g_1$, and $f_j^{(0)}$ transmutes by
$f^{(0)}_{j-1}\circ\cdots\circ f_1^{(0)}$.  Since $f$ is
assumed to be similarity free, $j=k$ and the transmutation of step 3.1
gives a decomposition consistent with $(f,(f_m,f_{m-1},\ldots,f_1))$ and $S$.

Now assume that {\tt FactorsToRight} finds a decomposition of $f$
consistent with $(f,(f_m,f_{m-1},\ldots,f_1))$ and $S$
if the cardinality of $S$ is less than $t$.  We must show it does so for
$S$ of cardinality $t$ as well.

Assume $S$ has cardinality $t$ and
that $(f,(g_m,g_{m-1},\ldots,g_1))\in\capdecall$ is a decomposition
of $f$ consistent with $(f,(f_m,f_{m-1},\ldots,f_1))$ and $S$.
With $\ell=1$ we know there exists a $k\in S$ such that
$f^{(0)}_k\sim g_1$ and by the argument for the basis case, 
$f_k^{(0)}=(f_{k-1}^{(0)}\circ\cdots\circ f_1^{(0)})\trans g_1$ and
$f_k^{(0)}$ transmutes by $f_{k-1}^{(0)}\circ\cdots\circ f_1^{(0)}$.
The algorithm may or may not transmute $f_k^{(0)}$ to the right of
the decomposition, depending on the choice in step 3.1.
Assuming we do choose this $f_k^{(0)}$ to transmute in step 3.1, we get
\[
\eqalign{
(f_m^{(1)},\ldots,f_1^{(1)})
= & (f_m^{(0)},\ldots,f_{k+1}^{(0)},\fbar^{(0)}_{k-1},\ldots,\fbar^{(0)}_1,\fbar^{(0)}_k)\cr
= & (f_m^{(0)},\ldots,f_{k+1}^{(0)},\fbar^{(0)}_{k-1},\ldots,\fbar^{(0)}_1,g_1)\cr
}
\]
where $\fbar^{(0)}_i\sim f_i$ for $1\leq i\leq k$ (we know
$f_k^{(0)}=g_1$ because $f$ is similarity free).
As $g_1$ is never referenced in the computation again, the remainder of
the algorithm is essentially finding a decomposition of $f\apdiv g_1$
(which has decomposition 
$(f\apdiv g_1,(f_m^{(1)},\ldots,f_2^{(1)}))\in\capdecall$) that is
consistent with $(f_m^{(1)},f_{m-1}^{(1)},\ldots,f_2^{(1)})$ and 
$S-\lbrace i\rbrace$.
Since $S-\lbrace i\rbrace$ has cardinality less than $t$,
{\tt FactorsToRight} finds such a decomposition by the inductive
hypothesis.

Suppose, however, that in step 3.1, with $\ell=1$, we transmute
$f_w^{(0)}\nsim g_1$, for some $w$ such that $1\leq w\leq m$,
$w\in S$ and $w\neq k$, 
to the right.  Then $f_w\sim g_j$ for some $j\leq t$.  Since
we know $\fbar^{(0)}_w\apdivides f$, $g_j$ transmutes by
$g_{j-1}\circ\cdots\circ g_1$.  Assume 
\[
g_j\circ (g_{j-1}\circ\cdots\circ g_1)
=\gbar_{j-1}\circ\cdots\circ\gbar_1\circ\gbar_j
\] 
where $g_v\sim\gbar_v$ for $1\leq v\leq j$.  Then 
\[
(f,(g_m,g_{m-1},\ldots,g_{j+1},\gbar_{j-1},\ldots,\gbar_1,\gbar_j))
\in\capdecall
\]
must be another decomposition of $f$ consistent with
$(f,(f_m,f_{m-1},\ldots,f_1))$ and $S$.  Since $f$ is
similarity free, $\gbar_j=\fbar^{(0)}_w$.  By the argument
for the case when $f_k^{(0)}$ was chosen in step 3.1,
{\tt FactorsToRight} finds a decomposition consistent with
$(f,(f_m,f_{m-1},\ldots,f_1))$ and $S$.
\QED
\end{proof}

In {\tt FactorsToRight} we execute $t\leq\nu$ iterations of the main loop
in step 3.  In iteration $i$ for $1\leq i\leq t$, step 3.1
will require up to $i$ transmutations of additive polynomials of
exponents at most $\nu-i$.  This requires
$O(i(\sf(p^{\nu-i})+p^{3i}))$ field operations.
Transforming the factors of this transmutation in step 3.3 using 
{\tt TransformComposition} requires $O((\nu-i)^3\log p)$ field operations.
The total number of field operations required in iteration $i$ is
\[
\eqalign{
&O(i\sf(p^{\nu-i})+ip^{3i}+(\nu-i)^3\log p)\cr
= &O(i\sf(p^{\nu-i})+ip^{3i}).\cr
}
\]
The number of field operations required for all $t=O(\log n)$ 
iterations is therefore
\[
\eqalign{
&\sum_{0\leq i\leq t} O(i\sf(p^{\nu-i})+ip^{3i})\cr
=&O(\nu\sf(p^\nu)+\nu p^{3\nu})\cr
=&O(\sf(n)\log n+n^3\log n).\cr
}
\]
We can now write a complete algorithm for the 
bidecomposition of similarity free
additive polynomials.

\medskip
{\tt\obeylines
SimFreeBidecomp: $\ap\times\nats^2\rightarrow\capdecall$
~~~~Input:~~-~$f\in\ap$, similarity free of degree $n=p^\nu$.
~~~~Output:~-~$(p^\rho,p^\sigma)$, an ordered factorisation of $n$.

~~~~1) Using CompleteDecomposition, attempt to find a 
~~~~~~~decomposition $(f,(f_m,f_{m-1},\ldots,f_1))\in\capdecall$.
~~~~2) For each subset $S$ of $\lbrace 1,\ldots,m\rbrace$
~~~~~~~~~2.1) if $\sum_{i\in S}\expn f_i=\sigma$, find a decomposition
~~~~~~~~~~~~~~$(f,(g_m,g_{m-1},\ldots,g_1))$ consistent with
~~~~~~~~~~~~~~$(f,(f_m,f_{m-1},\ldots,f_1))$ and $S$ using FactorsToRight.
~~~~~~~~~~~~~~If such a decomposition is found, goto step 4.
~~~~3) There is no decomposition of $f$ in $cAPDEC^F_{(p^\rho,p^\sigma)}$, quit.
~~~~4) Let $k$ be such that $g_m\circ g_{m-1}\circ\cdots\circ g_k$ has
~~~~~~~exponent $\rho$.
~~~~5) Return $(f,((g_m\circ g_{m-1}\circ\cdots\circ g_k), (g_{k-1}\circ g_{k-2}\cdots\circ g_1)))$.
}
\medskip

There are at most $2^\nu=O(n)$ subset $S$ of $\lbrace 1,\ldots, m\rbrace$
so the total number of field operations required is
\[
O(\sf(n)n\log n+n^4\log n).
\]
We have shown the following theorem:

\medskip
\begin{theorem}
Let $f\in\ap$ be similarity free of degree $n=p^\nu$.  Let $(p^\rho,p^\sigma)$
be an ordered factorisation of $n$.  Using {\tt SimFreeBidecomp}
we can determine if there 
exists a decomposition of $f$ in $APDEC^F_{(p^\rho,p^\sigma)}$,
and if so find one, with $O(\sf(n)n\log n+n^4\log n)$ field operations.
\end{theorem}

\subsection{Absolute Decompositions of Additive Polynomials}

As noted in chapter 3, additive polynomials decompose into $p$-linear
factors over their splitting fields.  We will now show how to compute
an arbitrary decomposition of an additive polynomial $f\in F[x]$
in an algebraic closure $\Fbar$
of $F$ (an absolute, complete decomposition).

Let $f\in\ap$ have exponent $\nu$, splitting field $K\subseteq\Fbar$ 
and kernel $V$.
We know any $p$-linear right factor $h=x^p-ax$
where $a\in K$ of $f$ has a one dimensional
kernel $W$ which is a subspace of $V$.  For any root $\alpha\neq 0$
of $h$, $a=\alpha^{p-1}$.  It follows that 
the only possible $p$-linear right composition
factors of $f$ have $(p-1)^{\rm st}$ powers of roots of $f/x$ in $K$
as the coefficients of their constant terms.  Therefore 
$(x^p-ax)\apdivides f$ if and
only if $a\in K$ is a root of $(f/x)\circ x^{1/(p-1)}$.
Assume $F_i$ is a field (to be defined later) for $1\leq i\leq\nu$ and
that $F=F_\nu\subseteq F_{\nu-1}\subseteq F_{\nu-2} \subseteq \cdots
\subseteq F_1\subsetneq \Fbar$.

\medskip
{\tt\obeylines
AbsAPDecomp: $\bbA_\Fbar\rightarrow cAPDEC^\Fbar_*$
~~~~Input~~:~-~$f^{(i)}\in\bbA_{F_i}$ monic of exponent $i$, 
~~~~~~~~~~~~~~~for some $i\in\nats$.
~~~~Output~:~-~a complete decomposition of $f$ in $cAPDEC^\Fbar_*$.

~~~~If $i=1$
~~~~~~~then return $f^{(1)}\in F_1[x]$.
~~~~Otherwise
~~~~~~~1) factor $h^{(i)}=(f^{(i)}/x)\circ x^{1\over{p-1}}\in F[x]$ 
~~~~~~~~~~such that $h^{(i)}=u_1^{e_1}u_2^{e_2}\cdots u_m^{e_m}$ 
~~~~~~~~~~where $u_j\in F[x]$ are distinct, monic and 
~~~~~~~~~~irreducible and $e_j\in\posnats$ for $1\leq j\leq m$.
~~~~~~~2) Let $a=z\bmod u_1\in F_{i-1}=F_i[z]/(u_1)$.
~~~~~~~3) Compute $g^{(i)}=f^{(i)}\apdiv (x^p-ax)\in F_{i-1}[x]$.
~~~~~~~4) Recursively compute an absolute decomposition 
~~~~~~~~~~$(g^{(i)},(v_m,v_{m-1},\ldots,v_2))\in cAPDEC^\Fbar_*$ using AbsAPDecomp.
~~~~~~~5) Return $(f^{(i)},(v_m,v_{m-1},\ldots,v_2,u_1))\in cAPDEC^\Fbar_*$.
}
\medskip

Each recursive stage $i$ (starting with stage $\nu$)
requires the factoring of a polynomial of degree at most 
$(p^i-1)/{(p-1)}\leq p^i$ in $F_i$.
The degree of $F_{i-1}$ over $F_i$ is at most 
$(p^i-1)/(p-1)\leq p^i$.  It follows that 
the degree of $F_i$ over $F=F_\nu$ is at most
\[
\eqalign{
\prod_{i< j\leq\nu} [F_{j-1}:F_j] &\leq \prod_{i<j\leq \nu} p^j\cr
&\leq p^{\nu(\nu-i)}.\cr
}
\]
Therefore, at recursive stage $i$, the number of field operations required is
at most 
\[
\sf(p^i)M(p^{\nu(\nu- i)})
\]
and the total cost is
\[
\sum_{0\leq i\leq \nu} \sf(p^i) M(p^{\nu(\nu-i)})
\leq M(p^{\nu^2})\sf(p^\nu).
\]

\noindent
We have shown the following:

\medskip
\begin{theorem}
Given $f\in\ap$ monic of degree $n=p^\nu$,
we can find an absolute decomposition of $f$ in $cAPDEC^\Fbar_*$ in
$O(M(p^{\nu^2})\sf(p^\nu))=n^{O(\log n)}$ 
field operations over $F$ provided $F$ supports a 
polynomial time factoring algorithm.
\end{theorem}

Suppose $F$ is finite.  It is conjectured that an additive polynomial
$f\in \ap$ of degree $n=p^\nu$ can have a splitting field $K$ of
degree at most $n^{O(1)}$ over $F$, and quite possibly at most $n$.
This would follow immediately from a (much stronger) unproven
conjecture of Ore[1933b] that the degrees of all irreducible factors
of $f$ divide the degree $t$ of the largest multiplicative factor of
$f$.  This would imply that $[K:F]=t$, and that the above algorithm
for absolute decomposition would run in a polynomial number (in $n$)
of field operations over $F$.


\newpage
\section{Rational Function Decomposition}

Let $f\in F(x)$ be a rational function in $x$.  A natural question to ask
is if $f$ can be represented as a composition of two other rational function
$g,h\in F(x)$, so that $f=g\circ h$.  This problem has polynomial
decomposition as a small subcase.
Mathematically, rational function decomposition has been examined
since Ritt[1923].  The Generalised Schur Problem for
rational functions involves the classification of so called 
``virtually one to one'' rational functions and their decompositions.  
In general, rational function decomposition is far from 
completely understood.  An in depth coverage
and survey of the problem is presented in Fried[1974],  
where a generalisation
of the tame case for polynomial decomposition in perfect fields is described
(and is well beyond the scope of this thesis). 
In this chapter we present a definition of
the rational function decomposition problem in a form similar to our 
presentation of the
polynomial decomposition problem.  We show that such decompositions
can be normalised in a manner similar to polynomial decomposition and
that the general problem is Cook reducible to the normal problem. We
then give a computational solution to the normal decomposition problem
for rational functions (which will require an exponential number of field
operations in the input degree and a factorisation algorithm over $F$).

\subsection{The Normalised Decomposition Problem}

If $f\in F(x)$ then $f=\fn/\fd$ for some $\fn,\fd\in F[x]$ of degrees 
$\nn$ and $\nd$ respectively.  We can assume that $\fn$ and $\fd$ 
are relatively prime 
and that $\fd$ is monic.  For any rational function $f$, there is
a unique pair of polynomials $\fn,\fd\in F[x]$ with
$\fd$ monic and $\gcd(\fn,\fd)=1$ such that $f=\fn/\fd$. 
With this in mind, define
\[
\bbU_F=\{ (f,(\fn,\fd))\in F(x)\times F[x]^2\,|
\,f=\fn/\fd,\;\gcd(\fn,\fd)=1,\fd~\hbox{monic}\}.
\]
If $(f,(\fn,\fd))\in\bbU_F$ and $\fn$ is monic, we say $f$ is monic.
Also define $\deg f=\nn+\nd$ and $\Delta(f)=\nn-\nd$.
The only automorphisms of the field $F(x)$ over $F$ are 
the fractional linear transformations
\[
x\mapsto {{t_1x+t_2}\over {t_3x+t_4}},
\]
where $t_1,t_2,t_3,t_4\in F$ and $t_1t_4-t_2t_3\neq 0$.  
The inverse of the the above map is
\[
x\mapsto \left({1\over {t_1t_4-t_2t_3}}\right){{t_4x-t_2}\over {-t_3x+t_1}}.
\]
Note that 
this group is isomorphic to $GL_2(F)$, the group of $2\times 2$ non-singular
matrices over $F$.  
Note also that if $f=\fn/\fd\in F(x)$, then
\[
{{t_1x+t_2}\over {t_3x+t_4}}\circ f={{t_1\fn+t_2\fd}\over {t_3\fn+t_4\fd}}.
\]

Let $f,g,h\in F(x)$ with $f=g\circ h$, and let $t\in F(x)$ be a fractional
linear transformation. We see that $f=(g\circ t^{-1})\circ (t\circ h)$ is
also a decomposition of $f$.  Two decomposition $f=g\circ h$
and $f=g'\circ h'$ are said to be
{\it linearly equivalent} if there exists a fractional
linear transformation $t\in F(x)$ such that $g=g'\circ t^{-1}$ and
$h=t\circ h'$.  Let $(f,(\fn,\fd))\in\bbU_F$ and $(\rn,\rd,\sn,\sd)\in\nats^4$.
Define 
\[
RATDEC^F_{(\rn,\rd,\sn,\sd)}=
\left\{
\eqalign{
(f,(g,&h))\in F(x)\times F(x)^2, \;f=g\circ h,\cr
& (g,(\gn,\gd)),(h,(\hn,\hd))\in\bbU_F,\cr
& \deg \gn=\rn,\; \deg \gd=\rd,\cr
& \deg \hn=\sn,\; \deg \hd=\sd.\cr
}
\right .
\]
For any $f\in F(x)$ and $(\rn,\rd,\sn,\sd)\in \nats^4$ there are 
potentially a large
(possibly infinite, depending upon $F$) number of decompositions of $f$
in $RATDEC^F_{(\rn,\rd,\sn,\sd)}$ (though up to linear equivalence we will
see there are at most a linearly exponential number). 
The {\it rational function decomposition
problem} is, given $f\in F(x)$ and $\rn,\rd,\sn,\sd\in\nats$,
to determine if there
exist any decompositions of $f$ in $RATDEC^F_{(\rn,\rd,\sn,\sd)}$, and,
if so, to find one or all of them up to linear equivalence.

Let $(f,(\fn,\fd)),(g,(\gn,\gd)),(h,(\hn,\hd))\in\bbU_F$. 
Assume  
$\fn,\fd,\gn,\gd,
\hn,\hd\in F[x]$ have degrees $\nn,\nd,\rn,\rd,\sn,\sd$
respectively, and that they are of the form
\[
\eqalign{
\fn= & \sum_{0\leq i\leq\nn} a_ix^i, & 
~~~~\fd= & \sum_{0\leq i\leq\nd}\abar_ix^i,\cr
\gn= & \sum_{0\leq i\leq\rn} b_ix^i, & 
~~~~\gd= & \sum_{0\leq i\leq\rd}\bbar_ix^i,\cr
\hn= & \sum_{0\leq i\leq\sn} c_ix^i, & 
~~~~\hd= & \sum_{0\leq i\leq\sd}\cbar_ix^i.\cr
}
\]
Let
\[
\eqalign{
A = & \sum_{0\leq i\leq \rn} b_i \hn^i\hd^{\rn-i} 
& = \hd^\rn(\gn\circ h) & \in F[x],\cr
B = & \sum_{0\leq j\leq \rd} \bbar_j \hn^j\hd^{\rd-j}
& = \hd^\rd(\gd\circ h) & \in F[x].\cr
}
\]
If $f=g\circ h$ then
\[
f={{A\hd^{-\rn}}\over {B\hd^{-\rd}}}=
\left\lbrace
\eqalign{
{A\over {B\hd^{\rn-\rd}}} & ~~\hbox{if}~\rn>\rd\cr
\noalign{\vskip 4pt}
{{A\hd^{\rd-\rn}}\over B} & ~~\hbox{if}~\rn\leq\rd\cr
}
\right .
\]
Note that 
\[
\deg A=
\left\{\eqalign{
    \rn\sn~~ & \hbox{if}~\sn>\sd,\cr
    \rn\sd~~ & \hbox{if}~\sn<\sd.\cr
    }\right .
\]
If $\sn=\sd$ then cancellation can occur and the strongest statement that
can be made is that $\deg A\leq\rn\sd$.
Similarly,
\[
\deg B= 
\left\{\eqalign{
    \rd\sn~~ & \hbox{if}~\sn>\sd,\cr
    \rd\sd~~ & \hbox{if}~\sn<\sd.\cr
    }\right .
\]
Once again, if $\sn=\sd$, cancellation can occur, and the most we can say
is that $\deg B\leq\rd\sd$.

\medskip
\begin{lemma}
$A$,$B$, and $\hd$ (as defined above) are pairwise relatively prime.
\end{lemma}
\begin{proof}
We first show $\gcd(A,B)=1$.  Suppose to the contrary that 
$\gcd(A,B)\neq 1$.  Then
$A$ and $B$ have a common root $\beta\in\Fbar$ (where $\Fbar$ is
an algebraic closure of $F$), and

\[
\eqalign{
A(\beta) = & 
      \mbox{$\displaystyle
        \cases{
        [\hd^\rn \gn(\hn/\hd)](\beta) & if $\hd(\beta)\neq 0$,\cr
         b_\rn\hn^\rn(\beta) & if $\hd(\beta)=0$,\cr
        }$}
\cr
   B(\beta) = & 
      \mbox{$\displaystyle
         \cases{
        [\hd^\rd \gd(\hn/\hd)](\beta) & if $\hd(\beta)\neq 0$,\cr
        \bbar_\rd\hn^\rd(\beta) & if $\hd(\beta)=0$.\cr
       }$}
}
\]
If $\hd(\beta)\neq 0$, it follows that $\gn(h(\beta))=\gd(h(\beta))=0$,
a contradiction since $\gn$ and $\gd$ are relatively prime.  If $\hd(\beta)=0$,
then $A(\beta)=b_\rn \hn^\rd(\beta)\neq 0$ (since $\gcd(\hn,\hd)=1$),
contrary to the assumptions.
Thus $\gcd(A,B)=1$.  We now show that $A$ and $\hd$ are
relatively prime.  First,
\[
\eqalign{
\gcd(A,\hd)
= & \gcd(\sum_{0\leq i\leq \rn} b_i\hn^i\hd^{\rn-i},\hd)\cr
= & \gcd(b_\rn\hn^\rn+\hd\sum_{0\leq i\leq\rn}b_i\hn^i\hd^{\rn-i-1},\hd)\cr
= & \gcd(b_\rn \hn^\rn,\hd)\cr
= & 1~~~~~\hbox{since}~\gcd(\hn,\hd)=1.\cr
}
\]
Similarly,
$$
\eqalign{
\gcd(B,\hd)
= & \gcd(\sum_{0\leq j\leq \rd} \bbar_j\hn^j\hd^{\rd-j},\hd)\cr
= & \gcd(\bbar_\rd\hn^\rd+\hd\sum_{0\leq j\leq\rd}\bbar_j\hn^j\hd^{\rd-j-1},
    \hd)\cr
= & \gcd(\bbar_\rd \hn^\rd,\hd)\cr
= & 1~~~~~\hbox{since}~\gcd(\hn,\hd)=1.\cr
}
\QED
$$
\end{proof}

\noindent
This implies that $f=(A/B)\hd^{\rd-\rn}$ is in ``lowest terms''.
For $f$ as above, we call $(\nn,\nd)$ the 
{\it degree pair} of $f$.  

\medskip
\begin{lemma}
Given $(f,(\fn,\fd)),(g,(\gn,\gd)),(h,(\hn,\hd))\in\uf$ with respective
degree pairs $(\nn,\nd),(\rn,\rd)$, and  $(\sn,\sd)$, where
$\Delta(f),\Delta(h)>0$ and $f=g\circ h$, it follows that $\Delta(g)>0$,
$\rn=\nn/\sn$ and
\[
\rd = {{\nd\sn-\nn\sd}\over {\sn(\sn-\sd)}}.
\]
\end{lemma}
\begin{proof} 
We know that $f=(A/B)\hd^{\rd-\rn}$, where $A,B\in F[x]$ are as defined
previously.  Assume $\rn\leq\rd$.  We have seen that 
$(f,(A\hd^{\rd-\rn},B))\in\uf$
and $\nn=\rn\sn+\rd\sd-\rn\sd>\rd\sn=\nd$.  A simple rearrangement reveals that
$\rn(\sn-\sd)>\rd(\sn-\sd)$ and since $\sn>\sd$, we find that $\rn>\rd$,
a contradiction.  It must then be true that $\rn>\rd$ and
$(f,(A,B\hd^{\rn-\rd}))\in\uf$.  From the previous discussion on the
degrees of $A$ and $B$, we know that 
$\nn=\rn\sn$ and $\nd=\rd\sn+\sd(\rn-\rd)$.  
Solving for $\rn$ and $\rd$ in these equations, we derive that
$\rn=\nn/\sn$ and
\[
\eqalign{
\rd(\sn-\sd) = & \nd -\sd\rn\cr
= & {{\nd\sn-\nn\sd}\over \sn}\cr
}
\]
and finally that
$$
\rd = {{\nd\sn-\nn\sd}\over {\sn(\sn-\sd)}}.
\QED
$$
\end{proof}
This implies that the degree pair of $g$ is totally determined by the
degree pairs of $f$ and $h$.  We will later see that in fact $f$ and $h$
uniquely determine $g$.

The set of degree pairs of $f$ and it images under
fractional linear transformations forms a highly structured set.

\medskip
\begin{lemma}
Let $f\in F(x)\setminus F$,
\[
T=\{t\circ f\setmid t\in F(x) ~\hbox{a fractional linear transformation}\}
\]
and 
\[
D=\{(c,d)\in\nats^2 \setmid (c,d) ~\hbox{a degree pair of some}~ g\in T\}.
\]
Then $D$ has exactly three elements,
and these are of the form $(a,b)$, $(b,a)$ and $(a,a)$, for some
$a,b\in \nats$ with $a>b$.
\end{lemma}
\begin{proof}
Assume $(f,(\fn,\fd))\in\uf$ and 
$f$ has degree pair $(\nn,\nd)$.  As noted earlier, 
if $t=(t_1x+t_2)/(t_3x+t_4)\in F(x)$ is a fractional linear transformation
then $t\circ f = (t_1\fn+t_2\fd)/(t_3\fn+t_4\fd)$.
We examine three cases.  If $\nn>\nd$, then for any fractional
linear transformation $t\in F(x)$, observation reveals
that $t\circ f$ has possible degree pairs $(\nn,\nd),\,(\nd,\nn)$ and
$(\nn,\nn)$.  Similarly, if $\nd>\nn$, $t\circ f$ has possible degree pairs
$(\nn,\nd),\,(\nd,\nn)$ and $(\nd,\nd)$.  If $\fn=\fd$, let $\an$
be the leading coefficient of $\fn$ and $\delta$ the degree of
$\fn-\an\fd$.  Then $t\circ f$ can have degree pairs $(\nn,\nn)$,
$(\delta,\nn)$, and $(\nn,\delta)$.  Since the fractional linear 
transformations form a group under composition, these are the only
degree pairs.
\QED
\end{proof}

This allows us to normalise the rational function decomposition problem
and show a reduction from the general problem to the normal problem.
For any $(f,(\fn,\fd))\in\bbU_F$,
let $\an\in F$ be the leading coefficient of $\fn$ and let $\gamma\in F$
be the leading coefficient of $\fn-\an\fd$.  Also, let $\alpha=\an-1\in F$.
Define $\Lambda_f$, a fractional linear transformation, as follows:
\[
\Lambda_f=
\left \{
\eqalign{
x/\an~~~~~~~~~~~~~~~ & \hbox{if}~ \Delta(f)>0,\cr
\gamma\cdot\left({{x-\alpha}\over{x-\an}}\right)~~~~~ &
\hbox{if}~ \Delta(f)=0,\cr
\an/x~~~~~~~~~~~~~~~ & \hbox{if}~ \Delta(f)<0.\cr
}
\right .
\]
$\Lambda_f$ is a fractional linear transformation.  Observation reveals
that $\Lambda_f\circ f$ is monic and $\Delta(\Lambda_f\circ f)>0$.
If $f=g\circ h$ for $g,h\in F(x)$, then
\[
\Lambda_f\circ f=(\Lambda_f\circ g\circ\Lambda_h^{-1})\circ (\Lambda_h\circ h).
\]
Therefore, we can assume for any 
decomposition of $f$ that $f$ and $h$ are monic
and $\Delta(f)$ and $\Delta(h)$ are both positive.
By lemma 6.2 we know $\Delta(g)$ is positive as well.
Because 
$\fn=\sum_{0\leq i\leq\rn} b_i\hn^i\hd^{\rn-i}$, $\sn>\sd$, and $\fn$ is monic,
we also see that $b_\rn=1$ and $g$ is monic as well.
A further normalisation can be made by noting that 
\[
\eqalign{
f = & g\circ h \cr
= & g(x+h(0))\circ (h-h(0)).\cr
}
\]
Assume $\Delta(h)$ is positive.
Then $h-h(0)=(\hn-h(0)\hd)/\hd$ has the same degree pair as $h$.  
If $\Delta(f)$ and $\Delta(g)$ are also positive, it follows
that $g(x+h(0))$ has the same degree pair as $g$ as well
since $\rn$ and $\rd$ are completely determined by $\nn,\nd,\sn$, and $\sd$, 
so we can assume
$h(0)=0$ in any decomposition. 
We call a decomposition of a monic rational function $f$ with
$\Delta(f)>0$ into two monic rational function $g,h\in F(x)$ such that
$\Delta(g)>0$, $\Delta(h)>0$ and $h(0)=0$ a {\it normal decomposition} of $f$.
Let $\nn,\nd,\rn,\rd,\sn,\sd\in\nats$ be such that $\nn=\rn\sn$
and $\nd=\rn\sd-\rd\sd+\rd\sn$.
Define 
\[
NRATDEC^F_{(\rn,\rd,\sn,\sd)}=\left\{
\eqalign{
&(f,(g,h))\in F(x)\times F(x)^2:\cr
&~~~~~f,g,h~\hbox{monic},\; h(0)=0,\;f=g\circ h,\cr
&~~~~~(f,(\fn,\fd)),(g,(\gn,\gd)),(h,(\hn,\hd))\in\bbU_F,\cr
&~~~~~\Delta(f),\Delta(g),\Delta(h)>0,\cr
&~~~~~\deg \gn=\rn,\,\deg \gd=\rd,\cr
&~~~~~\deg\hn=\sn,\,\deg\hd=\sd.\cr
}
\right\}
\]
Given a monic $f\in F(x)$ with $\Delta(f)>0$, 
and $(\rn,\rd,\sn,\sd)\in\nats^4$
such that $\nn=\rn\sn$ and $\nd=\rn\sd-\rd\sd+\rd\sn$,
the {\it normalised rational function decomposition problem} 
is to determine if there
exists
$(f,(g,h))\in NRATDEC^F_{(\rn,\rd,\sn,\sd)}$ and if so to find some
predetermined number of them.

Note that unlike the polynomial decomposition problem, the degrees of the
numerators and denominators of $f,g$ and $h$ are not left constant by
the normalisations.  We will now examine the relationship between the
normal problem and the general problem by showing a linear time 
(in the input degree) reduction
from the general problem to the normal problem.

Assume $f\in F(x)$ and $\rn,\rd,\sn,\sd\in\nats$ are given as in the
general problem.  Also assume $(f,(\fn,\fd))\in\uf$ 
and $\fbar = \Lambda_f\circ f$ has $(\fbar,(\fnbar,\fdbar))\in\uf$ and
degree pair $(\nnbar,\ndbar)$.  The easiest case occurs when $\sn>\sd$.
For each decomposition $(f,(g,h))\in RATDEC^F_{(\rn,\rd,\sn,\sd)}$ there
is a decomposition 
$(\fbar,(\gbar,\hbar))\in NRATDEC^F_{(\rnbar,\rdbar,\sn,\sd)}$,
where  $\rnbar$ and $\rdbar$ are determined as in lemma 6.2. 
Conversely, 
from any $(\fbar,(\gbar,\hbar))\in NRATDEC^F_{(\rnbar,\rdbar,\sn,\sd)}$ 
we can find a decomposition
$(f,(\Lambda_f^{-1}\circ\gbar,\hbar))\in RATDEC^F_{(\rn,\rd,\sn,\sd)}$.

If $\sn<\sd$ we have another easy case since $\Lambda_h\circ h$ has degree
pair $(\sd,\sn)$.  
For each decomposition $(f,(g,h))\in RATDEC^F_{(\rn,\rd,\sn,\sd)}$ there
is a decomposition
\[
(\fbar,(\gbar,\hbar))\in NRATDEC^F_{(\rnbar,\rdbar,\sd,\sn)}
\]
where once again, we compute $\rnbar,\rdbar$ as in lemma 6.2.
From any $(\fbar,(\gbar,\hbar))\in NRATDEC^F_{(\rnbar,\rdbar,\sd,\sn)}$ 
we can find a decomposition
$(f,(\Lambda_f^{-1}\circ\gbar\circ (1/x),(1/x)\circ\hbar))
\in RATDEC^F_{(\rn,\rd,\sn,\sd)}$.

Finally, if $\sn=\sd$, we have a difficulty in that 
the problem can be normalised in a number of different ways.
Let $\sdbar\in\nats$ with
$\sdbar<\sn$,  
and find $\rnbar,\rdbar\in\nats$ 
as in lemma 6.2
(if such an integer solution exists).
For each decomposition $(f,(g,h))\in RATDEC^F_{(\rn,\rd,\sn,\sd)}$ there
is a decomposition 
$(\Lambda_f\circ f,(\Lambda_f\circ g,h))\in NRATDEC^F_{(\rn',\rd',\sn,\sd)}$ 
for some appropriately calculated $\rn',\rd'\in\nats$
(as in lemma 6.2).
By lemma 6.3 there exists a fractional linear transformation $t\in F(x)$
such that $\Delta(t\circ h)>0$ and $t\circ h$ is monic and has degree
pair $(\sn,\sdbar)$ for some $\sdbar<\sn$.  Thus
$(\Lambda_f\circ f,(\Lambda_f\circ g\circ t^{-1},t\circ h))
\in NRATDEC^F_{(\rnbar,\rdbar,\sn,\sdbar)}$ for appropriately determined
$\rnbar,\rdbar\in\nats$.
For any $(\fbar,(\gbar,\hbar))\in NRATDEC^F_{(\rnbar,\rdbar,\sn,\sdbar)}$ 
we can find a decomposition
$(f,(\Lambda^{-1}_f\circ\gbar\circ [1/(x-1)],[(x+1)/x]\circ\hbar))
\in RATDEC^F_{(\rn,\rd,\sn,\sd)}$ (the fractional linear transformation
$t=(x+1)/x$ [whose inverse is $1/(x-1)$] is such that $t\circ\hbar$ has
degree pair $\sn,\sn$).
This requires the solution of at most
$\sn<\deg f$ normal problems.  We have shown the following:

\medskip
\begin{theorem}
Assume we have an algorithm such that, given a monic $\fbar\in F(x)$ of degree
$n$ with $\Delta(\fbar)>0$, and $(\rnbar,\rdbar,\snbar,\sdbar)\in\nats^4$,
we can determine if there exist any 
$(\fbar,(\gbar,\hbar))\in NRATDEC^F_{(\rn,\rd,\sn,\sd)}$, and if so, find
some predetermined number of them, in $O(T(n))$ field operations.
Then, given $f\in F(x)$ of degree $n$ 
and $(\rn,\rd,\sn,\sd)\in\nats^4$, we can determine
if there exist any $(f,(g,h))\in RATDEC^F_{(\rn,\rd,\sn,\sd)}$, and if so,
find some  predetermined number of them, in $O(\sn T(n))=O(nT(n))$
field operations.  
\end{theorem}

This is equivalent to saying that the general
rational function decomposition problem is Cook reducible to the
normal rational function decomposition problem, where the oracle for
the normal problem is consulted $\sn$ times.

\subsection{Decomposing Normalised Rational Functions}

In this section we present a general computational solution for the rational 
function decomposition problem.  
Throughout this section, for any $f\in F[x]$ of degree $n$ and any
$i\in\nats$ such that $0\leq i\leq n$, we let $\coeff(f,i)\in F$ 
be the coefficient of $x^i$ in $f$.
We begin by showing a preliminary lemma.

\medskip
\begin{lemma}
Given  $r\in\nats$, $u\in F[x]$ monic of degree $n$, 
and $h\in F(x)$ monic with $h(0)=0$, 
$(h,(\hn,\hd))\in\bbU_F$ and
$\Delta(h)>0$, we can determine if there exists a 
monic $v\in F[x]$
of degree $r$ such that $u=v(h)\hd^r$ in $O(n\log nM(n))$ field operations.
\end{lemma}
\begin{proof}
Assume $\hn,\hd$ have degree $\sn,\sd$ respectively.  
It follows that if $v$ exists, and is of the form
\[
v=\sum_{0\leq i\leq r} b_i x^i
\]
with $b_i\in F$ for $0\leq i\leq r$ then
\[
u=\sum_{0\leq i\leq r} b_i\hn^i\hd^{r-i}.
\]
Let $d=\max\{j:\,x^j\setmid\hn\}\geq 1$. We see that for $\ell\in\nats$,
\[
\eqalign{
\coeff(u,\ell d)
= & \coeff(\sum_{0\leq i\leq r} b_i\hn^i\hd^{r-i},\ell d)\cr
= & \coeff(\sum_{0\leq i\leq \ell} b_i\hn^i\hd^{r-i},\ell d)\cr
= & b_\ell\coeff(\hn,d)^\ell\coeff(\hd,0)^{r-\ell}
    +\coeff(\sum_{0\leq i< \ell} b_i\hn^i\hd^{r-i},\ell d).\cr
}
\]
Since $\coeff(\hn,d)\neq 0$ and $\coeff(\hd,0)\neq 0$, we know
\[
b_\ell={{\coeff(u,\ell d)-
       \coeff(\sum_{0\leq i< \ell} b_i\hn^i\hd^{r-i},\ell d)}
\over
        {\coeff(\hn,d)^\ell\coeff(\hd,0)^{r-\ell}}}.
\]
Using this recurrence we can compute the coefficients 
$b_0,b_1,\ldots,b_r\in F$ in order.
Because the system is over constrained, the computed coefficients may
not lead to a decomposition.  Thus we must check if 
in fact $u=v(g)\hd^r$.
The cost of this computation is dominated by the cost of computing 
$\hn^i\hd^{r-i}$ for $0\leq i\leq r$, which can be done with $O(n\log nM(n))$
field operations over $F$.
\QED
\end{proof}

\noindent
The previous lemma allows us to perform ``right division'' in
the ring of normal rational functions under composition.

\medskip
\begin{lemma}
Given $f,h\in F(x)$ monic with $h(0)=0$ and $\Delta(f),\Delta(h)>0$,
we can determine if there exists a monic $g\in F(x)$ (with $\Delta(g)>0$)
such that $f=g\circ h$, and if so compute it, with $O(n\log nM(n))$
field operations.
\end{lemma}
\begin{proof}
Assume $(f,(\fn,\fd)), (h,(\hn,\hd))\in\bbU_F$.
We want to find $(g,(\gn,\gn))\in\uf$ such that $f=g\circ h$.
We know $\fn=(\gn\circ h)\hd^{\rn}$ and using lemma 6.5 we can compute
$\gn\in F[x]$ if it exists.  We also know that 
$\fd/\hd^{\rn-\rd}=\gd(h)\hd^\rd$, and so we
can compute $\gd\in F[x]$ if it exists.
The total number of field operations required is $O(n\log nM(n))$.
\QED
\end{proof}

We can now give an algorithm for the normal rational function decomposition
problem.  It will require an exponential number of field operations.

\pagebreak
{\tt\obeylines
NormRatDec: $F(x)\times\nats^4\rightarrow \pwr(RATDEC^F)$

~~~~Input:~~-~$f\in F(x)$ with $\Delta(f)>0$ and
~~~~~~~~~~~~~~$(f,(\fn,\fd))\in\bbU_F$ where $\fn$ and $\fd$ have
~~~~~~~~~~~~~~degrees $\nn$ and $\nd$ respectively,
~~~~~~~~~~~~-~$\rn,\rd,\sn,\sd\in\nats$ such that $\nn=\rn\sn$
~~~~~~~~~~~~~~and $\nd=\rn\sd-\rd\sd+\rd\sn$.
~~~~Output:~-~the set of all decompositions of $f$ 
~~~~~~~~~~~~~~in $RATDEC^F_{(\rn,\rd,\sn,\sd)}$.

~~~~1) Let $T:=\emptyset$.
~~~~2) For each monic $\hd\in F[x]$ of degree $\sd$ such that
~~~~~~~$\hd^{\rn-\rd}|\fd$, and $\hd(0)\neq 0$, do
~~~~~~~~~~2.1) Let $B:=\fd/\hd^{\rn-\rd}$.
~~~~~~~~~~2.2) Let $\bbar_0:=\fd(0)/(\hd(0))^\rn$.
~~~~~~~~~~2.3) For each factor $\hn$ of $B-\bbar_0\hd^\rd$ of degree $\sn$,
~~~~~~~~~~~~~~~~~2.3.1) Let $h:=\hn/\hd$.
~~~~~~~~~~~~~~~~~2.3.2) Attempt to compute $g\in F(x)$ such that
~~~~~~~~~~~~~~~~~~~~~~~~$f=g\circ h$ using lemma 6.6.  If 
~~~~~~~~~~~~~~~~~~~~~~~~such a $g$ exists for the chosen $h$, 
~~~~~~~~~~~~~~~~~~~~~~~~add $(f,(g,h))$ to $T$.
~~~~3) Return $T$.
}

\medskip
We know that in any decomposition $\hd^{\rn-\rd}|\fd$, so
in step 2 we generate all potential candidates for $\hd$.
In step 2.2, since $\fd(0)=\hd^{\rn-\rd}(0)\bbar_0\hd^\rd(0)=\bbar_0\hd^\rn$,
we can compute $\bbar_0=\fd(0)/\hd^\rn(0)$.  We use
the identity
\[
\eqalign{
B-\bbar_0\hd^\rd = & 
\sum_{1\leq j\leq\rd}\bbar_j\hn^j\hd^{\rd-j}\cr
= & \hn\sum_{1\leq j\leq\rd}\bbar_j\hn^{j-1}\hd^{\rd-j}\cr
}
\]
to get all candidates for $\hn$, namely all degree $\rn$ factors of
$B-\bbar_0\hd^\rd$.  In step 2.3.1 we simply check whether the chosen
$h=\hn/\hd$ leads to a decomposition.
The algorithm certainly requires an exponential number of field operations
in the input size because for any $f\in F[x]$ of degree $n$, there are
potentially $2^n$ factors of $f$.  
Therefore, the cost of the algorithm is dominated by the cost of
computing step 2.3.2 as many as $(2^n)^2$ times, each time requiring
$O(n\log nM(n))$ field operations.  We have shown the following theorem.

\medskip
\begin{theorem}
The normal rational function decomposition problem can be solved with
$O(2^{2n}n\log nM(n))$ field operations.
\end{theorem}

\noindent
Using  theorem 6.4 we get the following corollary for the general
case:
\medskip
\begin{corollary}
The rational function decomposition problem can be solved with
$O(2^{2n}n^2\log nM(n))$ field operations.
\end{corollary}


\section{Conclusion.}

We formally presented the decomposition problem for polynomials
(both univariate and multivariate)
in a number of formulations and showed their equivalence.  We then
presented a survey of the known algorithms for the decomposition problem
in light of this consistent mathematical basis.  A reduction
is shown from the general (multiple composition factor) decomposition
problem to the bidecomposition problem for ``nice'' classes of polynomials. 
In the wild case we exhibited super-polynomial lower bound on the number
of decompositions of a polynomial which can exist by examining 
the additive polynomials, for which all decompositions are wild.
We dealt with the additive case algorithmically as well, demonstrating
a polynomial time algorithm for generating a complete decomposition
(and hence determining indecomposability).  It is shown that the
decomposition problem for additive polynomials can be solved in
quasi-polynomial time.  We also showed that the general decomposition problem
for completely reducible additive polynomials and 
the bidecomposition problem for similarity free
additive polynomials can be solved in polynomial time.  
The rational function decomposition problem is also defined and
it is shown how to normalise this problem appropriately, such
that the general problem is reducible to the normal one.  We then showed
how to solve the normalised rational function decomposition problem
in a polynomial number of field operations.

Many open questions remain in the wild case for polynomial decomposition.
The additive polynomials represent a small but important subcase
of these polynomials and yet even here no polynomial time algorithm
is known for even the bidecomposition problem.  It is strongly suspected
by the author that such an algorithm exists.
Interesting questions also remain
concerning the computation of absolute decompositions.  It may be
true that even over ``well-behaved'' fields such as finite fields
that the coefficients of an absolute decomposition
generate an extension of exponential degree over the ground field.
And of course, the main open question is still the existence
of a polynomial time algorithm for the rational polynomial decomposition
problem in the wild case.
The rational function decomposition problem is only dealt with briefly here and
many interesting questions remain unsolved.  Most of these
problems are extremely difficult, and the mathematical theory
is very incomplete.  Polynomial time algorithms, even for
special cases, would be of great interest.

\section*{References}

V.S. Alagar and M. Thanh, ``Fast Polynomial Decomposition Algorithms,''
{\sl Proceedings of EUROCAL 1985, Lecture Notes in Computer Science
{\bf 204}}, Springer Verlag, Heidelberg, 1985, pp. 150--153.

D.R. Barton and R. Zippel, ``Polynomial Decomposition Algorithms,''
{\sl Journal of Symbolic Computation, Vol. 1}, 1985, pp. 159--168.

R.P. Brent and H.T. Kung, ``Fast Algorithms for Composition and
Reversion of Multivariate
Power Series,'' {\sl Proceedings of the Conference on Theoretical
Computer Science}, University of Waterloo, Waterloo, Ontario,
Canada, 1977, pp. 149--158.

R.P. Brent and H.T. Kung, ``Fast Algorithms for Manipulating Formal
Power Series,'' {\sl Journal of the ACM, Vol. 25, No. 4}, 1978,
pp. 581--595.

D.G. Cantor and E. Kaltofen,``Fast Multiplication of Polynomials over
Arbitrary Rings'', Preliminary Report, 1987.

T.H.M. Crampton and G. Whaples, ``Additive Polynomials II,''
{\sl Transactions of the AMS, Vol. 78}, 1955, pp. 239--252.

M. Dickerson, ``Polynomial Decomposition Algorithms for Multivariate
Polynomials'', {\sl Department of Computer Science, Cornell University,
Technical Report 87-826}, 1987.

F. Dorey and G. Whaples, ``Prime and Composite Polynomials,''
{Journal of Algebra, Vol. 28}, 1974, pp. 88--101.

H.T. Engstrom, ``Polynomial Substitutions,''
{American Journal of Mathematics, Vol. 63}, 1941, pp. 249--255.

A. Evyatar and D.B. Scott, ``On Polynomials in a Polynomial,''
{Bulletin of the London Mathematical Society, Vol. 4}, 1972, pp. 176--178.

M.D. Fried and R.E. MacRae [1969a], ``On the Invariance of Chains of fields,''
{\sl Illinois Journal of Mathematics, Vol. 13}, 1969, pp. 165--171.

M.D. Fried and R.E. MacRae [1969b], ``On Curves with Separated Variables,''
{Annals of Mathematics, Vol. 180}, 1969, pp. 220--226.

M.D. Fried, ``Arithmetical Properties of Function Fields (II). The
Generalized Schur Problem,''
{\sl Acta Arithmetica, Vol. 25}, 1974, pp. 225--258.

J. von zur Gathen[1986], ``Parallel Arithmetic Computations: a survey.''
Proceedings of the $12^{\rm th}$ International Symposium on the Mathematical
Foundations of Computer Science, Bratislava, Springer Lecture Notes in
Computer Science 233, 1986, 93-112.

J. von zur Gathen[1987], ``Functional Decomposition of Polynomials: the Tame 
Case,'' 
Technical Report, Sondersforschungsbereich 124, Universit\"at
des Saarlandes, Saarbr\"ucken, August 1987.

J. von zur Gathen[1988], ``Functional Decomposition of Polynomials: the wild
Case,''
Manuscript in preparation, June 1988.

J. von zur Gathen, D. Kozen, S. Landau, ``Functional Decompositions of
Polynomials,'' Procedings of the $28^{th}$ annual IEEE Symposium 
on the Foundations of
Computer Science, Los Angeles CA, 1987, pp. 127-131.

G.H. Hardy and E.M. Wright, {\it An Introduction to the Theory of Numbers.}
Clarendon Press, Oxford, 1962.

D. Kozen and S. Landau, ``Polynomial Decomposition Algorithms,''
{\sl Department of Computer Science, Cornell University,
Technical Report 86-773}, 1986. (to appear in the Journal of Symbolic
Computation).

S. Landau and G.L. Miller, ``Solvability by Radicals is in Polynomial Time,''
{\sl Journal of Computer Systems Science, Vol. 30}, 1985, pp. 179--208.

H. Levi, ``Composite Polynomials with Coefficients in an Arbitrary
Field of Characteristic Zero,'' {\sl American Journal of Mathematics,
Vol. 64}, 1942, pp. 389--400.

R. Lidl and H. Niederreiter, ``Introduction to Finite Fields and
Their Applications,'' Cambridge University Press, Cambridge, 1986.

J.D. Lipson, ``Newton's Method: A Great Algebraic Algorithm,''
{\sl Proceedings of ACM Symposium on Symbolic and Algebraic Computation},
1976, pp. 260--270.

O. Ore [1933a], ``Theory of Non-Commutative Polynomials,''
{\sl Annals of Mathematics, Vol. 34, No. 2}, 1933, pp. 480--508.

O. Ore [1933b], ``On a Special Class of Polynomials,''
{\sl Transactions of the AMS, Vol. 35}, 1933, pp. 559--584.

O. Ore, ``Contributions to the Theory of Finite Fields,''
{\sl Transactions of the AMS, Vol. 36}, 1934, pp. 243--274.

J.F. Ritt[1922], ``Prime and Composite Polynomials,''
{\sl Transactions of the AMS, Vol. 23}, 1922, pp. 51--66.

J.F. Ritt[1923], ``Permutable Rational Functions,''
{\sl Transactions of the AMS, Vol. 25}, 1923, pp. 399-448.

J. B. Rosser and Lowell Schoenfeld, ``Approximate Formulas For Some
Functions Of Prime Numbers'', {\sl Illinois Journal of Mathematics,
Vol. 6}, 1962, pp. 64-94.

Sch${\ddot {\rm o}}$nhage, Schnelle Multipikation von Polynomen
\"uber K${\ddot {\rm o}}$rpern der Charakteristik 2.
Acta Informatica 7 (1977), 395-398.

J.T. Schwartz, ``Fast Probilistic Algorithms for the Verification
of Polynomial Identities'', {\sl Journal of the ACM, Vol. 27, No. 4},
October 1980, pp. 701-717.

B.L. van der Waerden, ``Modern Algebra, Vol. 1,'' Fredrick Unger Publishing
Co., New York, 1949.

G. Whaples, ``Additive Polynomials,''
{\sl Duke Mathematical Journal, Vol. 21}, 1954, pp. 55--63.

\end{document}